\documentclass[12pt,a4paper]{amsart}

\usepackage{amsmath,amsthm,amsfonts,amssymb,empheq}
\usepackage{srcltx}    
\usepackage{mathrsfs}
\usepackage[dvipdfmx]{graphicx}
\usepackage[all]{xy}   
\usepackage{color}
\usepackage{longtable}
\usepackage{hhline}
\usepackage{url}

\usepackage{here}  
\usepackage{caption}

\makeatletter
\def\maketag@@@#1{\hbox{\m@th\normalfont\normalsize#1}}
\makeatother

\newcommand{\e}{\mathrm{e}}

\def\be{\mathbf{e}}

\def\IN{\mathbb {N}}
\def\IZ{\mathbb {Z}}

\def\IC{\mathbb {C}}
\def\IP{\mathbb {P}}

\def\ba{\boldsymbol{a}}
\def\bb{\boldsymbol{b}}

\def\br{\boldsymbol{r}}
\def\bk{\boldsymbol{k}}
\def\bs{\boldsymbol{s}}
\def\bm{\boldsymbol{m}}
\def\bx{\boldsymbol{x}}

\def\bsig{\boldsymbol{\sigma}}
\def\bmu{\boldsymbol{\mu}}

\def\bdel{\boldsymbol{\delta}}
\def\bY{\boldsymbol{Y}}
\def\bW{\boldsymbol{W}}
\def\bN{\boldsymbol{N}}
\def\bB{\boldsymbol{B}}
\def\ffc{\mathfrak{c}}

\def\fq{\mathfrak{q}}
\def\ft{\mathfrak{t}}

\def\bft{\boldsymbol{\mathfrak{t}}}
\def\bfx{\boldsymbol{\mathsf{x}}}
\def\bfc{\boldsymbol{\mathfrak{c}}}

\def\ll{\left\lgroup}
\def\rr{\right\rgroup}

\def\inc#1{%
\setbox1\hbox{\includegraphics[scale=.375]{#1}}%
\vcenter{\box1}%
}

\newcommand{\wsq}{\square}

\DeclareMathOperator{\partit}{par}

\def\bell{\boldsymbol{\ell}}
\def\bX{\boldsymbol{X}}
\def\bfone{{\boldsymbol{1}}}

\newcommand{\cartan}{\mathfrak{h}}
\newcommand{\ocartan}{\overline{\cartan}}
\newcommand{\ofwt}{\overline{\Lambda}}
\newcommand{\orho}{\overline{\rho}}
\newcommand{\oA}{\overline{A}}
\newcommand{\dimfin}{{n}}    
\newcommand{\dimaff}{{n}}    
\newcommand{\levelaff}{{N}}    
\newcommand{\indfin}{\overline{\mathcal I}_\dimfin}
\newcommand{\indaff}{{\mathcal I}_\dimaff}
\newcommand{\indaffm}{\overline{\mathcal I}_\dimaff}

\newcommand{\inner}[2]{\left\langle{#1},{#2}\right\rangle}
\newcommand{\slfin}[1]{\mathfrak{sl}(#1)}
\newcommand{\slchap}[1]{\widehat{\mathfrak{sl}}(#1)}

\newcommand{\fwt}{\Lambda}
\newcommand{\pochinf}[2]{(#1;#2)_{\infty}}
\newcommand{\vpochinf}[2]{\left(#1;#2\right)_{\infty}}

\newcommand{\burgeset}[1]{\mathcal{C}^{#1}} 
\newcommand{\burgecolset}[2]{\mathcal{C}^{#1}_{#2}} 
\newcommand{\burgecolfun}[2]{{X}^{#1}_{#2}} 
\newcommand{\rburgecolset}[3]{\mathcal{C}^{#1}_{#2;#3}} 
\newcommand{\rburgecolfun}[3]{{X}^{#1}_{#2;#3}} 
\newcommand{\colfun}[1]{{X}_{#1}}
\newcommand{\rcolfun}[2]{{X}_{#1;#2}}

\newcommand{\multicol}[1]{\mathcal{M}^{#1}}
\newcommand{\multicolref}[2]{\mathcal{M}^{#1,#2}}
\newcommand{\multicolHL}[1]{\mathcal{M}^{#1}_\ast}

\newcommand{\partitset}{\textsl{Par}\,}

\newcommand{\xchi}{\bar\chi}
\newcommand{\xa}{\bar a}
\newcommand{\qchi}{\chi}
\newcommand{\qvar}{\fq}
\newcommand{\mydelta}{\eta}

\theoremstyle{plain}
  \newtheorem{prob}{Problem}[section]
  \newtheorem{conj}[prob]{Conjecture}

  \newtheorem{prop}[prob]{Proposition}
   \newtheorem{cor}[prob]{Corollary}
    \newtheorem{lemm}[prob]{Lemma}
    \newtheorem{defi}[prob]{Definition}
\theoremstyle{remark}
  \newtheorem{remark}[prob]{\bf Remark}

\newtheorem{exam}[prob]{\bf Example}

\usepackage[scale=0.84]{geometry}

\newdimen\tableauside\tableauside=1.0ex
\newdimen\tableaurule\tableaurule=0.4pt
\newdimen\tableaustep
\def\phantomhrule#1{\hbox{\vbox to0pt{\hrule height\tableaurule width#1\vss}}}
\def\phantomvrule#1{\vbox{\hbox to0pt{\vrule width\tableaurule height#1\hss}}}
\def\sqr{\vbox{%
  \phantomhrule\tableaustep
  \hbox{\phantomvrule\tableaustep\kern\tableaustep\phantomvrule\tableaustep}%
  \hbox{\vbox{\phantomhrule\tableauside}\kern-\tableaurule}}}
\def\squares#1{\hbox{\count0=#1\noindent\loop\sqr
  \advance\count0 by-1 \ifnum\count0>0\repeat}}
\def\tableau#1{\vcenter{\offinterlineskip
  \tableaustep=\tableauside\advance\tableaustep by-\tableaurule
  \kern\normallineskip\hbox
    {\kern\normallineskip\vbox
      {\gettableau#1 0 }%
     \kern\normallineskip\kern\tableaurule}%
  \kern\normallineskip\kern\tableaurule}}
\def\gettableau#1 {\ifnum#1=0\let\next=\null\else
  \squares{#1}\let\next=\gettableau\fi\next}

\makeatletter
\def\Left#1#2\Right{\begingroup%
   \def\ts@r{\nulldelimiterspace=0pt \mathsurround=0pt}%
   \let\@hat=#1%
   \def\sht@im{#2}%
   \def\@t{{\mathchoice{\def\@fen{\displaystyle}\k@fel}%
          {\def\@fen{\textstyle}\k@fel}%
          {\def\@fen{\scriptstyle}\k@fel}%
          {\def\@fen{\scriptscriptstyle}\k@fel}}}%
   \def\g@rin{\ts@r\left\@hat\vphantom{\sht@im}\right.}%
   \def\k@fel{\setbox0=\hbox{$\@fen\g@rin$}\hbox{%
      $\@fen \kern.3875\wd0 \copy0 \kern-.3875\wd0%
      \llap{\copy0}\kern.3875\wd0$}}%
      \def\pt@h{\mathopen\@t}\pt@h\sht@im%
      \Right}%
\def\Right#1{\let\@hat=#1%
   \def\st@m{\mathclose\@t}%
   \st@m\endgroup}
\makeatother

\makeatletter
 \renewcommand{\theequation}{%
       \thesection.\arabic{equation}}
 \@addtoreset{equation}{section}
\makeatother

\makeatletter
\def\eqnarray{%
 \stepcounter{equation}%
 \let\@currentlabel=\theequation
 \global\@eqnswtrue
 \global\@eqcnt\z@
 \tabskip\@centering
 \let\\=\@eqncr
 $$\halign to \displaywidth\bgroup\@eqnsel\hskip\@centering
 $\displaystyle\tabskip\z@{##}$&\global\@eqcnt\@ne
 \hfil$\displaystyle{{}##{}}$\hfil
 &\global\@eqcnt\tw@$\displaystyle\tabskip\z@{##}$\hfil
 \tabskip\@centering&\llap{##}\tabskip\z@\cr}
\makeatother

\begin{document}

\title[]{$\widehat{\mathfrak{sl}}(n)_N$ WZW conformal blocks from 
\\
$SU(N)$ instanton partition functions on ${\mathbb {C}}^2/{\mathbb {Z}}_n$
}

\author[]{Omar Foda, Nicholas Macleod, Masahide Manabe, 
and Trevor Welsh} 

\address{
School of Mathematics and Statistics, 
University of Melbourne, 
Royal Parade, Parkville, Victoria 3010, Australia
} 

\dedicatory{
To Prof Rodney J Baxter on the occasion of his 80th birthday
}

\email{omar.foda@unimelb.edu.au,
n.macleod@student.unimelb.edu.au, 
masahidemanabe@gmail.com,
twelsh1@unimelb.edu.au} 

\begin{abstract}
Generalizations of the AGT correspondence between 
4D $\mathcal{N}=2$ $SU(2)$ supersymmetric gauge 
theory on ${\mathbb {C}}^2$ with $\Omega$-deformation and 
2D Liouville conformal field theory 
include a correspondence between 
4D $\mathcal{N}=2$ $SU(N)$ supersymmetric gauge theories,
$N = 2, 3, \ldots$, 
on ${\mathbb {C}}^2/{\mathbb {Z}}_n$, 
$n = 2, 3, \ldots$, with $\Omega$-deformation and 
2D conformal field theories with
$\mathcal{W}^{\, para}_{N, n}$
($n$-th parafermion $\mathcal{W}_N$) 
symmetry and 
$\widehat{\mathfrak{sl}}(n)_N$ symmetry. 
In this work, we 
trivialize the factor with $\mathcal{W}^{\, para}_{N, n}$ 
symmetry in the 4D $SU(N)$ instanton partition functions 
on ${\mathbb {C}}^2/{\mathbb {Z}}_n$
(by using specific choices of parameters
and imposing specific conditions on 
the $N$-tuples of Young diagrams that label the states),
and extract the 2D $\widehat{\mathfrak{sl}}(n)_N$ WZW 
conformal blocks, $n = 2, 3, \ldots$, $N = 1, 2, \ldots\, .$
\end{abstract}

\maketitle

\section{Introduction}

\subsection{Algebras on the equivariant cohomology of instanton 
moduli spaces}

In \cite{Alday:2009aq}, Alday, Gaiotto and Tachikawa conjectured 
a profound correspondence between $SU(2)$ instanton 
partition functions in $\mathcal{N}=2$ supersymmetric gauge theories 
on ${\IC}^2$, with $\Omega$-deformation \cite{Nekrasov:2002qd}, and 
Virasoro conformal blocks on the sphere and on the torus 
(see \cite{Alba:2010qc} for a proof
\footnote{\,
In the context of geometric representation theory, the AGT correspondence 
for pure $SU(N)$ supersymmetric gauge theory on ${\IC}^2$ was proved in 
\cite{SchiffmannVasserot, Maulik:2012wi} 
(see \cite{Braverman:2014xca} for a generalization to all simply-laced 
gauge groups).
}). 

Their conjecture was further generalized to correspondences between 
$SU(N)$ instanton partition functions on ${\IC}^2$ and $\mathcal{W}_N$ 
conformal blocks \cite{Wyllard:2009hg, Mironov:2009by},  
$SU(2)$ instanton partition functions 
on ${\IC}^2/{\IZ}_2$ and $\mathcal{N}=1$ super-Virasoro conformal blocks 
\cite{Belavin:2011pp, Bonelli:2011jx, Belavin:2011tb, Bonelli:2011kv, 
Ito:2011mw, Belavin:2012aa},  
$SU(2)$ instanton partition functions on ${\IC}^2/{\IZ}_4$ and conformal 
blocks of $S_3$ parafermion algebra 
\cite{Wyllard:2011mn, Alfimov:2011ju},  \textit{etc.} 

In \cite{Nishioka:2011jk}, by considering $N$ M5-branes compactified on 
${\IC}^2/{\IZ}_n$ with $\Omega$-deformation, Nishioka and Tachikawa, following 
a proposal in \cite{Belavin:2011pp}, suggested that $\mathcal{N}=2$ $SU(N)$ 
supersymmetric gauge theories on ${\IC}^2/{\IZ}_n$ are in correspondence with 
2D CFTs with $n$-th parafermion $\mathcal{W}_N$ symmetry, which we refer to as 
$\mathcal{W}^{\, para}_{N, n}$, and affine $\widehat{\mathfrak{sl}}(n)_N$ symmetry.

In \cite{Belavin:2011sw}, it was proposed that the AGT correspondence 
for $U(N)$ supersymmetric gauge theory on ${\IC}^2/{\IZ}_n$ can 
be understood in terms of a 2D CFT based on the algebra
\begin{align}
\mathcal{A}(N,n;p)=\mathcal{H} \oplus 
\widehat{\mathfrak{sl}}(n)_N \oplus 
\frac{\widehat{\mathfrak{sl}}(N)_n \oplus \widehat{\mathfrak{sl}}(N)_{p-N}}
{\widehat{\mathfrak{sl}}(N)_{n+p-N}},
\label{alg_A_intro}
\end{align}
which acts on the equivariant cohomology of the moduli space of $U(N)$ 
instantons on ${\IC}^2/{\IZ}_n$, $n = 2, 3, \ldots\, .$ 
Here, 
the first factor $\mathcal{H} \cong \mathfrak{u}(1)$ is the affine Heisenberg 
algebra, the second factor is the affine $\mathfrak{sl}(n)$ level-$N$ algebra, 
and the third (coset) factor is the $\mathcal{W}^{\, para}_{N, n}$ algebra, 
whose parameter $p$, which controls the central charge
\footnote{\, 
In general $p \in \IC$.
}, 
is related to the $\Omega$-deformation parameters $\epsilon_1, \epsilon_2$ by 
\begin{align}
\frac{\epsilon_1}{\epsilon_2}=-1-\frac{n}{p}\,.
\label{rat_omega_intro}
\end{align}
The coset factor gives 
a Virasoro algebra when $(N,n)=(2,1)$,  
a $\mathcal{W}_N$ algebra when $(N,n)=(N,1)$,  
an $\mathcal{N}=1$ super-$\mathcal{W}_N$ algebra when $(N,n)=(N,N)$,  
and 
an $S_3$ parafermion algebra when $(N,n)=(2,4)$. 
 
\subsection{Burge conditions}\label{subsec:Burge_int}

Let $p \ge N$ be a positive integer. 
For $n=1$, the $\widehat{\mathfrak{sl}}(n)_N$ factor in the algebra 
$\mathcal{A}(N,n;p)$ is trivialised, while the coset (third) factor 
describes the $\mathcal{W}_N$ $(p, p+1)$-minimal model. 
In \cite{Bershtein:2014qma, Alkalaev:2014sma} for $(N, n)=(2,1)$ and 
further in \cite{Belavin:2015ria} for $(N, 1)$, $N = 3, 4, \ldots$ , 
it was shown that to obtain minimal model conformal blocks from the 
$SU(N)$ instanton partition functions on ${\IC}^2$ with 
$\Omega$-deformation \eqref{rat_omega_intro}, 
we need to remove the \textit{non-physical poles}, 
corresponding to $\mathcal{W}_N$ minimal model null states, from 
the instanton partition functions. 
These non-physical poles emerge when the Coulomb and mass parameters 
of the gauge theory take special values labeled by integers $r_I, s_I$, 
$I=1,\ldots,N-1$ with $N-1 \le \sum_{I=1}^{N-1}r_I \le p-1$, 
$N-1 \le \sum_{I=1}^{N-1}s_I \le p$. 
The conditions that exclude the non-physical poles 
were shown to be \textit{($N$-)Burge conditions} 
(see \cite{Burge, Foda:1997du} for $N=2$ and 
\cite{GesselKrattenthaler, Feigin:2010qea1, Feigin:2010qea2} 
for general $N$)
\begin{align}
Y_{I,i} \ge Y_{I+1, i+r_I-1} - s_I+1 \quad 
\textrm{for}\ i \ge 1,\ 0 \le I < N,
\label{Burge_intro}
\end{align}
for $N$-tuples of Young diagrams $(Y_1,\ldots,Y_N)$ 
which define the instanton partition functions, where 
$Y_{0}=Y_N$, and 
$r_0=p-\sum_{I=1}^{N-1}r_I$, $s_0=p+1-\sum_{I=1}^{N-1}s_I$. 
For  $n \ge 2$, the coset factor in the algebra $\mathcal{A}(N,n;p)$ 
is considered to describe a $\mathcal{W}^{\, para}_{N, n}$ $(p, p+n)$-minimal model. 
In this paper, we show that the same ($N$-)Burge conditions above 
also remove the non-physical poles from the $SU(N)$ 
instanton partition functions on ${\IC}^2/{\IZ}_n$ with $\Omega$-deformation
\eqref{rat_omega_intro}.

\subsection{Trivialization of the coset factor}

For $p=N$, the coset factor in the algebra $\mathcal{A}(N,n;p)$ is 
trivialized (the partition function reduces to $1$), 
\begin{align}
\mathcal{A}(N,n;N)=\mathcal{H} \oplus \widehat{\mathfrak{sl}}(n)_N,
\label{alg_A_triv_intro}
\end{align}
and the $SU(N)$ instanton partition functions on ${\IC}^2/{\IZ}_n$ 
provide the $\widehat{\mathfrak{sl}}(n)_N$ WZW conformal blocks. 
Since all parameters are now integral (or at least non-generic), 
the affine factor will include non-physical poles due to null 
states. To remove these, we need to impose the appropriate Burge 
conditions \eqref{Burge_intro} on the gauge theory side. 
In the present work, we show that the integrable 
$\widehat{\mathfrak{sl}}(n)_N$ WZW conformal blocks can be extracted 
from the instanton partition functions by an appropriate choice of 
the parameters and imposing the appropriate Burge conditions.

\subsection{Plan of the paper}

In Section \ref{sec:instantons} we briefly recall 
the generating functions of coloured Young diagrams and 
the instanton partition functions in $\mathcal{N}=2$ $U(N)$ 
supersymmetric gauge theories on ${\IC}^2/{\IZ}_n$ with 
$\Omega$-deformation. 
The relevant AGT-corresponding 2D CFTs are reviewed in Section 
\ref{sec:agt_rel}. 
In Section \ref{sec:Burge} we derive the Burge conditions 
(Proposition \ref{prop:Burge}) from the requirement that 
the $SU(N)$ instanton partition functions on ${\IC}^2/{\IZ}_n$, 
with $\Omega$-deformation \eqref{rat_omega_intro}, labeled by 
a positive integer $p$, do not have non-physical poles of the 
type described in Section \ref{subsec:Burge_int}. 
In subsequent sections, we only consider the Burge conditions 
that correspond taking $p=N$, which we need to trivialize the 
coset factor. 
In Section \ref{subsec:t_ref_red_ch}, by imposing the Burge 
conditions, we introduce what we refer to as 
\textit{Burge-reduced generating functions} of the coloured Young diagrams.
In Proposition \ref{Prop:Chern=Char} we show,
using the crystal graph theory developed by the Kyoto group
\cite{DJKMO:1989}, that these coincide
with the integrable $\widehat{\mathfrak{sl}}(n)_N$ WZW characters.
In Section \ref{subsec_red_inst_pf}, we introduce what we refer to
as \textit{Burge-reduced instanton partition functions}, by imposing the 
appropriate Burge conditions, and find that specific integrable
$\widehat{\mathfrak{sl}}(n)_N$ WZW conformal blocks are obtained from 
them (Conjectures \ref{conj:ex1}, \ref{conj:ex2} and \ref{conj:ex3}).
Our proposal, for computing WZW conformal blocks from models based on 
the $\mathcal{A} (N, n; N)$ algebra, is tested in Section \ref{sec:examples} 
for 
$(N,n)=(2,2), (2,3)$ and $(3,2)$. Finally, in Section \ref{sec:remarks} 
we make some remarks. 
In Appendix \ref{app:cartan}, we summarize the notation of Lie algebras that 
is used in this paper, 
in Appendix \ref{app:agt_check}, we review some AGT correspondences 
to confirm our conventions, and 
in Appendix \ref{app:wzw_4pt}, we recall a class of integrable 
$\widehat{\mathfrak{sl}}(n)_N$ WZW 4-point conformal blocks computed 
in \cite{Knizhnik:1984nr}, which we compare in Section \ref{sec:examples} 
with our results.

\section{$U(N)$ instanton counting on ${\IC}^2/{\IZ}_n$}\label{sec:instantons}

\textit{\noindent
We review how the moduli space of $U(N)$ instantons on ${\IC}^2/{\IZ}_n$ 
with $\Omega$-deformation is 
characterized by coloured Young diagrams, and define 
the generating functions of the coloured Young diagrams and 
the instanton partition functions.}

\subsection{Characterisation of the instanton moduli space by 
coloured Young diagrams}
\label{sec:characterisation}

Consider the $U(N)$ instantons on ${\IC}^2/{\IZ}_n$, where
${\IZ}_n$ acts on $(z_1,z_2)\in {\IC}^2$ by
\begin{align}
{\IZ}_n:
\quad 
\ll                          z_1,                             z_2 \rr \ \ 
\to \ \ 
\ll \e^{\frac{2 \pi i}{n}\, \sigma}\, z_1,\, \e^{-\frac{2 \pi i}{n} \, \sigma}\, z_2 \rr ,  
\qquad
\sigma=0,1,\ldots,n-1,
\label{orb_c2}
\end{align}
and introduce the $\Omega$-deformation parameters $(\epsilon_1, \epsilon_2)$ 
\cite{Nekrasov:2002qd, Nekrasov:2003rj}, by 
\begin{equation}
U(1)^2: \quad 
\ll                 z_1,                 z_2 \rr \to 
\ll \e^{\epsilon_1} z_1, \e^{\epsilon_2} z_2 \rr \,.
\end{equation}
Using localization, the $U(N)$ instantons on ${\IC}^2/{\IZ}_n$ are described 
by the fixed points, on the instanton moduli space, of the $U(1)^2 \times U(1)^{N}$ 
torus generated by $\e^{\epsilon_{1}}$, $\e^{\epsilon_{2}}$ and $\e^{a_{I}}$, 
where $a_{I}$, $I=1, \ldots, N$, are the Coulomb parameters in the $U(N)$ gauge 
theory. 
The Coulomb parameters have charges $\sigma_{I} \in \{0,1,\ldots,n-1\}$ 
under the ${\IZ}_n$ action
\begin{align}
{\IZ}_n:\quad a_I \ \ \to \ \ 
\e^{\frac{2 \pi i}{n}\, \sigma_{I}}\, a_I\,.
\label{orb_coulomb}
\end{align}
Let $Y^{\sigma}$ be a coloured Young diagram, with a ${\IZ}_n$ charge 
$\sigma\in \{0,1,\ldots,n-1\}$, in other words, $Y^{\sigma}$ is composed 
of boxes such that the box at position $(i,j) \in Y^{\sigma}$ is assigned 
the colour $\sigma-i+j$ (mod $n$). The fixed points of $U(N)$ $k$-instanton 
moduli space on ${\IC}^2/{\IZ}_n$ are labeled by $N$-tuples of coloured Young 
diagrams $\bY^{\bsig}=(Y_{1}^{\sigma_{1}},\ldots,Y_{N}^{\sigma_{N}})$ with 
\begin{align}
k=\left|\bY^{\bsig}\right|:=\sum_{I=1}^N \left|Y_{I}^{\sigma_{I}}\right|
\end{align}
total number of boxes \cite{KronheimerNakajima90, Fucito:2004ry}, 
where $Y_{I}^{\sigma_{I}}$ are charged by \eqref{orb_coulomb}, and 
$\left|Y_{I}^{\sigma_{I}}\right|$ denotes the number of boxes in $Y_{I}^{\sigma_{I}}$. 
Let $N_{\sigma}$ and $k_{\sigma}$ be the number of Young diagrams with charge 
$\sigma$ and the total number of boxes with colour $\sigma$, respectively. 
Then,
\begin{align}
\sum_{\sigma=0}^{n-1}N_{\sigma} = N,\qquad
\sum_{\sigma=0}^{n-1}k_{\sigma} = k\,.
\label{charge_N_k}
\end{align}
In what follows, we use $\bN$ to denote the
sequence $[N_0,\ldots,N_{n-1}]$.%
\footnote{\label{Ftn:Vector2Weight}
It is often convenient to regard $\bN$ 
as a vector on the basis of fundamental weights of $\slchap{n}$.
That each $N_\sigma\ge0$ with $\sum_{\sigma=0}^{n-1}N_{\sigma} = N$
then implies that $\bN\in P^{+}_{n,N}$.
Moreover, if we set $\lambda=\partit(\bN)$ using (\ref{Eq:Wt2Par}),
then $(\sigma_1,\sigma_2,\ldots)=\lambda^T$, the partition
conjugate to $\lambda$, by Lemma \ref{Lem:Par2Wt}.
}
Figure \ref{fig:young_ex} shows an example of a coloured Young diagram.
\begin{figure}[t]
\centering
\includegraphics[width=90mm]{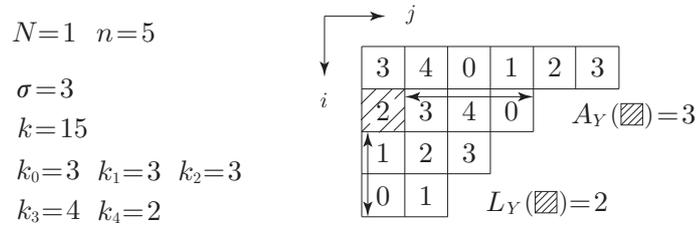}
\caption{An example of a coloured Young diagram $Y=Y^{\sigma}$ with 
charge $\sigma=3$, 
$k=15$, in the case of $(N,n)=(1,5)$. 
For $\wsq=(2,1)$, the arm length and the leg length 
defined in \eqref{arm_leg} are    
$A_{Y}(\wsq)=3$ and $L_{Y}(\wsq)=2$, respectively.}
\label{fig:young_ex}
\end{figure}

\begin{remark}\label{rem:order_charge}
In this paper, without less of generality, we assume
\begin{align}
\sigma_{1}\ge \sigma_{2}\ge \ldots \ge\sigma_{N},
\label{Eq:sigmaOrd}
\end{align}
by arranging the ordering of the Coulomb parameters. 
Then, given non-negative integers $\bN=[N_0,\ldots,N_{n-1}]$,
the charges $\sigma_1,\ldots,\sigma_N$ are uniquely fixed
by the above prescription.
This implies that $\bsig=(\sigma_1,\ldots,\sigma_N,0,0,\ldots)$
is a partition having at most $N$ non-zero parts and $\sigma_1<n$.
\end{remark}

Using the above notation of coloured Young diagrams, 
the $U(N)$ instantons on ${\IC}^2/{\IZ}_n$ are 
further characterized by the first Chern class of the gauge bundle
\begin{align}
c_1 = \sum_{\sigma=0}^{n-1} \mathfrak{c}_{\sigma}\, c_1(\mathcal{T}_{\sigma}),
\label{first_chern}
\end{align}
where
\begin{align}
\label{fc_label}
\mathfrak{c}_{\sigma} 
& = 
N_{\sigma} + \delta k_{\sigma-1} - 2\delta k_{\sigma} + \delta k_{\sigma+1} 
\\
& = 
N_{\sigma} - \sum_{i=0}^{n-1} A_{\sigma i}\, \delta k_i,
\qquad
\sigma=0,1,\ldots,n-1, \quad \delta k_{\sigma}:=k_{\sigma}-k_0,
\nonumber
\end{align}
with $k_n = k_0$ and $k_{-1} = k_{n-1}$. 
Here, $c_1(\mathcal{T}_{\sigma})$ is the first Chern class 
of the vector bundle $\mathcal{T}_{\sigma}$ on the ALE space with holonomy 
$\e^{2\pi i \sigma/n}$, and $A$ is the Cartan matrix of $\slchap{n}$.
Note that $c_1(\mathcal{T}_0)=0$, and the instanton moduli space is labelled 
by the $n-1$ integers $\bfc=(\mathfrak{c}_1,\ldots,\mathfrak{c}_{n-1})$. 
The inverse of the Cartan matrix $\oA$ of the
finite dimensional $\slfin{n}$ (see Appendix \ref{Sec:Notationfin})
allows the relations \eqref{fc_label} to be inverted, giving
\begin{align}
\delta k_{\sigma} = \sum_{i=1}^{n-1} \Bigl(\oA^{-1}\Bigr)_{\sigma\, i} 
\bigl(N_{i} - \mathfrak{c}_{i}\bigr) ,\qquad
\Bigl(\oA^{-1}\Bigr)_{\sigma i}
=\min\{\sigma,i\}-\frac{\sigma i}{n}
\label{delk_inv_cartan}
\end{align}
for $1\le\sigma<n$.

\subsection{Generating functions of coloured Young diagrams}\label{subsec:character}

Let $\mathcal{P}_{\bsig;\bdel \bk}$ be the set of 
$N$-tuples of coloured Young diagrams $\bY^{\bsig}$ 
labelled by the charges $\bsig=(\sigma_1,\ldots,\sigma_N)$ and  
$\bdel \bk=(\delta k_1, \ldots, \delta k_{n-1})$. 
We introduce a generating function of the coloured Young diagrams, 
that counts the number of torus fixed points of the $U(N)$ instanton moduli 
space on ${\IC^2}/{\IZ}_n$, as
\begin{equation}
{X}_{\bsig;\bdel \bk}(\fq)=
\sum_{\bY^{\bsig}\in \mathcal{P}_{\bsig;\bdel \bk}}
\fq^{\, \frac{1}{n}\left|\bY^{\bsig}\right|}\,.
\label{inst_ch_cpt}
\end{equation}

\begin{exam}[$n=1$]
For  $n=1$, $\bdel \bk=\emptyset$, 
the generating function \eqref{inst_ch_cpt} is
\begin{align}
{X}_{\boldsymbol{0};\emptyset}(\fq)&=
\chi_{\mathcal{H}}(\fq)^N:=
\frac{1}
{\left( \fq;\fq \right)_{\infty}^N}\,,
\label{ex_ch_n_1}
\end{align}
where 
\begin{align}
\left(  a;\fq \right)_{\infty}=\prod_{n=0}^{\infty}
\left( 1-a\, \fq^n \right) \,.
\end{align}
\end{exam}

\begin{exam}[$N=1$, see \cite{Fujii:2005dk, Alfimov:2013cqa}]
For  $N=1$, the generating function \eqref{inst_ch_cpt} with a charge 
$\sigma\in \{0,1,\ldots,n-1\}$ and 
$\bdel \bk=(\delta k_1, \ldots, \delta k_{n-1})$ is 
\begin{align}
{X}_{(\sigma);\bdel \bk}(\fq) =
\frac{1}{\left( \fq;\fq \right)_{\infty}^n}\, 
\fq^{\sum_{i=1}^{n-1}\ll \delta k_{i}^2 + \frac{\delta k_{i}}{n} 
-\delta k_{i-1}\, \delta k_{i} - \delta_{\sigma i}\, \delta k_{i}\rr }\,.
\label{ex_ch_N1n_g}
\end{align}
For example, for $(N, n)=(1, 2)$, the generating functions are
\begin{align}
{X}_{(0);(\ell)}(\fq) =
\frac{\fq^{\frac12 \ell (2\ell+1)}}
{\left( \fq;\fq \right)_{\infty}}\,,
\quad
{X}_{(1);(\ell)}(\fq) =
\frac{\fq^{\frac12 \ell (2\ell-1)}}
{\left( \fq;\fq \right)_{\infty}}\,.
\label{ex_ch_N1n2}
\end{align}
\end{exam}

\begin{exam}[$N=2, n=2$, see \cite{Belavin:2012aa}]\label{ex:N2n2_ch}
For $(N, n)=(2, 2)$, the generating functions \eqref{inst_ch_cpt} are 
\begin{align}
\begin{split}
&
{X}_{(0,0);(\ell)}(\fq) + {X}_{(1,1);(1+\ell)}(\fq) = 
\frac{\fq^{\frac12{\ell}(\ell+1)}}
{\left( \fq;\fq \right)_{\infty}}\, 
\chi_{\text{NS}}(\fq)^2,
\\
&
{X}_{(0,1);(\ell)}(\fq) = {X}_{(1,0);(\ell)}(\fq) = 
\frac{\fq^{\frac12{\ell^2}}}
{\left( \fq;\fq \right)_{\infty}}\, 
\chi_{\text{R}}(\fq)^2,
\end{split}
\end{align}
where
\begin{align}
\chi_{\textrm{NS}}(\fq)=
\frac{\left( -\fq^{\frac12};\fq \right)_{\infty}}
{\left( \fq;\fq \right)_{\infty}}\,,
\qquad
\chi_{\textrm{R}}(\fq)=
\frac{\left( -\fq;\fq \right)_{\infty}}
{\left( \fq;\fq \right)_{\infty}}\,,
\label{ns_r_ch}
\end{align}
are, respectively, the NS sector and Ramond sector characters in 
$\mathcal{N}=1$ super-Virasoro algebra.
\end{exam}

\subsection{Instanton partition functions}

To define instanton partition function, we introduce 
a fundamental building block, which is associated with 
$U(N)\times U(N)$ gauge symmetry, 
with coloured Young diagrams $\bY^{\bsig}=(Y_{1}^{\sigma_{1}},\ldots,Y_{N}^{\sigma_{N}})$ 
and $\bW^{\bsig^{\prime}}=(W_{1}^{\sigma_{1}^{\prime}},\ldots,W_{N}^{\sigma_{N}^{\prime}})$ by
\footnote{\, 
By shifting $a_J^{\prime} \to a_J^{\prime} - \mu$, 
it is possible to introduce the mass parameter $\mu$ of bifundamental hypermultiplet.}
\begin{align}
\begin{split}
Z_{\mathrm{bif}}\left( \ba, \bY^{\bsig}; \ba^{\prime}, \bW^{\bsig^{\prime}} \right) =
\prod_{I,J=1}^N 
&\prod_{\wsq \in Y_{I}^{\sigma_{I}}}^{*} 
E\left( -a_{I}+a_{J}^{\prime}, Y_{I}^{\sigma_{I}}(\wsq), W_{J}^{\sigma^{\prime}_{J}}(\wsq)\right) 
\\
&
\times \prod_{\wsq \in W_{J}^{\sigma^{\prime}_{J}}}^{*} 
\ll \epsilon_1+\epsilon_2-E\left( a_{I}-a_{J}^{\prime}, 
W_{J}^{\sigma^{\prime}_{J}}(\wsq), Y_{I}^{\sigma_{I}}(\wsq)\right)  \rr ,
\label{z_bif}
\end{split}
\end{align}
where
\begin{align}
E\left( P, Y(\wsq), W(\wsq)\right) =P - \epsilon_1\, L_{W}(\wsq) + 
\epsilon_2 \ll A_{Y}(\wsq)+1\rr \,.
\end{align}
Here the arm length $A_{Y}(\wsq)$ and the leg length $L_{Y}(\wsq)$ of a Young diagram $Y$ are 
defined by
\begin{align}
A_{Y}(\wsq)=Y_{i}-j,\quad
L_{Y}(\wsq)=Y_{j}^{T}-i, \quad \textrm{for}\ 
\textrm{all}\ \wsq=(i,j)\in {\IN}^2,
\label{arm_leg}
\end{align}
where $Y_{i}$ (resp. $Y_{j}^{T}$) is the length of the $i$-row in $Y$ 
(resp. the $j$-row in the transposed Young diagram $Y^{T}$ of $Y$, 
\textit{i.e.} the $j$-column in $Y$). 
The product $\prod_{\wsq \in Y}^{*}$ in \eqref{z_bif} means to take 
the ${\IZ}_n$ invariant factors in the product, 
modulo $2\pi i$, 
under the shift of parameters following \eqref{orb_c2} and \eqref{orb_coulomb},
\begin{align}
\epsilon_1\ \to \ \epsilon_1 + \frac{2 \pi i}{n},\quad
\epsilon_2\ \to \ \epsilon_2 - \frac{2 \pi i}{n},\quad
a_{I}\ \to \ a_{I} + \sigma_{I}\, \frac{2 \pi i}{n},\quad
a_{J}^{\prime}\ \to \ a_{J}^{\prime} + \sigma_{J}^{\prime}\,
\frac{2 \pi i}{n}\,.
\end{align}
Thus, the factors in the first and second products of \eqref{z_bif} 
are constrained, respectively, by
\begin{align}
\begin{split}
-\sigma_{I}+\sigma^{\prime}_{J} -L_{W_{J}^{\sigma^{\prime}_{J}}}(\wsq)
-A_{Y_{I}^{\sigma_{I}}}(\wsq)-1 \equiv 0 \ \ (\mathrm{mod}\ n),
\\
\sigma_{I}-\sigma^{\prime}_{J} -L_{Y_{I}^{\sigma_{I}}}(\wsq)
-A_{W_{J}^{\sigma^{\prime}_{J}}}(\wsq)-1 \equiv 0 \ \ (\mathrm{mod}\ n)\,.
\label{zn_condition}
\end{split}
\end{align}

\begin{defi}
Using the building block \eqref{z_bif}, the $U(N)$ instanton partition 
function on ${\IC}^2/{\IZ}_n$ with $N$ fundamental and $N$ anti-fundamental 
hypermultiplets, which is defined by an equivariant integration over the moduli 
space of instantons \cite{Nekrasov:2002qd} 
(see also \cite{Nekrasov:2003rj, Nakajima:2003uh, Shadchin:2005mx}), 
is \cite{Alfimov:2011ju} (see also \cite{Fucito:2004ry, Belavin:2012aa}),
\begin{align}
Z_{\bsig;\bdel \bk}^{\bb,\bb^{\prime}}\left(  \ba, \bm, \bm^{\prime}; \fq \right) =
\sum_{\bY^{\bsig}\in \mathcal{P}_{\bsig;\bdel \bk}}
\frac{Z_{\mathrm{bif}}\left( \bm, {\boldsymbol \emptyset}^{\bb}; \ba, \bY^{\bsig} \right) 
Z_{\mathrm{bif}}\left( \ba, \bY^{\bsig}; -\bm^{\prime}, 
{\boldsymbol \emptyset}^{\bb^{\prime}} \right) }
{Z_{\mathrm{vec}}\left( \ba, \bY^{\bsig}\right) }\,
\fq^{\, \frac{1}{n}\left|\bY^{\bsig}\right|},
\label{inst_pf}
\end{align}
where $\bm=(m_{1}, \ldots, m_{N})$ and 
$\bm^{\prime}=(m_{1}^{\prime}, \ldots, m_{N}^{\prime})$ 
are the mass parameters, associated with $U(N)^2$ flavor symmetry, 
of $N$ fundamental and $N$ anti-fundamental hypermultiplets, respectively. 
The denominator, which is the contribution from the $U(N)$ vector multiplet 
with Coulomb parameters $\ba=(a_{1}, \ldots, a_{N})$, is   
\begin{align}
Z_{\mathrm{vec}}\left( \ba, \bY^{\bsig}\right) =
Z_{\mathrm{bif}}\left(  \ba, \bY^{\bsig}; \ba, \bY^{\bsig} \right)\,.
\label{inst_vec_pf}
\end{align}
The instanton partition function \eqref{inst_pf} depends on not only 
the Chern classes $\bfc=(\mathfrak{c}_1,\ldots,\mathfrak{c}_{n-1})$, 
but also the ${\IZ}_n$ boundary charges $\bb=(b_{1}, \ldots, b_{N})$ and 
$\bb^{\prime}=(b_{1}^{\prime}, \ldots, b_{N}^{\prime})$, 
which take values in $\{0,1,\ldots,n-1\}$, assigned to 
the empty Young diagrams. Similar to \eqref{Eq:sigmaOrd}, we assume 
\begin{align}
b_{1} \ge b_{2} \ge \ldots \ge b_{N},\qquad 
b_{1}^{\prime} \ge b_{2}^{\prime} \ge \ldots \ge b_{N}^{\prime},
\end{align}
by arranging the ordering of the mass parameters.
\end{defi}

\section{2D CFT for $U(N)$ instantons on ${\IC}^2/{\IZ}_n$}
\label{sec:agt_rel}

\textit{\noindent
We recall various versions of the AGT correspondence, focusing on 
the algebra acting on the equivariant cohomology of the moduli 
space of $U(N)$ instantons on ${\IC}^2/{\IZ}_n$, and on 
explicit parameter relations.}

\subsection{Algebra on the moduli space of instantons and 2D CFT}

In \cite{Belavin:2011pp, Belavin:2011sw}, it has been proposed 
that the algebra
\begin{align}
\mathcal{A}(N,n;p)=
\mathcal{H} \oplus 
\widehat{\mathfrak{sl}}(n)_N \oplus 
\frac{\widehat{\mathfrak{sl}}(N)_n \oplus \widehat{\mathfrak{sl}}(N)_{p-N}}
{\widehat{\mathfrak{sl}}(N)_{p^{\prime}-N}},
\qquad
p^{\prime}=p+n,
\label{alg_A}
\end{align} 
naturally acts on the equivariant cohomology of 
the moduli space of $U(N)$ instantons on ${\IC}^2/{\IZ}_n$ with $\Omega$-deformation, 
where $\mathcal{H}\cong \mathfrak{u}(1)$ is the Heisenberg algebra.%
\footnote{\, 
The first works on this subject, in the absence of an $\Omega$-deformation, 
are by Nakajima \cite{Nakajima:1994nid, Nakajima:1998}.
}  
The parameter $p$ is identified with 
the $\Omega$-deformation parameters $\epsilon_1$, $\epsilon_2$ by the relation
\begin{align}
\frac{\epsilon_1}{\epsilon_2} = -\frac{p^{\prime}}{p} = -1-\frac{n}{p}.
\label{omega_pp}
\end{align}
This proposal implies that there exists a combined system of 2D CFTs,  
one with $\mathcal{H} \oplus \widehat{\mathfrak{sl}}(n)_N$ symmetry and 
the other with $\mathcal{W}^{\, para}_{N, n}$ symmetry, corresponding to 
4D $\mathcal{N}=2$ $U(N)$ supersymmetric gauge theory on ${\IC}^2/{\IZ}_n$ 
with $\Omega$-deformation \cite{Nishioka:2011jk}.
The central charges of these CFTs are   
\begin{align}
\begin{split}
c\left( \mathcal{H} \oplus \widehat{\mathfrak{sl}}(n)_N\right) &=
1+\frac{N\,(n^2-1)}{n+N},
\\
c\left( \mathcal{W}^{\, para}_{\, N, n} \right) 
&=
\frac{n\,(N^2-1)}{n+N} + 
\frac{N\,(N^2-1)}{n}\,
\frac{(\epsilon_1+\epsilon_2)^2}{\epsilon_1\, \epsilon_2}
\\
&\hspace{-0.25em}\mathop{=}\limits^{(\ref{omega_pp})}
\frac{n\,(N^2-1)}{n+N} -
\frac{n\,N\,(N^2-1)}{p\,(p+n)}.
\label{central_charge}
\end{split}
\end{align}
In \eqref{alg_A}, the first and second factors are realized by 
the $\widehat{\mathfrak{sl}}(n)_N$ WZW model with 
an additional $\mathfrak{u}(1)$ symmetry, 
and the third (coset) factor is realized by 
a $\mathcal{W}^{\, para}_{N, n}$ $(p, p+n)$-minimal model,%
\footnote{\,
While the $\mathcal{W}^{\, para}_{N, n}$ $(p, p+n)$-minimal models 
are in general not well-understood except in special cases, see 
\cite{Bouwknegt:1992wg} and references therein, in this work, 
we only need to assume that they exist.} 
where $p$ is taken to be a positive integer with $p \ge N$.

\subsection{Instanton partition functions as 4-point conformal blocks}
\label{subsec:4pt_inst}


We now provide the relations between the parameters of the instanton partition 
function \eqref{inst_pf} for $N\ge 2$ and those of the conformal blocks of 
the 4-point function on ${\IP}^1$ of primary fields 
$\psi_{\bmu_r}$ with momenta $\bmu_{r}$, $r=1,2,3,4$ 
(see Remark \ref{rem:cft_convention}), 
\begin{align}
\left<
\psi_{\bmu_1}( \infty)\, \psi_{\bmu_2}(1)\, \psi_{\bmu_3}(\fq)\, \psi_{\bmu_4}(0) 
\right>_{{\IP}^1}^{\mathcal{W}^{\, para}_{N, n}}
\label{para_WN_block}
\end{align}
in the $\mathcal{W}^{\, para}_{N, n}$ CFT described by the coset factor in 
\eqref{alg_A}. Using the notation of 
the finite dimensional Lie algebra $\mathfrak{sl}(N)$ in 
Appendix \ref{Sec:Notationfin} for $M=N$, 
we propose that the mass parameters $\bm$ and $\bm^{\prime}$ 
in \eqref{inst_pf} are related to the external momenta $\bmu_{r}$ of the four 
primary fields by
\begin{align}
\begin{split}
&
2\, \bmu_1 = \ll \epsilon_1+\epsilon_2\rr  \overline{\rho} + 
\sum_{I=1}^{N-1} \ll m_{I}-m_{I+1}\rr  \overline{\Lambda}_I,\quad 
2\, \bmu_2 =
\ll \sum_{I=1}^{N} m_{I} \rr  \overline{\Lambda}_{N-1},
\\
&
2\, \bmu_4 = \ll \epsilon_1+\epsilon_2\rr  \overline{\rho} - 
\sum_{I=1}^{N-1} \ll m_{I}^{\prime}-m_{I+1}^{\prime}\rr  \overline{\Lambda}_I,\quad
2\, \bmu_3 =
\ll \sum_{I=1}^{N} m_{I}^{\prime} \rr  \overline{\Lambda}_{1}.
\label{mass_rel}
\end{split}
\end{align}
We consider this as a generalisation of the $n=1$ case 
in \cite{Alday:2009aq, Wyllard:2009hg, Mironov:2009by, Kanno:2009ga, Itoyama:2009sc} 
to positive integer $n$.
By writing $\bmu_{2}=\mu_{2}\, \overline{\Lambda}_{N-1}$, $\bmu_{3}=\mu_{3}\, \overline{\Lambda}_{1}$, 
and $\bmu_{r}=\sum_{I=1}^{N-1}\mu_{r,I}\, \overline{\Lambda}_I$ for $r=1,4$, 
the relations \eqref{mass_rel} are equivalent to
\begin{align}
\begin{split}
&
2\, \mu_{1,I} = \ll \epsilon_1+\epsilon_2\rr  + \ll m_{I}-m_{I+1}\rr ,
\quad 
2\, \mu_2 = \sum_{I=1}^{N} m_{I},
\\
&
2\, \mu_{4,I} = \ll \epsilon_1+\epsilon_2\rr  - 
\ll m_{I}^{\prime} - m_{I+1}^{\prime}\rr ,
\quad
2\, \mu_3 = \sum_{I=1}^{N} m_{I}^{\prime},
\\
\Longleftrightarrow\quad
&m_I=\ll I-\frac{N+1}{2}\rr \ll \epsilon_1+\epsilon_2\rr +
\frac{2}{N}\ll -\sum_{J=1}^{I-1}J\, \mu_{1,J}+
\sum_{J=I}^{N-1}\ll N-J\rr \mu_{1,J} + \mu_2\rr ,
\\
&m_I^{\prime}=-\ll I-\frac{N+1}{2}\rr \ll \epsilon_1+\epsilon_2\rr +
\frac{2}{N}\ll \sum_{J=1}^{I-1}J\, \mu_{4,J}-
\sum_{J=I}^{N-1}\ll N-J\rr \mu_{4,J} + \mu_3\rr .
\label{mass_rel_base}
\end{split}
\end{align}
Note that the momenta $\bmu_2$ and $\bmu_3$ of two of the primary fields 
are taken to be  
proportional to  $\overline{\Lambda}_1$ or $\overline{\Lambda}_{N-1}$, \textit{i.e.} $\mathcal{W}$-null, 
which ensures the matching of the number of free parameters 
$\{m_I, m_I^{\prime}\}_{I=1,\ldots,N}$ and $\{\mu_{1,I}, \mu_{2}, \mu_{3}, \mu_{4,I}\}_{I=1,\ldots,N-1}$ 
\cite{Wyllard:2009hg, Mironov:2009by, Kanno:2009ga}. 
The Coulomb parameters $\ba$ in \eqref{inst_pf} are related to  
the internal momenta $\bmu^v=\sum_{I=1}^{N-1} \mu^v_{I}\, \overline{\Lambda}_I$ by
\begin{align}
\begin{split}
&
2\, \bmu^v = \ll \epsilon_1 + \epsilon_2\rr \overline{\rho} + 
\sum_{I=1}^N a_I\, \boldsymbol{\varepsilon}_I,\qquad
\boldsymbol{\varepsilon}_I:=\be_I- \be_0,
\\
\Longleftrightarrow\quad
&
a_I - \frac{1}{N}\sum_{I=1}^N a_I 
= \inner{2\, \bmu^v - \ll \epsilon_1 + \epsilon_2\rr \overline{\rho}}{\be_I}.
\label{coulomb_rel}
\end{split}
\end{align}

\begin{remark}[$U(1)$ factor]
The $U(N)$ instanton partition function \eqref{inst_pf} contains a $U(1)$ factor 
coming from the Heisenberg algebra $\mathcal{H}$ in the algebra $\mathcal{A}(N,n;p)$. 
To obtain it, we need to impose the traceless condition 
\begin{align}
\sum_{I=1}^N a_I=0.
\end{align}
Then, following \cite{Alday:2009aq, Wyllard:2009hg, Mironov:2009by, Belavin:2011tb, Alfimov:2011ju}, 
we find an overall $U(1)$ factor for general $N$ and $n$ in 
the instanton partition function \eqref{inst_pf},
\begin{align}
Z_{\mathcal{H}}\left( \bm, \bm^{\prime}; \fq \right)  :=
\left( 1-\fq \right) ^{\frac{\ll \sum_{I=1}^N m_I\rr 
\ll \epsilon_1+\epsilon_2 - \frac{1}{N}\sum_{I=1}^N m_I^{\prime}\rr }{n\, \epsilon_1\, \epsilon_2}}.
\label{u1_factor}
\end{align}
\end{remark}

In Appendix \ref{app:agt_check}, we confirm the above parameter relations 
and the $U(1)$ factor by checking some AGT correspondences.

By analogy with known minimal model CFTs, we propose that, 
in the $\mathcal{W}^{\, para}_{N, n}$ 
$(p, p+n)$-minimal models, the momenta should take the degenerate values
\begin{align}
2\, \bmu^{\br,\bs} = - \sum_{I=1}^{N-1} 
\ll \ll r_I-1 \rr \epsilon_1 + \ll s_I-1 \rr \epsilon_2\rr 
\overline{\Lambda}_I,
\label{deg_momenta}
\end{align}
where $\br=[r_0,r_1,\ldots,r_{N-1}]$ and $\bs=[s_0,s_1,\ldots,s_{N-1}]$
are sequences of positive integers for which
\begin{align}
\sum_{I=0}^{N-1}r_I = p\,,\qquad 
\sum_{I=0}^{N-1}s_I = p^{\prime}=p+n.
\label{condition_rs}
\end{align}
It will be useful to note that if $\br$ and $\bs$ are regarded as
vectors on the basis of fundamental weights of $\slchap{N}$
then $\br\in P^{++}_{N,p}$ and $\bs\in P^{++}_{N,p'}$.

\begin{remark}[Free field realization]\label{rem:cft_convention}
We check our normalization conventions 
by focusing on the well-understood $n=1$ CFT with 
$\mathcal{W}_N$ symmetry. 
In this case, one can introduce the energy-momentum tensor by
\begin{align}
T(z)=\frac{1}{2}\, \sum_{I=1}^N : \partial \phi_I(z)^2 : + \,
Q \inner{\overline{\rho}}{\partial^2 \boldsymbol{\phi}(z)},\quad
Q=\frac{\epsilon_1+\epsilon_2}{g_s},\quad g_s^2=-\epsilon_1\,\epsilon_2,
\label{em_tensor}
\end{align}
where $:\cdots:$ is the normal ordered product, 
\begin{align}
\boldsymbol{\phi}(z)=\sum_{I=1}^N \phi_I(z)\, \boldsymbol{\varepsilon}_I,
\qquad
\sum_{I=1}^N \phi_I(z)=0,\quad
\boldsymbol{\varepsilon}_I=\be_I- \be_0,
\end{align}
are $N$ free chiral bosons with 
\begin{align}
\partial \phi_I (z)\, \phi_J(w) 
= 
\frac{\inner{\boldsymbol{\varepsilon}_I}{\boldsymbol{\varepsilon}_J}}{z-w} + 
: \partial \phi_I (z)\, \phi_J(w) :
\, ,\quad
\inner{\boldsymbol{\varepsilon}_I}{\boldsymbol{\varepsilon}_J}=
\delta_{IJ} - \frac{1}{N},
\end{align}
and $g_s$ is introduced as a mass parameter just for a convention. 
Then, the Virasoro central charge
\begin{align}
c = ( N-1)  - N \left( N^2 - 1\right)  Q^2=
(N-1)  + 
N\,(N^2-1)\, \frac{(\epsilon_1+\epsilon_2)^2}{\epsilon_1\, \epsilon_2},
\end{align}
which is the one in \eqref{central_charge} for $n=1$, is obtained.
One can also introduce the primary field with momenta $\bmu$ by
\begin{align}
\psi_{\bmu}(z) = : \e^{\inner{2\, \frac{\bmu}{g_s}}{\boldsymbol{\phi}(z)}}:
\, , \quad
\bmu=\sum_{I=1}^{N-1} \mu_I\, \overline{\Lambda}_I,
\end{align}
which has the conformal dimension
\begin{align}
\Delta_{\bmu} = 2\inner{\frac{\bmu}{g_s}}{\frac{\bmu}{g_s}-Q\, \overline{\rho}} =
-\frac{2}{\epsilon_1\, \epsilon_2}\inner{\bmu}
{\bmu-\ll \epsilon_1+\epsilon_2 \rr \overline{\rho}},
\label{conf_dim_n1_w}
\end{align}
under the action of the energy-momentum tensor \eqref{em_tensor}. 
For example, when $N=2$ with the $\Omega$-background 
$\frac{\epsilon_1}{\epsilon_2}=-\frac{p^{\prime}}{p}$ 
(Virasoro $(p, p^{\prime})$-minimal model case), 
the conformal dimension of 
the primary field with degenerate momentum 
$2 \mu^{r,s}=-(r-1)\epsilon_1-(s-1)\epsilon_2$ is   
\begin{align}
\Delta_{\mu^{r,s}} = 
\frac{\mu^{r,s} \left(\epsilon_1+\epsilon_2-\mu^{r,s}\right)}
{\epsilon_1\, \epsilon_2} = 
\frac{\left(r\, p^{\prime} - s\, p\right)^2-
\left( p^{\prime} - p \right)^2}
{4\, p\, p^{\prime}}.
\end{align}
Similarly, for general $n$ the central charge and the conformal 
dimension of the primary field $\psi_{\bmu}(z)$ are found to be 
$c(\mathcal{W}^{\, para}_{\, N, n})$ in 
\eqref{central_charge} and $\Delta_{\bmu}/n$ in 
\eqref{conf_dim_n1_w}, 
respectively.
\end{remark}

\section{Burge conditions from $SU(N)$ instanton partition 
functions on ${\IC}^2/{\IZ}_n$}
\label{sec:Burge}

\textit{\noindent
We deduce the Burge conditions in Proposition \ref{prop:Burge} 
by looking at the non-physical poles of the $SU(N)$ instanton partition 
function \eqref{inst_pf}, with $\sum_{I=1}^N a_I=0$, on ${\IC}^2/{\IZ}_n$ 
with the rational $\Omega$-deformation \eqref{omega_pp}.}


For the rational $\Omega$-background \eqref{omega_pp}, \textit{i.e.} 
$p\,\epsilon_1+p^{\prime}\,\epsilon_2=0$, $p \ge N$, 
we see that the instanton partition function \eqref{inst_pf} 
with $\sum_{I=1}^N a_I=0$ has poles at
the values 
\begin{align}
a_I= a_I^{\br,\bs} := -\sum_{J=1}^{N-1} 
\inner{\overline{\Lambda}_J}{\be_I}  
\ll r_J\, \epsilon_1 + s_J\, \epsilon_2 \rr
=
-\sum_{J=I}^{N-1} \ll r_J\, \epsilon_1 + s_J\, \epsilon_2 \rr 
+\frac{1}{N}\sum_{J=1}^{N-1} J \ll r_J\, \epsilon_1 + s_J\, \epsilon_2\rr 
\label{sp_coulomb}
\end{align}
of the Coulomb parameters \eqref{coulomb_rel} 
corresponding to the degenerate momenta \eqref{deg_momenta}.
These poles correspond to the propagation of null-states and need to be removed.
Taking a shift of the central $U(1)$ factor in the $U(N)$ gauge symmetry, from 
\eqref{orb_c2} and \eqref{orb_coulomb}, into account, the ${\IZ}_n$ charges $\sigma_I$ 
assigned to the $a_I$ are related to $\br$ and $\bs$ by 
\begin{align}
\sigma_I - \sigma_{I+1} \equiv -r_I + s_I \quad (\mathrm{mod}\ n),
\quad
I=1,\ldots,N-1.
\label{charge_burge}
\end{align}
We refer to \eqref{charge_burge} as the ${\IZ}_n$ \textit{charge conditions}.

\begin{defi}[Burge Conditions]
For sequences $\br=[r_0,r_1,\ldots,r_{\levelaff-1}]$ and
$\bs=[s_0,s_1,\ldots,s_{\levelaff-1}]$ of positive integers,
the $N$-tuple $\bY=(Y_1,\ldots,Y_\levelaff)$ of Young diagrams
is said to satisfy the Burge conditions if
\begin{equation}\label{Eq:Burge}
Y_{I,i}\ge Y_{I+1,i+r_I-1}-s_I+1
\quad
\text{for}\ i \ge 1, \ 0 \le I < \levelaff,
\end{equation}
where we set $Y_0=Y_{\levelaff}$.
Then define $\burgeset{\br,\bs}$ to be the set of $\levelaff$-tuples
$\bY=(Y_1,\ldots,Y_\levelaff)$ of Young diagrams that satisfy
\eqref{Eq:Burge}.%
\footnote{\,
The Burge conditions specify sets of tuples of partitions
that are equivalent to certain cylindric (plane) partitions
defined in \cite{GesselKrattenthaler}
(what we define as $\burgeset{\br,\bs}$ here is equivalent to
the set of cylindric partitions denoted
${\mathcal C}^{\levelaff}_{\bs-\rho,\br-\rho}$
in \cite[Section 3.4]{Foda:2015bsa}).} 
\end{defi}

The following result is easily obtained by exchanging the roles
of the rows and columns in \eqref{Eq:Burge}:

\begin{lemm}[\cite{Bershtein:2014qma}]\label{Lem:BurgeDual}
Let $\bY=(Y_1,\ldots,Y_N)$ and, by conjugating the Young diagrams therein,
define $\bY^{T}=(Y^{T}_1,\ldots,Y^{T}_{N})$. Then
\begin{equation}\label{Eq:BurgeDual}
\bY\in\burgeset{\br,\bs} 
\quad
\Longleftrightarrow
\quad
\bY^{T}\in\burgeset{\bs,\br}.
\end{equation}
\end{lemm}

As with the Young diagrams in Section \ref{sec:characterisation},
we colour those in the $N$-tuples in  $\burgeset{\br,\bs}$.
Given a Young diagram $Y$, and $0\le\sigma<N$,
let $Y^\sigma$ denote a Young diagram in which the
box $(i,j)$ of $Y^{\sigma}$ is coloured $(\sigma-i+j)$ modulo $n$.
The value of $\sigma$ is known as the charge of $Y^{\sigma}$.
By Remark \ref{rem:order_charge},
$\bsig=(\sigma_1,\sigma_2,\ldots)$ is a partition
that has at most $N$ non-zero parts with $\sigma_1<n$.
For $\bY=(Y_1,\ldots,Y_N)\in\burgeset{\br,\bs}$,
we define $\bY^{\bsig}=(Y_1^{\sigma_1},\ldots,Y_N^{\sigma_N})$.
Set $\burgecolset{\br,\bs}{\bsig}$ to be the set of all
such $N$-tuples.
We will often drop the superscripts on $\bY^{\bsig}$ or
$Y^{\sigma}$ if these can be determined from the context.

\begin{prop}\label{prop:Burge}
If $\bY^{\bsig}\in \burgecolset{\br,\bs}{\bsig}$
then the instanton partition function \eqref{inst_pf}
at $a_I= a_I^{\br,\bs}$, in the background
$p\,\epsilon_1+p^{\prime}\,\epsilon_2=0$ where $p^{\prime}=p+n$, 
does not have poles.
\end{prop}

\begin{proof}
We follow the proof of \cite{Belavin:2015ria} 
for $n=1$ (see also \cite{Bershtein:2014qma, Alkalaev:2014sma} for 
$N=2$, $n=1$). 
At $a_I= a_I^{\br,\bs}$ with \eqref{omega_pp}, 
the instanton partition function \eqref{inst_pf} has poles if and 
only if the denominator vanishes, \textit{i.e.} 
there exists $\wsq \in Y_{I}$ such that
\begin{align}
E_{I, J}^{\br,\bs}(\wsq)+\mydelta=0,\quad
\mydelta=0\ \ \textrm{or}\ \ n,
\label{e_vanish}
\end{align}
where $n=p^{\prime}-p$ and
\begin{align}
\begin{split}
E_{I, J}^{\br,\bs}(\wsq)&=
\frac{p}{\epsilon_2}\, E\Big(a_{J}^{\br,\bs}-a_{I}^{\br,\bs}, Y_{I}(\wsq), Y_{J}(\wsq)\Big) 
\\
&=
\sum_{K=1}^{N-1} \inner{\overline{\Lambda}_K}{\be_J-\be_I}  
\ll r_K\, p^{\prime} - s_K\, p \rr 
+ p^{\prime}\, L_{Y_{J}}(\wsq) + 
p \ll A_{Y_{I}}(\wsq)+1\rr \,.
\end{split}
\end{align}
Because $E_{I, I}^{\br,\bs}(\wsq) \ne 0$ for $\wsq \in Y_{I}$, to find $\wsq \in Y_{I}$ 
which satisfies \eqref{e_vanish} we only need to consider the case (i) $I>J$ and case 
(ii) $I<J$.

\subsubsection*{Case (i) $I>J$}
In this case, the zero-condition $\eqref{e_vanish}$ is    
$E_{I+\ell, I}^{\br,\bs}(\wsq)+\mydelta=0$ 
for $\wsq \in Y_{I+\ell}$, where 
$1 \le I \le N-1$ and $1 \le \ell \le N-I$. 
By $\sum_{K=1}^{N-1} 
\inner{\overline{\Lambda}_K}{\be_I-\be_{I+\ell}}
=\sum_{K=1}^{N-1}\sum_{J=1}^{\ell}\delta_{K,I+J-1}$, 
this zero-condition is written as
\begin{align}
\sum_{J=1}^{\ell} \ll r_{I+J-1}\, p^{\prime} - s_{I+J-1}\, p \rr 
+ p^{\prime}\, L_{Y_{I}}(\wsq) + 
p \ll A_{Y_{I+\ell}}(\wsq)+1\rr +\mydelta=0,\quad 
\wsq \in Y_{I+\ell}\,.
\label{zero_c_i}
\end{align}
Let $d=\gcd(p, p^{\prime})$, $p=d\, p_d$ and $p^{\prime}=d\, p^{\prime}_d$, then 
the zero-condition \eqref{zero_c_i} is equivalent to
\begin{align}
\begin{split}
&
L_{Y_{I}}(\wsq)=
-\sum_{J=1}^{\ell} r_{I+J-1} -\gamma\, p_d - \delta_{\mydelta n},
\\
&
A_{Y_{I+\ell}}(\wsq)=
\sum_{J=1}^{\ell} s_{I+J-1} + \gamma\, p^{\prime}_d  -1 + \delta_{\mydelta n},
\quad \wsq \in Y_{I+\ell},
\label{zero_c_i_2}
\end{split}
\end{align}
where $\gamma$ is an indeterminate integer. 
For $\wsq=(i,j)\in Y_{I+\ell}$, using 
$L_{Y_I}(\wsq)=Y_{I,j}^{T}-i$, 
the zero-conditions \eqref{zero_c_i_2} imply that an obvious condition 
for any Young diagrams,
\begin{align}
Y_{I+\ell,j+A_{Y_{I+\ell}}(\wsq)}^{T} \ge i,
\end{align}
yields
\begin{align}
Y_{I+\ell,j+\sum_{J=1}^{\ell} s_{I+J-1} + \gamma\, p^{\prime}_d  -1 + \delta_{\mydelta n}}^{T} 
\ge 
Y_{I,j}^{T}
+ \sum_{J=1}^{\ell} r_{I+J-1} +\gamma\, p_d + \delta_{\mydelta n}\,.
\label{zero_c_i_3}
\end{align}

For the above zero-conditions, 
$\wsq \in Y_{I+\ell}$ needs to be restricted 
by the ${\IZ}_n$ condition like \eqref{zn_condition}   
\begin{align}
\sigma_I-\sigma_{I+\ell} - L_{Y_I}(\wsq) - A_{Y_{I+\ell}}(\wsq)-1 \equiv 0
\quad (\mathrm{mod}\ n).
\label{zn_c_i}
\end{align}
By the ${\IZ}_n$ charge conditions \eqref{charge_burge} and 
the zero-conditions \eqref{zero_c_i_2}, the ${\IZ}_n$ condition 
\eqref{zn_c_i} yields
\begin{align}
0 \equiv \gamma\, (p_d - p^{\prime}_d) \equiv -\frac{n}{d}\, \gamma \ \ 
(\mathrm{mod}\ n) 
\quad \Longleftrightarrow \quad
\gamma = d\, \gamma_d,
\label{zn_cd}
\end{align}
where the indeterminate integer $\gamma_d$ should be $\gamma_d \ge 0$ by 
$A_{Y_{I+\ell}}(\wsq)\ge 0$ and \eqref{condition_rs}.
As a result, the zero-condition \eqref{zero_c_i_3} yields 
\begin{align}
Y_{I+\ell,j+\sum_{J=1}^{\ell} s_{I+J-1} + \gamma_d\, p^{\prime}  -1 + \delta_{\mydelta n}}^{T} 
\ge 
Y_{I,j}^{T}
+ \sum_{J=1}^{\ell} r_{I+J-1} + \gamma_d\, p + \delta_{\mydelta n}.
\label{zero_condition_i}
\end{align}

Therefore, if conditions
\begin{align}
Y_{I,j}^{T} \ge 
Y_{I+\ell,j+\sum_{J=1}^{\ell} s_{I+J-1} + \gamma_d\, p^{\prime}  -1 + \delta_{\mydelta n}}^{T} 
-\sum_{J=1}^{\ell} r_{I+J-1} - \gamma_d\, p + 1 - \delta_{\mydelta n},\quad \gamma_d \ge 0
\label{non_zero_c_i}
\end{align}
are satisfied, there does not exist 
$\wsq \in Y_{I+\ell}$ such that $E_{I+\ell, I}^{\br,\bs}(\wsq)+\mydelta=0$.
These non-zero conditions follow from the ones for $\gamma_d=0$ and $\mydelta=0$:
\begin{align}
Y_{I,j}^{T} \ge 
Y_{I+\ell,j+\sum_{J=1}^{\ell} s_{I+J-1} -1}^{T} 
-\sum_{J=1}^{\ell} r_{I+J-1} + 1.
\label{non_zero_c_i_0}
\end{align}
All these non-zero conditions \eqref{non_zero_c_i_0} for 
$1 \le \ell \le N-I$ are obtained from the ones for $\ell=1$, \textit{i.e.} 
we arrive at the strongest non-zero conditions among them as
\begin{align}
Y_{I,j}^{T} \ge 
Y_{I+1,j + s_{I} -1}^{T} - r_{I} + 1,\quad I=1,\ldots,N-1,
\label{non_zero_condition_i}
\end{align}
which are the $I=1,\ldots,N-1$ cases of \eqref{Eq:Burge}
with $\br$ and $\bs$ interchanged.

\subsubsection*{Case (ii) $I<J$}
In this case, 
the zero-condition $\eqref{e_vanish}$ is    
$E_{I, I+\ell}^{\br,\bs}(\wsq)+\mydelta=0$ 
for $\wsq \in Y_{I}$, where 
$1 \le I \le N-1$ and $1 \le \ell \le N-I$. 
We repeat the proof of case (i). 
As \eqref{zero_c_i} and \eqref{zero_c_i_2}, the zero-condition 
$E_{I, I+\ell}^{\br,\bs}(\wsq)+\mydelta=0$ is
\begin{align}
-\sum_{J=1}^{\ell} \ll r_{I+J-1}\, p^{\prime} - s_{I+J-1}\, p \rr 
+ p^{\prime}\, L_{Y_{I+\ell}}(\wsq) + 
p \ll A_{Y_{I}}(\wsq)+1\rr +\mydelta=0,\quad 
\wsq \in Y_{I},
\label{zero_c_ii}
\end{align}
which is equivalent to
\begin{align}
\begin{split}
&
L_{Y_{I+\ell}}(\wsq)=
\sum_{J=1}^{\ell} r_{I+J-1} - \gamma_d\, p - \delta_{\mydelta n},
\\
&
A_{Y_{I}}(\wsq)=
-\sum_{J=1}^{\ell} s_{I+J-1} + \gamma_d\, p^{\prime}  -1 + \delta_{\mydelta n},
\quad \wsq \in Y_{I},
\label{zero_c_ii_2}
\end{split}
\end{align}
where we have used \eqref{zn_cd} obtained from the ${\IZ}_n$ condition.
From $A_{Y_{I}}(\wsq)\ge 0$ and \eqref{condition_rs}, 
the indeterminate integer $\gamma_d$ should be $\gamma_d \ge 1$.  
As in \eqref{zero_c_i_3}, these conditions yield a zero-condition
\begin{align}
Y_{I,j-\sum_{J=1}^{\ell} s_{I+J-1} + \gamma_d\, p^{\prime}  -1 + \delta_{\mydelta n}}^{T} 
\ge 
Y_{I+\ell,j}^{T}
- \sum_{J=1}^{\ell} r_{I+J-1} + \gamma_d\, p + \delta_{\mydelta n}.
\label{zero_condition_ii}
\end{align}

Therefore, if conditions
\begin{align}
Y_{I+\ell,j}^{T} \ge 
Y_{I,j-\sum_{J=1}^{\ell} s_{I+J-1} + \gamma_d\, p^{\prime}  -1 + \delta_{\mydelta n}}^{T} 
+\sum_{J=1}^{\ell} r_{I+J-1} - \gamma_d\, p + 1 - \delta_{\mydelta n},\quad \gamma_d \ge 1
\label{non_zero_c_ii}
\end{align}
are satisfied, there does not exist 
$\wsq \in Y_{I}$ such that $E_{I, I+\ell}^{\br,\bs}(\wsq)+\mydelta=0$. 
Among the non-zero conditions \eqref{non_zero_c_ii}, the strongest ones are 
   $\gamma_d=1$ and $\mydelta=0$:
\begin{align}
Y_{I+\ell,j}^{T} \ge 
Y_{I,j-\sum_{J=1}^{\ell} s_{I+J-1} + p^{\prime}  -1}^{T} 
+\sum_{J=1}^{\ell} r_{I+J-1} - p + 1.
\label{non_zero_c_ii_0}
\end{align}
In particular, for $\ell=N-I$, one obtains
\begin{align}
Y_{N,j}^{T} \ge 
Y_{1,j + s_{0} -1}^{T} - r_{0} + 1,
\label{non_zero_condition_ii}
\end{align}
which is the $I=0$ case of \eqref{Eq:Burge} with $r_0$ and $s_0$
exchanged. Together with \eqref{non_zero_condition_i} from case (i)
we thus obtain all cases of \eqref{Eq:Burge} with $\br$ and $\bs$
exchanged.
It is straightforward to see that together the conditions
\eqref{non_zero_condition_i} and \eqref{non_zero_condition_ii}
are stronger than the non-zero conditions \eqref{non_zero_c_ii_0}. 
Use of Lemma \ref{Lem:BurgeDual} then completes the proof.
\end{proof}

\section{Burge-reduced generating functions 
of coloured Young diagrams and $\slchap{n}_N$ WZW characters}
\label{subsec:t_ref_red_ch}

\noindent
\emph{%
In this and in subsequent sections, we concentrate on the case of $p=N$ 
in the algebra $\mathcal{A}(N,n;p)$.
This choice of parameters trivializes the coset factor,%
\footnote{\, 
When $p=N$, the central charge 
$c( \mathcal{W}^{\, para}_{\, N, n}) =0$ in (\ref{central_charge}).}
and we obtain  
$\mathcal{A}(N,n;N) = \mathcal{H} \oplus \widehat{\mathfrak{sl}}(n)_N$.
In this case, on imposing the Burge conditions \eqref{Eq:Burge}
on the generating functions of coloured Young diagrams
\footnote{\,
Without the Burge conditions, the generating functions
correspond to the partition functions of $\mathcal{N}=4$
topologically twisted $U(N)$ supersymmetric gauge theories 
\cite{Nakajima:1994nid, Vafa:1994tf}.
A description of the generating functions/WZW characters 
in terms of \textit{\lq orbifold partitions\rq}, and a realization in terms 
of intersecting D4 and D6-branes can be found in 
\cite{Dijkgraaf:2007sw, Dijkgraaf:2007fe}.}
and instanton partition functions, 
the $\widehat{\mathfrak{sl}}(n)_N$ WZW characters and conformal blocks emerge.
In this section we discuss these generating functions,
while the instanton partition functions are discussed
in Section \ref{subsec_red_inst_pf}.
We make use of notation and results pertaining to the
representation theory of $\slchap{n}$ that are described in
Appendix \ref{app:cartan}: for the current purposes the symbols
$M$ and $m$ in the appendix are replaced by $n$ and $N$, respectively.
}


\renewcommand{\dimaff}{{n}}    
\renewcommand{\levelaff}{{N}}    

In the case of the algebra
$\mathcal{A}(1,n;p)=\mathcal{H} \oplus \widehat{\mathfrak{sl}}(n)_1$ 
for $U(1)$ ($N=1$) instantons on ${\IC}^2/{\IZ}_n$,
the highest-weight representations have no null-states, 
and thus there is no restriction on the tuples of partitions $\bY$.
However, for $\levelaff>1$, eliminating the null states requires
the Burge conditions \eqref{Eq:Burge} to be imposed.

Throughout this section, $\bsig=(\sigma_1,\sigma_2,\ldots)$ is a
partition for which $\sigma_1<\dimaff$ and $\sigma_{\levelaff+1}=0$.
Thus, in particular, $\bsig$ has at most $\levelaff$ non-zero parts.
In addition, because we are restricting to the $p=\levelaff$ case,
the condition \eqref{condition_rs}
dictates that $r_0=r_1=\ldots=r_{\levelaff-1}=1$.
Thus, we define $\bfone=[1,1,\ldots,1]$ ($\levelaff$ components),
and use $\br=\bfone$ throughout this section.
The generating functions that we define below then depend
on a sequence $\bs=[s_0,s_1,\ldots,s_{\levelaff-1}]$ of positive integers
which satisfies \eqref{condition_rs} with $p'=\dimaff+\levelaff$,
and which satisfies the ${\IZ}_n$ charge conditions \eqref{charge_burge}.
With $r_0=r_1=\ldots=r_{\levelaff-1}=1$,
these requirements are readily accomplished by setting
\begin{equation}\label{Eq:sigma2s}
s_I=\sigma_{I}-\sigma_{I+1}+1
\end{equation}
for $1\le I <\levelaff$. 
Note then that \eqref{condition_rs} gives%
\footnote{\,
In fact, this solution is unique except
in the case where
$\sigma_1=\sigma_2=\cdots=\sigma_{\levelaff}$.
In that particular case, any other solution leads to the given solution
by cyclic permutation of the indices $I\in \{0,1,\ldots,\levelaff-1\}$.}
\begin{equation}\label{Eq:sigma2s0}
s_0=\sigma_\levelaff-\sigma_1+\dimaff+1\,.
\end{equation}

Using the above $\br=\bfone$ and $\bs$, we now define sets of
$\levelaff$-tuples of coloured Young diagrams
$\bY^{\bsig}=(Y_1^{\sigma_1},\ldots,Y_N^{\sigma_N})$
that respect the Burge conditions \eqref{Eq:Burge},
which now take the simplified form
\begin{align}
Y_{I,i}^{\sigma_I} \ge Y_{I+1, i}^{\sigma_{I+1}} - s_I+1 \quad
\textrm{for}\ i \ge 1,\ 0 \le I < N,
\label{sp_burge_c}
\end{align}
where we set $Y_0^{\sigma_0}=Y_N^{\sigma_N}$.
So define $\burgecolset{\bs}{\bsig}$ to be the set of all
$\levelaff$-tuples of coloured Young diagrams $\bY$
that respect \eqref{sp_burge_c}.
In order to impose the Chern relations \eqref{fc_label},
for each $\bY \in \burgecolset{\bs}{\bsig}$,
define $k_{\sigma}(\bY)$ to be the number of boxes in $\bY$
that are coloured $\sigma$ for each $\sigma\in\{0,1,\ldots,\dimaff-1\}$,
and then set $\delta k_{\sigma}(\bY)=k_{\sigma}(\bY)-k_0(\bY)$.

Now, to generalise $\rcolfun{\bsig}{\bell}(\qvar)$ in \eqref{inst_ch_cpt} 
and be able to relate the Chern classes \eqref{fc_label} 
to the representation theory 
of $\slchap{\dimaff}$,
we define two generating functions
$\burgecolfun{\bs}{\bsig}(\qvar,\bft)$ and
$\rburgecolfun{\bs}{\bsig}{\bell}(\qvar)$.
For fugacities $\bft=(\ft_1,\ldots,\ft_{\dimaff-1})$,
define the $U(N)$ $\ft$-refined Burge-reduced generating function 
\begin{equation}\label{Eq:BurgeGFdef}
\burgecolfun{\bs}{\bsig}(\qvar,\bft)
=
\sum_{\bY\in\burgecolset{\bs}{\bsig}}
\qvar^{\, \frac{1}{\dimaff}|\bY|}
\prod_{i=1}^{\dimaff-1} \ft_i^{\, \ffc_i(\bY)},
\end{equation}
where the $\ffc_i(\bY)$ are given by the Chern classes
(\textit{cf.}~\eqref{fc_label})
\begin{equation}\label{Eq:delta2chern}
\ffc_i(\bY)=N_i+\delta k_{i-1}(\bY)-2\delta k_i(\bY)+\delta k_{i+1}(\bY)
\end{equation}
for each $i\in\indaffm$,
with $\delta k_0(\bY)=\delta k_\dimaff(\bY)=0$.
In addition, for a vector
$\bell=(\ell_1,\ldots,\ell_{\dimaff-1})
 \in\IZ^{\dimaff-1}$,
define $\rburgecolset{\bs}{\bsig}{\bell}\subset\burgecolset{\bs}{\bsig}$
to be the set of all $\bY\in\burgecolset{\bs}{\bsig}$
for which $\delta k_i(\bY)=\ell_i$ for each $i\in\indaffm$.
Then define the Burge-reduced generating function 
\begin{equation}
\rburgecolfun{\bs}{\bsig}{\bell}(\qvar)
=
\sum_{\bY\in\rburgecolset{\bs}{\bsig}{\bell}}
\qvar^{\, \frac{1}{\dimaff}|\bY|}\,.
\label{red_ch_N}
\end{equation}
Because the values of $\ffc_i(\bY)$ are constant on each set
$\rburgecolset{\bs}{\bsig}{\bell}$, the generating function
\eqref{Eq:BurgeGFdef} can be written as
\begin{equation}\label{Eq:Chern2String}
\burgecolfun{\bs}{\bsig}(\qvar,\bft)
=
\sum_{\bell\in\IZ^{\dimaff-1}}
\rburgecolfun{\bs}{\bsig}{\bell}(\qvar)
\prod_{i=1}^{\dimaff-1} \ft_i^{\, N_i+\ell_{i-1}-2\ell_i+\ell_{i+1}}
\end{equation}
with $\ell_0=\ell_{\dimaff}=0$.
We also introduce the generating function
$\burgecolfun{\bs}{\bsig}(\qvar)$ defined as the specialisation
$\ft_1=\ft_2=\cdots=\ft_{\dimaff-1}=1$ of
$\burgecolfun{\bs}{\bsig}(\qvar,\bft)$:
\begin{align}
\burgecolfun{\bs}{\bsig}(\qvar)=
\sum_{\bY\in\burgecolset{\bs}{\bsig}}
\fq^{\, \frac{1}{\dimaff}|\bY|}=
\sum_{\bell\in\IZ^{\dimaff-1}}
\rburgecolfun{\bs}{\bsig}{\bell}(\qvar)\,.
\label{princ_gen_ch}
\end{align}

We now give expressions for $\burgecolfun{\bs}{\bsig}(\qvar,\bft)$
in terms of (graded) WZW characters
$\qchi^{\slchap{\dimaff}_\levelaff}_\Lambda(\qvar,\hat{\bft})$
that arise from level-$\levelaff$ irreducible highest weight
modules $L(\Lambda)$ of $\slchap{\dimaff}$.
The characters of these $\slchap{\dimaff}$-modules
are described in Appendix \ref{Sec:WZWChars}.
Let $\Lambda\in P^{+}_{\dimaff,\levelaff}$.
As also described in Appendix \ref{Sec:WZWChars},
the Virasoro algebra \emph{Vir} also acts on $L(\Lambda)$.
The central charge $c$ and conformal dimension $h_\Lambda$
of this \emph{Vir}-module are given by
\begin{equation}\label{Eq:CC-CD}
c=\frac{\levelaff(\dimaff^2-1)}{\levelaff+\dimaff},
\qquad
h_\Lambda=\frac{\inner{\Lambda}{\Lambda+2\rho}}{2(\levelaff+\dimaff)},
\end{equation}
respectively.

For indeterminates $\hat{\bft}=(\hat{\ft}_1,\ldots,\hat{\ft}_{\dimaff-1})$,
we define the (graded) character of $L(\Lambda)$ to be:
\begin{equation}\label{Eq:WZWCharDef}
\chi_{\Lambda}^{\slchap{\dimaff}_\levelaff}(\qvar,\hat{\bft})
=
\mathrm{Tr}_{L(\Lambda)}\,
\qvar^{\, L_0}
\prod_{i=1}^{\dimaff-1} \hat{\ft}_{i}^{\, H_i },
\end{equation} 
where $L_0$ is a Virasoro generator and $H_i$ are 
Chevalley elements in the Cartan subalgebra of $\slchap{\dimaff}$.
Making use of the crystal graph description of characters
of $\slchap{\dimaff}$ (see Appendix \ref{Sec:Crystal})
then leads to the following:

\begin{prop}\label{Prop:Chern=Char}
For a partition $\bsig=(\sigma_1,\sigma_2,\ldots)$ for which
$\sigma_1<\dimaff$ and $\sigma_{\levelaff+1}=0$,
define $\bs=[s_0,s_1,\ldots,s_{\levelaff-1}]$ by
\eqref{Eq:sigma2s} and \eqref{Eq:sigma2s0},
and set $\Lambda=\sum_{i=1}^{\levelaff}\fwt_{\sigma_i}$.
Then
\begin{equation}\label{Eq:Chern=Char}
\burgecolfun{\bs}{\bsig}(\qvar,\bft)=
\frac{\qvar^{\, w_{\Lambda}-h_\Lambda}}{\pochinf{\qvar}{\qvar}} \,
\qchi^{\slchap{\dimaff}_{\levelaff}}_\Lambda(\qvar,\hat{\bft}),
\end{equation}
where $\hat{\bft}=(\hat{\ft}_1,\ldots,\hat{\ft}_{\dimaff-1})$ is related to 
$\bft=(\ft_1,\ldots,\ft_{\dimaff-1})$ by
\begin{equation}\label{Eq:t2hat}
\hat{\ft}_i=\qvar^{\, -\frac{1}{2\dimaff}i(\dimaff-i)} \, \ft_i
\end{equation}
for $1\le i<\dimaff$, and
\begin{equation}\label{Eq:wdef}
w_\Lambda=
\frac{1}{2\dimaff}\sum_{i=1}^{\dimaff-1} i(\dimaff-i) N_i,
\end{equation}
when $\Lambda=[N_0,N_1,\ldots,N_{\dimaff-1}]$.
\end{prop}

\begin{proof}
Comparison of the conditions \eqref{Eq:Cylindrical} with \eqref{Eq:Burge}
shows that there is a bijection
$\multicol{\bsig}\to\burgeset{\bs,\bfone}$, with the map
$\bY\mapsto\bX$ from the former to the latter being obtained
by ignoring the colours.
Combining this with the bijection described by \eqref{Eq:BurgeDual}
then yields a bijection $\multicol{\bsig}\to\burgeset{\bfone,\bs}$
described by $\bY\mapsto \bX\mapsto \bX^T$.
Moreover, because of the differing ways in which the colours are
ordered in $\multicol{\bsig}$ and $\burgecolset{\bfone,\bs}{\bsig}$,
colouring $\bX^T$ to give an element of $\burgecolset{\bfone,\bs}{\bsig}$,
results in $\bY^T$.
Thus, in the expression \eqref{Eq:BurgeGFdef},
$\burgecolset{\bs}{\bsig}\equiv\burgecolset{\bfone,\bs}{\bsig}$
can be replaced by $\multicol{\bsig}$.
Noting that $|\bY|=\sum_{i=0}^{\dimaff-1}k_i(\bY)
=\dimaff k_0(\bY)+\sum_{i=1}^{\dimaff-1}\delta k_i(\bY)$,
and using \eqref{Eq:delta2chern}, then gives
\begin{equation}
\begin{split}
\burgecolfun{\bs}{\bsig}(\qvar,\bft)
&=
\sum_{\bY\in\multicol{\bsig}}
\qvar^{\, k_0(\bY)}
\prod_{i=1}^{\dimaff-1}
\qvar^{\, \frac{1}{\dimaff}\delta k_i(\bY)}
\ft_i^{\, N_i+\delta k_{i-1}(\bY)-2\delta k_i(\bY)+\delta k_{i+1}(\bY)}\\
&=
\sum_{\bY\in\multicol{\bsig}}
\qvar^{\, k_0(\bY)}
\prod_{i=1}^{\dimaff-1}
\ft_i^{\, N_i}
\left(\qvar^{\, \frac{1}{\dimaff}}
\frac{\ft_{i-1}\ft_{i+1}}{\ft_{i}^2}\right)^{\delta k_i(\bY)},
\end{split}
\end{equation}
where we set $\ft_0=\ft_\dimaff=1$.
Substituting for each $\ft_i$ using \eqref{Eq:t2hat} then shows that
\begin{equation}
\burgecolfun{\bs}{\bsig}(\qvar,\bft)
=
\qvar^{\, w_{\Lambda}}
\sum_{\bY\in\multicol{\bsig}}
\qvar^{\, k_0(\bY)}
\prod_{i=1}^{\dimaff-1}
\hat{\ft}_i^{\, N_i}
\left(\frac{\hat{\ft}_{i-1}\hat{\ft}_{i+1}}{\hat{\ft}_{i}^2}\right)
  ^{\delta k_i(\bY)}
\end{equation}
which yields \eqref{Eq:Chern=Char} using \eqref{Eq:SLCharDynkin3}.
\end{proof}
This result enables $\rburgecolfun{\bs}{\bsig}{\bell}(\qvar)$
to be expressed in terms of the
$\slchap{\dimaff}$ string functions $\xa^\Lambda_{\bell}(\qvar)$ in \eqref{Eq:SLstring}:
\begin{cor}\label{Cor:Chern=Char1}
For a partition $\bsig=(\sigma_1,\sigma_2,\ldots)$ for which
$\sigma_1<\dimaff$ and $\sigma_{\levelaff+1}=0$,
define $\bs=[s_0,s_1,\ldots,s_{\levelaff-1}]$ by
\eqref{Eq:sigma2s} and \eqref{Eq:sigma2s0},
and set $\Lambda=\sum_{i=1}^{\levelaff}\fwt_{\sigma_i}$.
Then for each $\bell=(\ell_1,\ldots,\ell_{\dimaff-1})\in\IZ^{\dimaff-1}$,
\begin{equation}\label{Eq:Chern=Char1}
\rburgecolfun{\bs}{\bsig}{\bell}(\qvar)=
\frac{\qvar^{\, \frac1{\dimaff}|\bell|}}{\pochinf{\qvar}{\qvar}} \,
\xa^\Lambda_{\bell}(\qvar),
\end{equation}
where we set $|\bell|=\sum_{i\in\indaff}\ell_i$.
\end{cor}
\begin{proof}
This results from reexpressing the left and right sides of
\eqref{Eq:Chern=Char} using
\eqref{Eq:Chern2String} and \eqref{Eq:SLCharDynkin} respectively,
using \eqref{Eq:t2hat},
and then using the fact that the matrix $\oA$ is invertible.
\end{proof}

Combining \eqref{Eq:Chern=Char} with \eqref{Eq:Principal}
enables a product expression to be given for 
$\burgecolfun{\bs}{\bsig}(\qvar)$:
\begin{cor}\label{Cor:Chern=Char2}
For a partition $\bsig=(\sigma_1,\sigma_2,\ldots)$ for which
$\sigma_1<\dimaff$ and $\sigma_{\levelaff+1}=0$, let
$\bs=[s_0,s_1,\ldots,s_{\levelaff-1}]$ be defined by
\eqref{Eq:sigma2s} and \eqref{Eq:sigma2s0}.
Then
\begin{equation}\label{Eq:ChernProd}
\burgecolfun{\bs}{\bsig}(\qvar)=
\frac{1}
{\vpochinf{\qvar^{1+\frac{\levelaff}{\dimaff}}}
 {\qvar^{1+\frac{\levelaff}{\dimaff}}}}
\prod_{\begin{subarray}{c}
        1\le i<j\le \dimaff+\levelaff\\[.5ex]
        i\notin\Omega,j\in\Omega
       \end{subarray}}
\frac{1}
{\vpochinf{\qvar^{\frac{j-i}{\dimaff}}}
          {\qvar^{1+\frac{\levelaff}{\dimaff}}}}
\prod_{\begin{subarray}{c}
        1\le i<j\le \dimaff+\levelaff\\[.4ex]
        i\in\Omega,j\notin\Omega
       \end{subarray}}
\frac{1}
{\vpochinf{\qvar^{1+\frac{\levelaff-j+i}{\dimaff}}}
          {\qvar^{1+\frac{\levelaff}{\dimaff}}}},
\end{equation}
where $\Omega=\{\levelaff+j-\sigma^T_j\,|\,j=1,\ldots,\dimaff\}$.
\end{cor}
\begin{proof}
From Appendix \ref{Sec:Principal},
we have $\Omega_{\Lambda}=\{\levelaff+j-\lambda_j\,|\,j=1,\ldots,\dimaff\}$,
where $\lambda=\partit(\Lambda)$ for
$\Lambda=\sum_{i=1}^{\levelaff}\fwt_{\sigma_i}$.
Lemma \ref{Lem:Par2Wt} shows that $\lambda=\sigma^T$
and thus $\Omega_{\Lambda}=\Omega$.


From \eqref{princ_gen_ch}, by \eqref{Eq:Chern=Char1}, 
\eqref{Eq:CrystalChar2}, and \eqref{Eq:PrincipalDef}, we obtain
\begin{equation*}
\begin{split}
\burgecolfun{\bs}{\bsig}(\qvar)
&=
\frac{1}{\pochinf{\qvar}{\qvar}}
\sum_{\bell\in\IZ^{\dimaff-1}}
\qvar^{\, \frac1{\dimaff}|\bell|}\,
\xa^\Lambda_{\bell}(\qvar)
=
\frac{e^{-\Lambda}}{\pochinf{\qvar}{\qvar}}
\left.
\xchi^{\, \slchap{\dimaff}}_\Lambda(\qvar,\bx)
\right|_{\{x_i\to \qvar^{-i/\dimaff},1\le i\le\dimaff\}}\\
&=
\frac{1}{\pochinf{\qvar}{\qvar}} \,
\mathrm{Pr}\,{\chi}^{\slchap{\dimaff}}_\Lambda(\qvar).
\end{split}
\end{equation*}
Use of \eqref{Eq:Principal} then gives \eqref{Eq:ChernProd}.
\end{proof}

In the cases in which $\levelaff=1$, the Burge conditions \eqref{Eq:Burge}
are vacuous (assuming that $r_1$ and $s_1$ are both positive).
Then $\rburgecolfun{\bs}{\bsig}{\bell}(\qvar)$ coincides with
$\rcolfun{\bsig}{\bell}(\qvar)$ defined by \eqref{inst_ch_cpt}.
In this $\levelaff=1$ case, we also set
$\colfun{\bsig}(\qvar,\bft)=\burgecolfun{\bs}{\bsig}(\qvar,\bft)$
and
$\colfun{\bsig}(\qvar)=\burgecolfun{\bs}{\bsig}(\qvar)$.
Because Proposition \ref{Prop:Chern=Char} and
Corollaries \ref{Cor:Chern=Char1} and \ref{Cor:Chern=Char2}
remain valid in this case,
when combined with expressions from Appendix \ref{Sec:WZWChars},
they lead to the following:

\begin{cor}\label{Cor:Chern=Char3}
Let $0\le k<\dimaff$.
Then
\begin{align}
\label{Eq:Chern=Char3a}
\colfun{(k)}(\qvar,\bft)
&=
\frac{1}{\pochinf{\qvar}{\qvar}^{\dimaff}}
\sum_{\bell\in\IZ^{\dimaff-1}}
\qvar^{-\ell_k+\sum_{i=1}^{\dimaff-1}(\ell_i^2-\ell_i\ell_{i-1}
         +\frac{1}{\dimaff}\ell_i)}
\prod_{i=1}^{\dimaff-1}
\ft_i^{\, \delta_{ik}+\ell_{i-1}-2\ell_i+\ell_{i+1}}\,,
\\
\rcolfun{(k)}{\bell}(\qvar)
\label{Eq:Chern=Char3b}
&=
\frac{1}{\pochinf{\qvar}{\qvar}^{\dimaff}}\,
\qvar^{-\ell_k+\sum_{i=1}^{\dimaff-1}(\ell_i^2-\ell_i\ell_{i-1}
         +\frac{1}{\dimaff}\ell_i)}\,,
\\
\colfun{(k)}(\qvar)
\label{Eq:Chern=Char3c}
&=
\frac{1}{\vpochinf{\qvar^{\frac1{\dimaff}}}{\qvar^{\frac1{\dimaff}}}},
\end{align}
where $\bell=(\ell_1,\ell_2,\ldots,\ell_{\dimaff-1})$
and $\ell_0=\ell_{\dimaff}=0$.
\end{cor}
\begin{proof}
Set $\Lambda=\fwt_k$.
Firstly, combining \eqref{Eq:Chern=Char1} with \eqref{Eq:CharAffLevel1xa}
gives \eqref{Eq:Chern=Char3b}.
Substituting \eqref{Eq:Chern=Char3b} into \eqref{Eq:Chern2String}
then gives \eqref{Eq:Chern=Char3a}.
For \eqref{Eq:Chern=Char3c} we use Corollary \ref{Cor:Chern=Char2}
with $\bsig=(k)$.
Then $\Omega=\{1,2,\ldots,\dimaff\}\backslash\{k+1\}$,
whereupon \eqref{Eq:ChernProd} yields \eqref{Eq:Chern=Char3c}.
\end{proof}
Note that \eqref{Eq:Chern=Char3b} accords with \eqref{ex_ch_N1n_g}.

Let $\bN=[N_0,N_1,\ldots,N_{\dimaff-1}]$ be such that
each $N_i\ge0$ with $\sum_{i=0}^{\dimaff-1}N_{i} = \levelaff$.
If we regard $\bN$ as an $\slchap{\dimaff}$ weight,
then $\bN=\sum_{i=0}^{n-1}N_{i}\fwt_i\in P^{+}_{\dimaff,\levelaff}$.
Let $\bsig=(\sigma_1,\sigma_2,\ldots)$ be the partition with
$\sigma_{\levelaff+1}=0$ such that
$\bN=\sum_{j=1}^\levelaff\fwt_{\sigma_j}$.
Then $\bsig=\lambda^T$, the partition
conjugate to $\lambda=\partit(\bN)$ defined by \eqref{Eq:Wt2Par}.
Now define $\bs=[s_0,s_1,\ldots,s_{\levelaff-1}]$ by
\eqref{Eq:sigma2s} and \eqref{Eq:sigma2s0},
and define the $SU(N)$ $\ft$-refined Burge-reduced generating function 
of coloured Young diagrams, by subtracting 
the Heisenberg factor $\mathcal{H}$ whose character is    
$\chi_{\mathcal{H}}(\fq)=\left(\fq;\fq\right)_{\infty}^{-1}$ in 
\eqref{ex_ch_n_1},
\begin{equation}
\widehat{X}^{\mathrm{red}}_{\bN}(\qvar,\bft)
=\pochinf{\qvar}{\qvar}
\times
\burgecolfun{\bs}{\bsig}(\qvar,\ft).
\label{inst_red_ch}
\end{equation}
Proposition \ref{Prop:Chern=Char} immediately shows that:

\begin{cor}\label{prop:Ngch}
If $\bN\in P^{+}_{\dimaff,\levelaff}$, then
\begin{equation}
\widehat{X}^{\mathrm{red}}_{\bN}(\qvar,\bft)
=
\qvar^{\, w_{\bN}-h_{\bN}}\,
\qchi^{\slchap{\dimaff}_{\levelaff}}_{\bN}(\qvar,\hat{\bft}),
\label{eqn:Ngch}
\end{equation}
where $\hat{\bft}$ is related to $\bft$ by \eqref{Eq:t2hat},
and $h_{\bN}$ and $w_{\bN}$ are given by \eqref{Eq:CC-CD} and \eqref{Eq:wdef}.
\end{cor}

This corollary implies that the Chern classes \eqref{fc_label} 
on the gauge side are identified with 
the eigenvalues of Cartan elements $H_i$ of $\slchap{\dimaff}$ on the CFT side. 

\begin{exam}
In the case of $N=1$, \eqref{eqn:Ngch} is particularly simple,
because then $h_{\bN}=w_{\bN}$.
For instance, for $(N,n)=(1,2)$,
\begin{align}
\begin{split}
\widehat{X}_{[1,0]}(\fq, \ft) &=\left( \fq;\fq \right)_{\infty}\,
\sum_{\ell \in {\IZ}} {X}_{(0);(-\ell)}(\fq)\, \ft^{2\ell}
=
\frac{1}{\left( \fq;\fq \right)_{\infty}}\, 
\sum_{j \in {\IZ}} \fq^{j^2}\, \hat{\ft}^{\, 2j}
=
\qchi^{\slchap{\dimaff}_{\levelaff}}_{[1,0]}(\qvar,\hat{\ft}),
\\
\widehat{X}_{[0,1]}(\fq, \ft) &=\left( \fq;\fq \right)_{\infty}\,
\sum_{\ell \in {\IZ}} X_{(1);(-\ell)}(\fq)\, \ft^{2\ell+1}
=
\frac{1}{\left( \fq;\fq \right)_{\infty}}\, 
\sum_{j \in {\IZ}+\frac12} \fq^{j^2}\, \hat{\ft}^{\, 2j}
=
\qchi^{\slchap{\dimaff}_{\levelaff}}_{[0,1]}(\qvar,\hat{\ft}),
\end{split}
\end{align}
where $\hat{\ft} =\fq^{-\frac14}\, \ft$.
\end{exam}

In Section \ref{sec:examples}, we will give explicit examples of  
Corollary \ref{prop:Ngch} for $(N,n)=(2,2), (2,3)$ and $(3,2)$ 
by comparing with the $\widehat{\mathfrak{sl}}(n)_N$ WZW characters 
computed using the Weyl-Kac character formula \eqref{wzw_ch_formula}.

Note that in the principally specialised case $\bft=(1,\ldots,1)$,
\begin{equation}
\widehat{X}^{\mathrm{red}}_{\bN}(\qvar,(1,\ldots,1))
=\pochinf{\qvar}{\qvar} \times 
\burgecolfun{\bs}{\bsig}(\qvar)
\label{inst_pr_red_ch}
\end{equation}
is immediately evaluated using the right side of \eqref{Eq:ChernProd}
with $\Omega=\{j+\sum_{i=0}^{j-1}N_i\,|\,j=1,\ldots,\dimaff\}$, and 
gives the $\widehat{\mathfrak{sl}}(n)$ principally specialised
character $\mathrm{Pr}\, \chi_{\bN}^{\widehat{\mathfrak{sl}}(n)}(\fq)$
in Appendix \ref{Sec:Principal}.


\section{Burge-reduced instanton partition functions 
and $\widehat{\mathfrak{sl}}(n)_N$ WZW conformal blocks}
\label{subsec_red_inst_pf}

\textit{\noindent We discuss how the integrable
$\widehat{\mathfrak{sl}}(n)_N$ WZW conformal 
blocks are extracted from the $SU(N)$ instanton 
partition functions on ${\IC}^2/{\IZ}_n$ with 
$\sum_{I=1}^N a_I=0$. 
}

\subsection{$U(1)$ instanton partition function}

In the $U(1)$ case, as was mentioned in Section \ref{subsec:t_ref_red_ch}, 
for generic $p$ (generic $\Omega$-background) one obtains the algebra 
$\mathcal{A}(1,n;p)=\mathcal{H} \oplus 
\widehat{\mathfrak{sl}}(n)_1$ acting on the equivariant cohomology of 
the moduli space of $U(1)$ instantons on ${\IC}^2/{\IZ}_n$. 
Let us consider the instanton partition function \eqref{inst_pf} for $N=1$ 
with vanishing Coulomb parameter $a=0$ and labelled by 
$\bN_{\sigma}=[N_0,\ldots,N_{n-1}]$, $N_i=\delta_{i \sigma}$, and $\delta k_i=0$. 
Following Corollary \ref{prop:Ngch}, the corresponding module in 
$\widehat{\mathfrak{sl}}(n)_1$ is the highest-weight module with 
$\Lambda=\Lambda_{\sigma}$. 
We define
\begin{align}
Z_{\bN_{\sigma}}^{b,b^{\prime}}(m, m^{\prime}; \fq)=
Z_{(\sigma);\boldsymbol{0}}^{b,b^{\prime}}(0, m, m^{\prime}; \fq),
\label{u1_triv_pf}
\end{align}
and make the following conjecture.

\begin{conj}
The $U(1)$ instanton partition function \eqref{u1_triv_pf} on ${\IC}^2/{\IZ}_n$ 
with $b^{\prime} = b$ and $\bN_{0}=[1,0,\ldots,0]$ is   
\begin{align}
Z_{\bN_{0}}^{b,b}(m, m^{\prime}; \fq)=
\left( 1-\fq \right) ^{\frac{m
\ll \epsilon_1+\epsilon_2-m^{\prime}\rr }{n\, \epsilon_1\, \epsilon_2}}
\left( 1-\fq \right) ^{-2\, h_{b}},
\label{u1_conj_inst}
\end{align}
where $h_{b}=h_{\bN_b}=\frac{b \left(n-b\right)}{2n}$ is the conformal 
dimension of the highest-weight state 
$\left|N_b\right>$ in the $\widehat{\mathfrak{sl}}(n)_1$ WZW model. 
The first factor is the $U(1)$ factor 
$Z_{\mathcal{H}}\left( m, m^{\prime}; \fq \right)$ in \eqref{u1_factor} 
for $N=1$, 
and the second factor is the 2-point function of 
$\widehat{\mathfrak{sl}}(n)_1$ WZW primary fields with 
highest-weights $\Lambda_{b}$ and $\Lambda_{n-b}$
\end{conj}

\subsection{$SU(N)$ Burge-reduced instanton partition functions}

For $N\ge 2$, in the same way that we defined the Burge-reduced
generating function \eqref{red_ch_N} of coloured Young diagrams,
we now introduce a reduced version of the instanton partition
function \eqref{inst_pf} by imposing the
specialized ones \eqref{sp_burge_c} for 
the Burge conditions \eqref{Eq:Burge} with 
$\br=\bfone \in P^{++}_{N,N}$ 
and $\bs \in P^{++}_{N,N+n}$,
\begin{align}
\mathcal{Z}_{\bsig;\bell}^{\bs; \bb,\bb^{\prime}}
\left(  
\ba, \bm, \bm^{\prime}; \fq 
\right) 
=
\sum_{\bY^{\bsig}\in \mathcal{C}^{\bs}_{\bsig;\bell}}
\frac{Z_{\mathrm{bif}}\left( \bm, {\boldsymbol \emptyset}^{\bb}; \ba, \bY^{\bsig} \right) 
Z_{\mathrm{bif}}\left( \ba, \bY^{\bsig}; -\bm^{\prime}, 
{\boldsymbol \emptyset}^{\bb^{\prime}} \right) }
{Z_{\mathrm{vec}}\left( \ba, \bY^{\bsig}\right) }\,
\fq^{\,\frac{1}{n} \left|\bY^{\bsig}\right|},
\label{inst_red_pf}
\end{align}
where $\sum_{I=1}^N a_I=0$ is imposed.
The Coulomb parameters $\ba=(a_1,\ldots,a_N)$, 
and the mass parameters $\bm=(m_1,\ldots, m_N)$, 
$\bm^{\prime}=(m_1^{\prime},\ldots, m_N^{\prime})$,  
are related to the internal momenta $\bmu^{v}$, and 
the external momenta $\bmu_{r=1,2,3,4}$, of a 4-point conformal block 
in a $\mathcal{W}^{\, para}_{N, n}$ CFT, 
by the relations \eqref{coulomb_rel} and \eqref{mass_rel_base}, respectively. 
The gauge theory in the rational $\Omega$-background \eqref{omega_pp} for $p=N$, 
\begin{align}
\frac{\epsilon_1}{\epsilon_2}=-1-\frac{n}{N},
\label{omega_trivial}
\end{align}
is expected to describe a minimal model CFT whose 
momenta take values in the degenerate momenta \eqref{deg_momenta} for $r_I=1$, 
\begin{align}
\begin{split}
2\, \bmu^{v} &= - \sum_{I=1}^{N-1} 
\ll s_{I}-1\rr  \epsilon_2\, \overline{\Lambda}_I,
\\
\mathop{\Longrightarrow}\limits^{(\ref{coulomb_rel})} \quad
a_I&=a_I^{\bs}:= -\sum_{J=1}^{N-1} \inner{\overline{\Lambda}_J}{\be_I} 
\ll s_J-1-\frac{n}{N}\rr \epsilon_2
\\
&=
-\sum_{J=I}^{N-1} \ll s_J -1-\frac{n}{N}\rr \epsilon_2
+\frac{1}{N}\sum_{J=1}^{N-1} J \ll s_J-1-\frac{n}{N}\rr \epsilon_2,
\label{deg_momenta_v_triv}
\end{split}
\end{align}
parametrized by 
$\bs=[s_0,s_1,\ldots,s_{N-1}] \in P^{++}_{N,N+n}$, and
\begin{align}
\begin{split}
2\, \bmu_1 &= - \sum_{I=1}^{N-1} 
\ll s_{1,I}-1\rr  \epsilon_2\, \overline{\Lambda}_I,\quad
2\, \bmu_2 = 
- \ll s_{2,N-1}-1\rr  \epsilon_2\, \overline{\Lambda}_{N-1},
\\
2\, \bmu_4 &= - \sum_{I=1}^{N-1} 
\ll s_{4,I}-1\rr  \epsilon_2\, \overline{\Lambda}_I,\quad
2\, \bmu_3 = 
- \ll s_{3,1}-1\rr  \epsilon_2\, \overline{\Lambda}_{1},
\\
\quad \mathop{\Longrightarrow}\limits^{(\ref{mass_rel_base})} \quad
m_I&=m_I^{\bs_1,\bs_2}
:=
-\ll I-\frac{N+1}{2}\rr \frac{n}{N}\, \epsilon_2
\\
&\hspace{2.5em}+
\frac{1}{N}\ll \sum_{J=1}^{I-1}J\, \ll s_{1,J}-1\rr 
-
\sum_{J=I}^{N-1}\ll N-J\rr \ll s_{1,J}-1\rr  - \ll s_{2,N-1}-1\rr \rr \epsilon_2,
\\
m_I^{\prime}&=m_I^{\prime\, \bs_3, \bs_4}
:=
\ll I-\frac{N+1}{2}\rr \frac{n}{N}\, \epsilon_2
\\
&\hspace{2.5em}+
\frac{1}{N}\ll -\sum_{J=1}^{I-1}J\, \ll s_{4,J}-1\rr +
\sum_{J=I}^{N-1}\ll N-J\rr \ll s_{4,J}-1\rr  - 
\ll s_{3,1}-1\rr \rr \epsilon_2,
\label{deg_momenta_triv}
\end{split}
\end{align}
parametrized by 
$\bs_1=[s_{1,0},s_{1,1},\ldots,s_{1,N-1}] \in P^{++}_{N,N+n}$,
$\bs_4=[s_{4,0},s_{4,1},\ldots,s_{4,N-1}] \in P^{++}_{N,N+n}$, and
\begin{align}
\begin{split}
\bs_2&=[s_{2,0},s_{2,1},\ldots,s_{2,N-1}]=[s_{2,0},1,\ldots,1,s_{2,N-1}]
\in P^{++}_{N,N+n},
\\
\bs_3&=[s_{3,0},s_{3,1},\ldots,s_{3,N-1}]=[s_{3,0},s_{3,1},1,\ldots,1]
\in P^{++}_{N,N+n},
\label{deg_mom_c}
\end{split}
\end{align}
following $\bmu_2 \propto \overline{\Lambda}_{N-1}$ and 
$\bmu_3 \propto \overline{\Lambda}_{1}$.

\begin{remark}[Fixing $\bs_1, \bs_4$]\label{rem:fix_s}
By \eqref{Eq:sigma2s}, the (Coulomb) parameters in $\bs$ are 
determined as
$s_I=\sigma_I-\sigma_{I+1}+1$, from the ordered charges 
$\sigma_{1}\ge \ldots \ge\sigma_{N}$. 
Similarly, we fix the (mass) parameters in $\bs_1$ and $\bs_4$. 
Taking a shift by the central $U(1)$ factor in the $U(N)$ flavor symmetry,
from \eqref{mass_rel_base}, into account, 
one obtains the ${\IZ}_n$ boundary charge conditions
\begin{align}
s_{1,I}-1 \equiv b_I - b_{I+1} \quad (\mathrm{mod}\ n),
\quad 
s_{4,I}-1 \equiv b_I^{\prime} - b_{I+1}^{\prime} \quad (\mathrm{mod}\ n),
\quad I=1,\ldots,N-1.
\label{boundary_charge_condition}
\end{align}
We can then determine the independent parameters in $\bs_1$ and $\bs_{4}$ as
\begin{align}
s_{1,I} =b_I - b_{I+1} + 1, \quad
s_{4,I} =b_I^{\prime} - b_{I+1}^{\prime} + 1, 
\quad I=1,\ldots,N-1.
\label{boundary_fix}
\end{align}
The remaining independent parameters $s_{2,N-1}$ and $s_{3,1}$ in \eqref{deg_mom_c} are determined 
in Remark \ref{rem:fix_s_c}.
\end{remark}

By subtracting the $U(1)$ factor \eqref{u1_factor},
as in the case of the $\ft$-refined Burge-reduced generating function 
\eqref{inst_red_ch}, we define a Burge-reduced 
instanton partition function labelled by 
$\bN=[N_0,\ldots,N_{n-1}] \in P^{+}_{n,N}$, 
$\bell=(\ell_1,\ldots,\ell_{n-1}) \in {\IZ}^{n-1}$, and 
${\IZ}_n$ boundary charges $\bb=(b_1,\ldots,b_N)$ and 
$\bb^{\prime}=(b_1^{\prime},\ldots,b_N^{\prime})$.

\begin{defi}\label{def:norm_red_inst}
The $SU(N)$ Burge-reduced instanton partition function is defined by
\begin{align}
\widehat{\mathcal{Z}}_{\bN;\bell}^{\bb,\bb^{\prime}}( \fq ) =
Z_{\mathcal{H}}\left( \bm^{\bs_1,\bs_2}, \bm^{\prime\, \bs_3,\bs_4}; \fq \right)^{-1}\times
\mathcal{Z}_{\bsig;\bell}^{\bs; \bb,\bb^{\prime}}
\left(  \ba^{\bs}, \bm^{\bs_1,\bs_2}, \bm^{\prime\, \bs_3,\bs_4}; \fq \right).
\label{norm_red_inst_pf}
\end{align}
Here
the Coulomb parameters $\ba^{\bs}=(a_1^{\bs},\ldots,a_N^{\bs})$ 
are given by \eqref{deg_momenta_v_triv} with 
$s_I=\sigma_I-\sigma_{I+1}+1$ in \eqref{Eq:sigma2s}, and 
the mass parameters 
$\bm^{\bs_1,\bs_2}=(m_1^{\bs_1,\bs_2},\ldots, m_N^{\bs_1,\bs_2})$ and 
$\bm^{\prime\, \bs_3,\bs_4}=(m_1^{\prime\, \bs_3,\bs_4},\ldots, m_N^{\prime\, \bs_3,\bs_4})$ are given by \eqref{deg_momenta_triv} with 
$s_{1,I}$, $s_{4,I}$ in \eqref{boundary_fix} 
and $s_{2,N-1}$, $s_{3,1}$ determined in Remark \ref{rem:fix_s_c}. 
%
\end{defi}

By Corollary \ref{prop:Ngch}, the set $\bN$, determined from 
the ${\IZ}_n$ charges $\bsig$, indicates level-$N$ dominant 
integral highest-weight in $\widehat{\mathfrak{sl}}(n)_N$ WZW model. 
We propose that, in the $\widehat{\mathfrak{sl}}(n)_N$ WZW 
4-point conformal blocks, the integrable representations 
of two of the four external primary fields are also 
determined from the ${\IZ}_n$ boundary charges 
$\bb=(b_1,\ldots,b_N)$ and $\bb^{\prime}=(b_1^{\prime},\ldots,b_N^{\prime})$ by
\begin{align}
\bB=\sum_{I=1}^{N} \Lambda_{b_I}
=[B_0,B_1,\ldots,B_{n-1}],
\qquad
\bB^{\prime}=\sum_{I=1}^{N} \Lambda_{b_I^{\prime}}
=[B_0^{\prime},B_1^{\prime},\ldots,B_{n-1}^{\prime}].
\end{align}
We now represent the Burge-reduced instanton partition function 
\eqref{norm_red_inst_pf}, graphically, as
\begin{align}
\widehat{\mathcal{Z}}_{\bN;\bell}^{\bb,\bb^{\prime}}( \fq ) 
\quad 
= \quad \inc{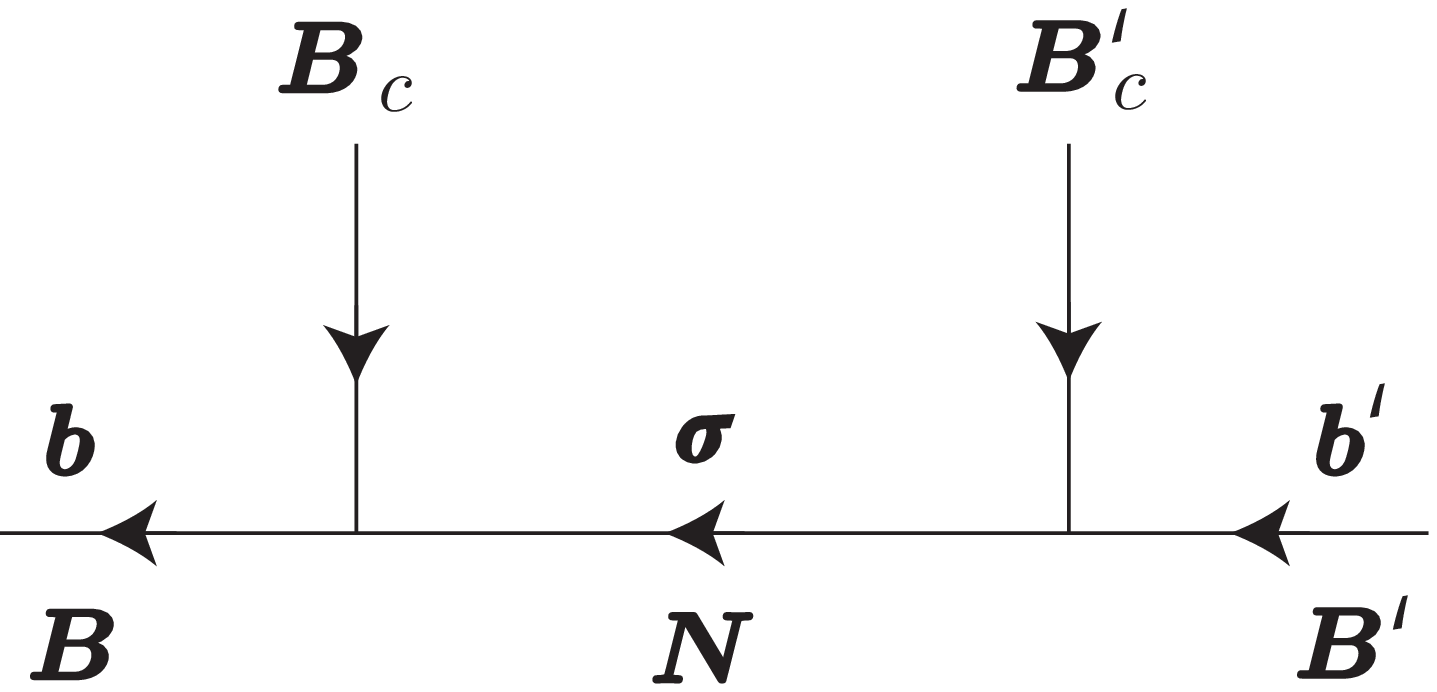}
\label{norm_red_inst_pf_fig}
\end{align}
We also represent \eqref{norm_red_inst_pf_fig} schematically by 
$\bB - \bB_c - (\bN) - \bB_c^{\prime} - \bB^{\prime}$. 
The representations $\bB_c$ and $\bB_c^{\prime}$ of 
the remaining two of the four external primary fields need to be taken 
so that they respect the fusion rules, 
which apply from right to left in \eqref{norm_red_inst_pf_fig}, 
of the $\widehat{\mathfrak{sl}}(n)_N$ WZW model when 
$\bN$, $\bB$ and $\bB^{\prime}$ are fixed 
(see \textit{e.g.} Chapter 16 of \cite{DiFrancesco:1997nk}). 
Then, the choice of the integers $\bell$ on the left hand side of 
\eqref{norm_red_inst_pf_fig}, which indicate 
the states of internal channel following Corollary \ref{prop:Ngch}, 
is also restricted by the fusion rules of 
$\bB^{\prime}$ and $\bB_c^{\prime}$.

\begin{remark}[Fixing the remaining parameters $s_{2,N-1}, s_{3,1}$]
\label{rem:fix_s_c}
In Remark \ref{rem:fix_s}, the parameters in 
$\bs_{1}$ and $\bs_{4}$ were fixed using the ${\IZ}_n$ boundary 
charge conditions. We now fix the remaining parameters 
$s_{2,N-1}, s_{3,1}$ in \eqref{deg_mom_c} using the fusion rules. 
Let $\bb_c=(b_{c,1},\ldots,b_{c,N})$ and 
$\bb_c^{\prime}=(b_{c,1}^{\prime},\ldots,b_{c,N}^{\prime})$ be 
boundary charges associated with $\bB_c$ and $\bB_c^{\prime}$, 
respectively.%
\footnote{\,
We will not assume the ordering of the boundary charges $\bb_c$ and $\bb_c^{\prime}$.}
We propose that they satisfy the same type of 
boundary charge conditions with \eqref{boundary_charge_condition} as 
$s_{2,I}-1 \equiv b_{c,I+1} - b_{c,I}$ ($\mathrm{mod}$ $n$) and 
$s_{3,I}-1 \equiv b_{c,I}^{\prime} - b_{c,I+1}^{\prime}$ ($\mathrm{mod}$ $n$) 
for the parameters in \eqref{deg_mom_c}. 
As a result, these boundary charges are   
\begin{align}
\begin{split}
\bb_c &\equiv (b_c,b_c,\ldots,b_c,b_c+s_{2,N-1}-1)
\quad (\mathrm{mod}\ n),
\\
\bb_c^{\prime} &\equiv (b_c^{\prime}+s_{3,1}-1,b_c^{\prime},b_c^{\prime},\ldots,b_c^{\prime})
\quad (\mathrm{mod}\ n),
\label{boundary_fix_c}
\end{split}
\end{align}
where $b_c, b_c^{\prime} \in \{0,1,\ldots,n-1\}$, 
and $s_{2,N-1}$, $s_{3,1}$ should 
be determined by the fusion rules. 
For definiteness, we restrict $s_{2,N-1}, s_{3,1} \in \{1,\ldots,n\}$, and 
if $N=2$ we take $b_c+s_{2,1} \le n$, $b_c^{\prime}+s_{3,1} \le n$ so that 
the boundary charges are $\bb_c = (b_c,b_c+s_{2,1}-1)$ and 
$\bb_c^{\prime} = (b_c^{\prime}+s_{3,1}-1,b_c^{\prime})$. 
\end{remark}

\subsection{Conjectures}

We propose the following conjectures on the relation between 
the $SU(N)$ Burge-reduced instanton 
partition functions \eqref{norm_red_inst_pf} on ${\IC}^2/{\IZ}_n$ 
and the $\widehat{\mathfrak{sl}}(n)_N$ WZW conformal blocks.
To describe our conjectures, we represent 
$\Lambda=[d_0,d_1,\ldots,d_{n-1}]\in P^{+}_{n,N}$ 
as a Young diagram 
by a partition $\lambda=\partit(\Lambda)$ using \eqref{Eq:Wt2Par}.

\begin{conj}[$\emptyset - \emptyset - (\emptyset) - \emptyset - \emptyset$]
\label{conj:ex1}
The $\widehat{\mathfrak{sl}}(n)_N$ WZW 2-point conformal block of the type 
$$
\left<\emptyset(1)\,  \emptyset(\fq) \right>_{\IP^1}^{\widehat{\mathfrak{sl}}(n)_N}
$$
agrees with the following Burge-reduced instanton partition function
\begin{align}
\label{Ex_1_fig}
\widehat{\mathcal{Z}}_{[N,0,\ldots,0];\mathbf{0}}^{\mathbf{0},\mathbf{0}}( \fq ) 
\quad = \quad 
\inc{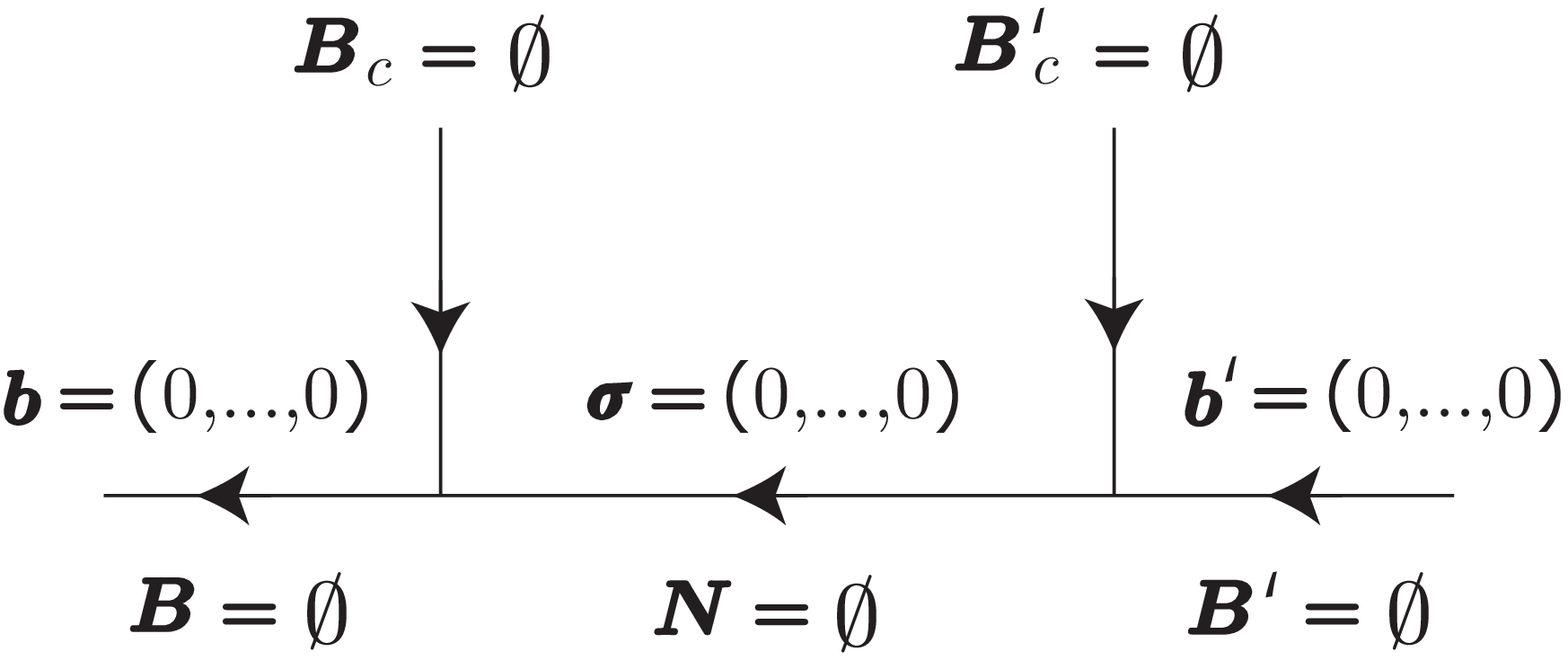}
\quad = \quad 
\left(  1-\fq \right) ^{-2h_{\emptyset}}=1.
\end{align}
Here $\bs=\bs_1=\bs_2=\bs_3=\bs_4=[n+1,1,\ldots,1]$ 
are fixed by \eqref{Eq:sigma2s}, \eqref{boundary_fix} and \eqref{boundary_fix_c}, 
and $h_{\emptyset}=0$ is the conformal dimension for the representation $\emptyset=[N,0,\ldots,0]$.
\end{conj}

\begin{conj}[$\emptyset - {[}N-1,0,\ldots,0,1{]} - (\tableau{1}) - \tableau{1} - \emptyset$]
\label{conj:ex2}
The $\widehat{\mathfrak{sl}}(n)_N$ WZW 2-point conformal block of the type 
$$
\left<\overline{\tableau{1}}(1)\,  \tableau{1}(\fq) \right>_{\IP^1}^{\widehat{\mathfrak{sl}}(n)_N}
$$
agrees with the following Burge-reduced instanton partition function
\begin{align}
\label{Ex_2_fig}
\widehat{\mathcal{Z}}_{[N-1,1,0\ldots,0];\mathbf{0}}^{\mathbf{0},\mathbf{0}}(\fq) 
\quad = \quad 
\inc{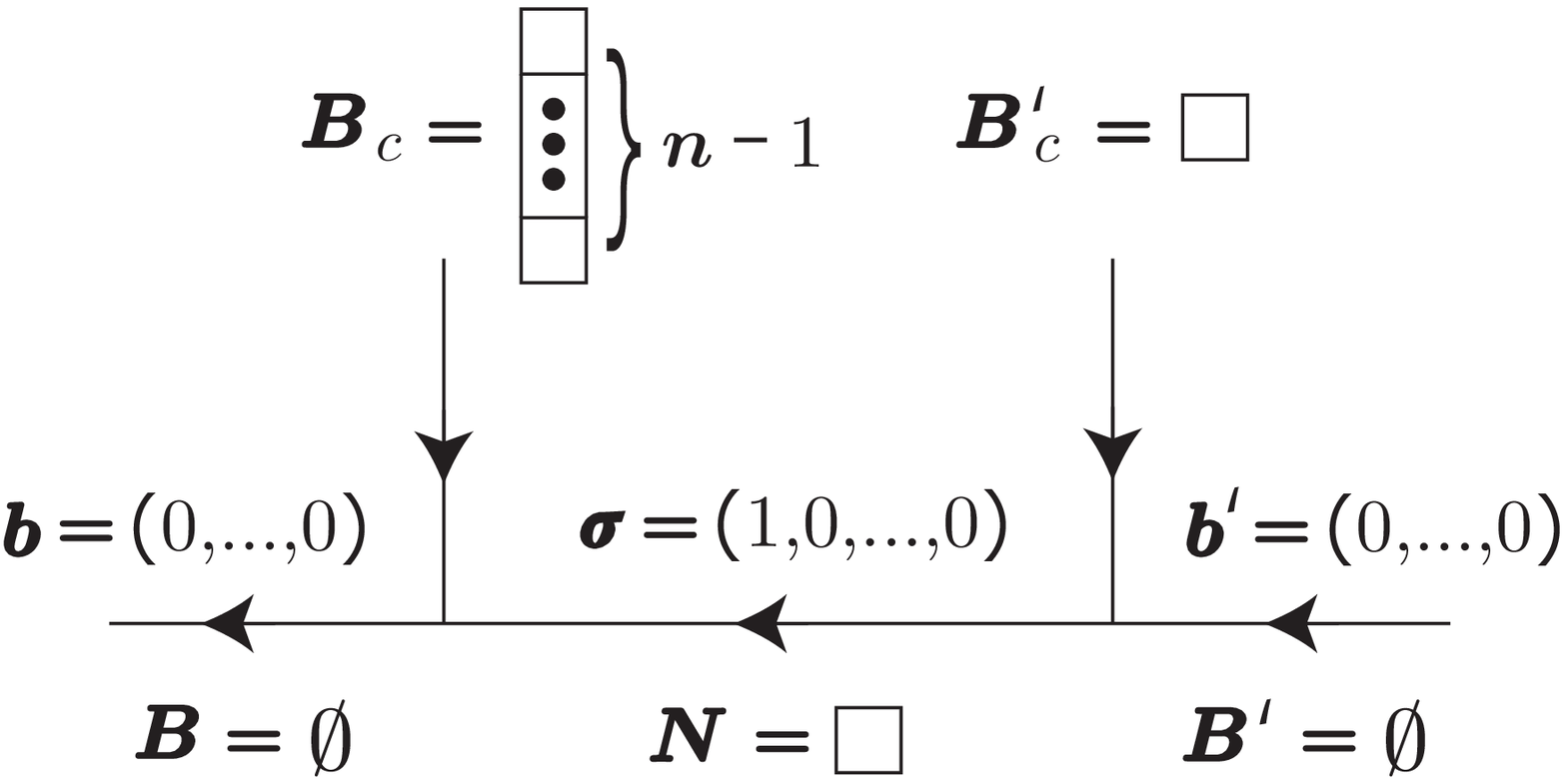}
\quad = \quad 
\left(  1-\fq \right) ^{-2h_{\tableau{1}}}.
\end{align}
Here $\bs=\bs_3=[n,2,1,\ldots,1]$,
$\bs_1=\bs_4=[n+1,1,\ldots,1]$ and $s_{2,N-1}=n$ are fixed by 
\eqref{Eq:sigma2s}, \eqref{boundary_fix} and \eqref{boundary_fix_c}, 
and $h_{\tableau{1}}=\frac{n^2-1}{2n(n+N)}$ 
is the conformal dimension for the representation 
$\tableau{1}=[N-1,1,0\ldots,0]$.
\end{conj}

\begin{conj}[$\tableau{1} - \tableau{1} - 
(\emptyset\ \textrm{or}\ {[}N-2,1,0,\ldots,0,1{]}) - \tableau{1} - {[}N-1,0,\ldots,0,1{]}$]
\label{conj:ex3}
The $\widehat{\mathfrak{sl}}(n)_N$ WZW 4-point conformal blocks of the type 
$$
\left<\overline{\tableau{1}}(\infty)\,  \tableau{1}(1)\,  \tableau{1}(\fq)\, 
\overline{\tableau{1}}(0) \right>_{\IP^1}^{\widehat{\mathfrak{sl}}(n)_N},
$$
which are    \eqref{wzw_bc_kz_trans} in Appendix \ref{app:wzw_4pt}, 
agree with, up to certain overall factors, 
the following Burge-reduced instanton partition functions,%
\footnote{\,
(\ref{Ex_3_fig}) and (\ref{Ex_4_fig}) correspond to, respectively, 
the 4-point WZW conformal blocks 
$\widehat{\mathcal{F}}_{i=1,2}^{(0)}(\fq)$ and 
$\widehat{\mathcal{F}}_{i=2,1}^{(1)}(\fq)$ 
in (\ref{wzw_bc_kz_trans}).}
\begin{align}
\label{Ex_3_fig}
&
\widehat{\mathcal{Z}}_{[N,0,\ldots,0];
\bell}^{(1,0,\ldots,0),(n-1,0,\ldots,0)}(\fq) 
\quad = \quad 
\inc{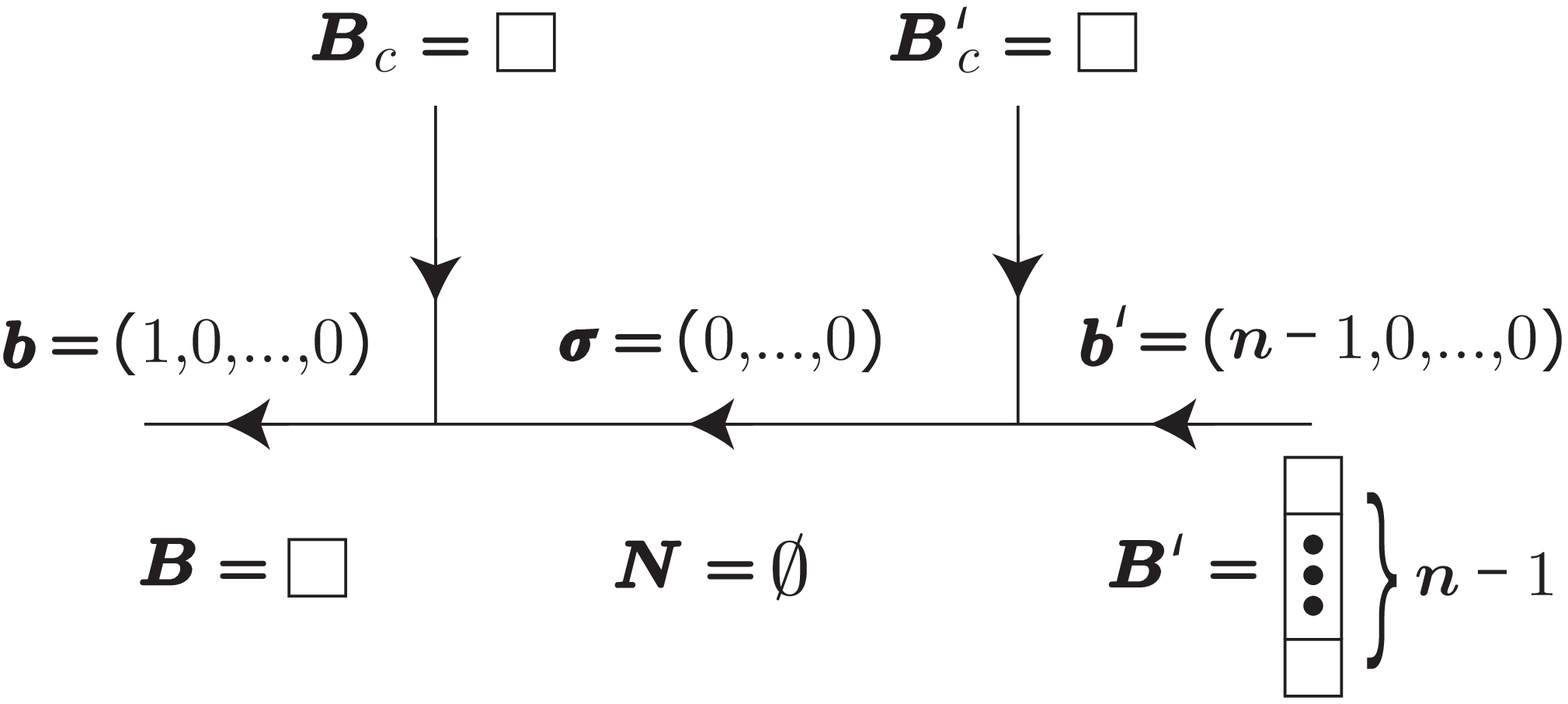}
\\
&=
\begin{cases}
\left( 1-\fq \right) ^{2h_{\tableau{1}}-\frac{n+1}{n+N}}
{_{2}F_{1}}\left( -\frac{1}{n+N},\frac{N-1}{n+N};\frac{N}{n+N};\fq\right) ,
\quad& \textrm{for}\ \  \bell=\mathbf{0},
\\
\frac{1}{N}\, \fq^{\frac{1}{n}}
\left( 1-\fq \right) ^{2h_{\tableau{1}}-\frac{n+1}{n+N}}
{_{2}F_{1}}\left( \frac{N-1}{n+N},1-\frac{1}{n+N};1+\frac{N}{n+N};\fq\right) ,
\quad& \textrm{for}\ \  \bell=(-1,\dots,-1),
\end{cases}
\nonumber
\end{align}
and
\begin{align}
\label{Ex_4_fig}
&
\widehat{\mathcal{Z}}_{[N-2,1,0,\ldots,0,1];
\bell}^{(1,0,\ldots,0),(n-1,0,\ldots,0)}(\fq) 
\quad = \quad 
\inc{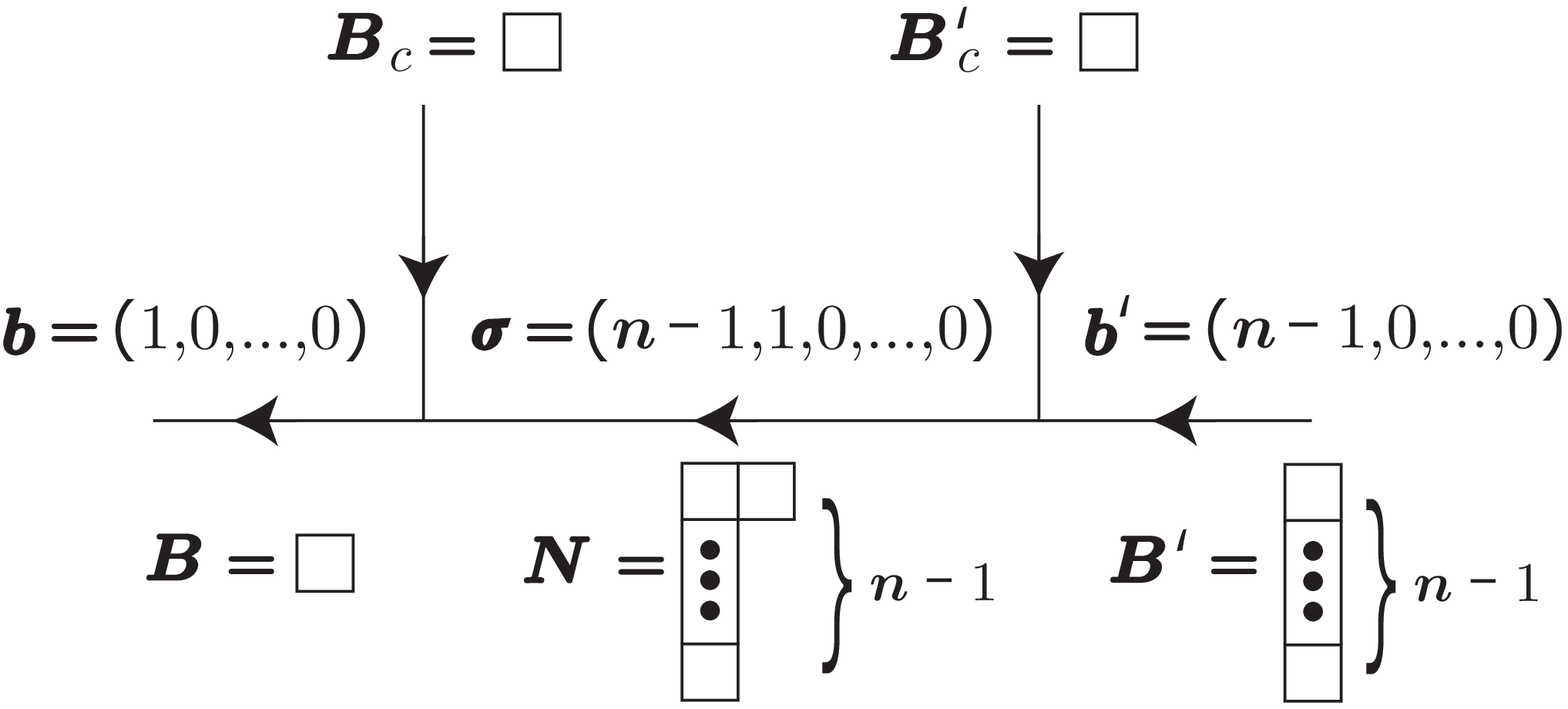}
\\
&=
\begin{cases}
\left( 1-\fq \right) ^{2h_{\tableau{1}}-\frac{n+1}{n+N}}
{_{2}F_{1}}\left( -\frac{1}{n+N},\frac{n-1}{n+N};\frac{n}{n+N};\fq\right) ,
\quad& \textrm{for}\ \  \bell=\mathbf{0},
\\
\frac{1}{n}\, \fq^{1-\frac{1}{n}} 
\left( 1-\fq \right) ^{2h_{\tableau{1}}-\frac{n+1}{n+N}}
{_{2}F_{1}}\left( \frac{n-1}{n+N},1-\frac{1}{n+N};1+\frac{n}{n+N};\fq\right) ,
\quad& \textrm{for}\ \  \bell=(1,\dots,1).
\end{cases}
\nonumber
\end{align}
Here, by \eqref{Eq:sigma2s}, \eqref{boundary_fix} and \eqref{boundary_fix_c}, 
for \eqref{Ex_3_fig}  
$\bs=[n+1,1,\ldots,1]$, $\bs_1=\bs_3=[n,2,1,\ldots,1]$,
$\bs_4=[2,n,1,\ldots,1]$ and $s_{2,N-1}=2$ are fixed, 
and for \eqref{Ex_4_fig} $\bs=[2,n-1,2,1,\ldots,1]$, $\bs_1=\bs_3=[n,2,1,\ldots,1]$, $\bs_4=[2,n,1,\ldots,1]$ 
and $s_{2,N-1}=2$ are fixed, 
where when $N=2$, $[2,n-1,2,1,\ldots,1]$ means $[3,n-1]$. 
The integers $\bell=\bdel \bk$ are taken 
so that the corresponding modules on the CFT side, 
following Corollary \ref{prop:Ngch}, are in the fundamental chamber under 
the action of affine Weyl group of $\widehat{\mathfrak{sl}}(n)$, and 
the second ones in \eqref{Ex_3_fig} and \eqref{Ex_4_fig} respect 
the fusion rules by 
\begin{align}
\begin{split}
\bN&=
\left[N, 0,\ldots,0\right]=\emptyset
\quad\mathop{\longrightarrow}\limits^{\bdel \bk=(-1,\dots,-1)}\quad
\bfc=\left[N-2,1,0,\ldots,0,1\right],
\\
\bN&=
\left[N-2,1,0,\ldots,0,1\right]
\quad\mathop{\longrightarrow}\limits^{\bdel \bk=(1,\dots,1)}\quad
\bfc=\left[N, 0,\ldots,0\right]=\emptyset,
\label{fusion_check}
\end{split}
\end{align}
where
$\bfc=[\ffc_0,\ffc_1,\ldots,\ffc_{n-1}]$ are defined by 
the Chern classes \eqref{fc_label}. 
When $n=2$, $\left[N-2,1,0,\ldots,0,1\right]$ means 
$\left[N-2, 2\right]=\tableau{2}$ and then $\bsig=(1,1,0,\ldots,0)$.
\end{conj}

\section{Examples of $SU(N)$ Burge-reduced instanton counting on 
${\IC}^2/{\IZ}_n$}\label{sec:examples}

\textit{\noindent
We illustrate the statement of Corollary \ref{prop:Ngch} 
and check Conjectures \ref{conj:ex1}, \ref{conj:ex2} and \ref{conj:ex3} for 
$(N,n)=(2,2), (2,3)$ and $(3,2)$. In particular we demonstrate how 
one can extract their 
$\widehat{\mathfrak{sl}}(n)_N$ WZW conformal blocks from 
the Burge-reduced instanton partition functions.}%
\footnote{\,
The computations in this section heavily rely on Mathematica. 
We have also checked Conjectures \ref{conj:ex1}, \ref{conj:ex2} 
and \ref{conj:ex3} for $(N,n)=(2,4)$ up to $O(\fq^{5})$.
}

\subsection{$(N,n)=(2,2)$ and $\widehat{\mathfrak{sl}}(2)_2$ WZW model}

For  $(N,n)=(2,2)$, there are three highest-weight representations
\begin{align}
\emptyset=[2,0],\quad
\tableau{1}=[1,1],\quad
\tableau{2}=[0,2],
\end{align}
with conformal dimensions 
\begin{align}
h_{[k_0,k_1]}=\frac{k_1 \, (k_1+2)}{16}:\quad
h_{\emptyset}=0,\quad
h_{\tableau{1}}=\frac{3}{16},\quad
h_{\tableau{2}}=\frac{1}{2}.
\end{align}

\subsubsection{Burge-reduced generating functions of coloured Young diagrams}

The $\ft$-refined Burge-reduced generating functions \eqref{inst_red_ch} for $(N,n)=(2,2)$ are obtained as
\begin{align}
\begin{split}
\widehat{X}_{[2,0]}^{\mathrm{red}}(\fq, \ft)&=
\left( \fq;\fq \right)_{\infty}\, 
\sum_{\ell \in {\IZ}}
{X}_{(0,0);(-\ell)}^{[3,1]}(\fq)\, \ft^{2\ell} 
=
X^{[2,0]}_{[2,0]}(\fq)\, f_{0}(\fq, \hat{\ft} )
+
X^{[2,0]}_{[0,2]}(\fq)\, f_{1}(\fq, \hat{\ft} ),
\\
\widehat{X}_{[0,2]}^{\mathrm{red}}(\fq, \ft)&=
\left( \fq;\fq \right)_{\infty}\, 
\sum_{\ell \in {\IZ}}
X_{(1,1);(-\ell)}^{[3,1]}(\fq)\, \ft^{2\ell+2} 
=
X^{[0,2]}_{[0,2]}(\fq)\, f_{1}(\fq, \hat{\ft} )
+
X^{[0,2]}_{[2,0]}(\fq)\, f_{0}(\fq, \hat{\ft} ),
\\
\widehat{X}_{[1,1]}^{\mathrm{red}}(\fq, \ft)&=
\left( \fq;\fq \right)_{\infty}\, 
\sum_{\ell \in {\IZ}}
{X}_{(1,0);(-\ell)}^{[2,2]}(\fq)\, \ft^{2\ell+1}=
X^{[1,1]}_{[1,1]}(\fq)\,
g(\fq, \hat{\ft} ),
\label{N2n2_character}
\end{split}
\end{align}
where $\hat{\ft} =\fq^{-\frac14}\, \ft$, 
\begin{align}
\begin{split}
X^{[2,0]}_{[2,0]}(\fq)&=
1 + \fq + 3 \fq^2 + 5 \fq^3 + 10 \fq^4 + 16 \fq^5 + 28 \fq^6 + 43 \fq^7
+ 70 \fq^8 + 105 \fq^9 + 161 \fq^{10} + \cdots \,  \, ,
\\
X^{[2,0]}_{[0,2]}(\fq)&=
\fq^{\frac12} + 
2 \fq^{\frac32} + 4 \fq^{\frac52}+7 \fq^{\frac72}+13 \fq^{\frac92}+21 
\fq^{\frac{11}{2}}+35 \fq^{\frac{13}{2}}+55 \fq^{\frac{15}{2}}+ 86 \fq^{\frac{17}{2}}+130 
\fq^{\frac{19}{2}}+\cdots \,  \, ,
\\
X^{[1,1]}_{[1,1]}(\fq)&=
1+2 \fq+4 \fq^2+8 \fq^3+14 \fq^4+24 \fq^5+40 \fq^6+64 \fq^7+100 \fq^8 +\cdots\,  \, ,
\\
X^{[0,2]}_{[0,2]}(\fq)&=X^{[2,0]}_{[2,0]}(\fq), \quad
X^{[0,2]}_{[2,0]}(\fq)=X^{[2,0]}_{[0,2]}(\fq),
\label{string_inst_22}
\end{split}
\end{align}
and
\begin{align}
f_{\sigma}(\fq, \hat{\ft} )=
\sum_{j \in 4\, {\IZ}+2\, \sigma}
\fq^{\frac18\, j^2}\, \hat{\ft}^{\, j}, \quad \sigma=0,1,
\quad
g(\fq, \hat{\ft} )=
\sum_{j \in 2\, {\IZ}+1}
\fq^{\frac18\, j^2 + \frac18}\, \hat{\ft}^{\, j}.
\end{align}
The Burge-reduced generating functions \eqref{N2n2_character} agree with 
the $\widehat{\mathfrak{sl}}(2)_2$ WZW characters computed by \eqref{wzw_ch_formula},
\begin{align}
\widehat{X}_{[2,0]}^{\mathrm{red}}(\fq, \ft)=
\chi_{[2,0]}^{\widehat{\mathfrak{sl}}(2)_2}( \fq, \hat{\ft} ) ,
\quad
\widehat{X}_{[0,2]}^{\mathrm{red}}(\fq, \ft)=
\chi_{[0,2]}^{\widehat{\mathfrak{sl}}(2)_2}( \fq, \hat{\ft} ) ,
\quad
\widehat{X}_{[1,1]}^{\mathrm{red}}(\fq, \ft)=
\fq^{\frac{1}{16}}\,
\chi_{[1,1]}^{\widehat{\mathfrak{sl}}(2)_2}( \fq, \hat{\ft} ) ,
\end{align}
and Corollary \ref{prop:Ngch} is confirmed. Up to an overall factor, 
the functions \eqref{string_inst_22} are 
the $\widehat{\mathfrak{sl}}(2)$ string functions of level-2 
in \cite{Kac:1984mq} and given by
(\textit{cf.} Corollary \ref{Cor:Chern=Char1}),
\begin{align}
X^{[2,0]}_{[2,0]}(\fq) + X^{[2,0]}_{[0,2]}(\fq)=
\frac{\left( -\fq^{\frac12};\fq \right)_{\infty}}
{\left( \fq;\fq \right)_{\infty}},
\quad
X^{[2,0]}_{[2,0]}(\fq) - X^{[2,0]}_{[0,2]}(\fq)=
\frac{\left( \fq^{\frac12};\fq^{\frac12} \right)_{\infty}}
{\left( \fq;\fq \right)_{\infty}^2},
\quad
X^{[1,1]}_{[1,1]}(\fq)=
\frac{\left( \fq^2;\fq^2 \right)_{\infty}}{\left( \fq;\fq \right)_{\infty}^2}.
\end{align}
Note that they are related to the NS sector and Ramond sector characters 
in \eqref{ns_r_ch} by 
$X^{[2,0]}_{[2,0]}(\fq) + X^{[2,0]}_{[0,2]}(\fq)=\chi_{\textrm{NS}}(\fq)$ and 
$X^{[1,1]}_{[1,1]}(\fq)=\chi_{\textrm{R}}(\fq)$.
Using the Jacobi triple product identity
\begin{align}
\sum_{\ell \in {\IZ}} x^{\ell}\, y^{\frac12 \ell(\ell-1)} = 
\left( -x ; y \right)_{\infty} \left( -\frac{y}{x} ; y \right)_{\infty}
\left( y ; y \right)_{\infty},
\label{JTPI}
\end{align}
one can easily obtain \eqref{inst_pr_red_ch} for the principal
characters of $\widehat{\mathfrak{sl}}(2)$,
\begin{align}
\begin{split}
&
\widehat{X}_{[2,0]}^{\mathrm{red}}(\fq, 1)=
\widehat{X}_{[0,2]}^{\mathrm{red}}(\fq, 1)=
\mathrm{Pr}\, \chi_{[2,0]}^{\widehat{\mathfrak{sl}}(2)}(\fq)=
\left( -\fq^{\frac12};\fq^{\frac12} \right)_{\infty}
\left( -\fq;\fq \right)_{\infty},
\\
& 
\widehat{X}_{[1,1]}^{\mathrm{red}}(\fq, 1)= 
\mathrm{Pr}\, \chi_{[1,1]}^{\widehat{\mathfrak{sl}}(2)}(\fq)=
\left( -\fq^{\frac12};\fq^{\frac12} \right)_{\infty}
\left( -\fq^{\frac12};\fq \right)_{\infty}.
\end{split}
\end{align}


\subsubsection{Burge-reduced instanton partition functions}

For $N=2$ with general $n$, the Burge-reduced instanton partition functions \eqref{norm_red_inst_pf} 
are determined 
by the parameters in $\bs=[s_0,s_1]\in P^{++}_{2,n+2}$ and 
$\bs_r=[s_{r,0},s_{r,1}]\in P^{++}_{2,n+2}$, $r=1,2,3,4$, 
fixed by the relations \eqref{Eq:sigma2s}, \eqref{boundary_fix}:
\begin{align}
s_1=\sigma_1-\sigma_2+1,\quad
s_{1,1} = b_1 - b_2 + 1,\quad 
s_{4,1} = b_1^{\prime} - b_2^{\prime} + 1,
\label{bc_N2_1}
\end{align}
and \eqref{boundary_fix_c} 
from the ordered charges $\sigma_{1}\ge \sigma_{2}$, 
$b_1 \ge b_2$ and $b_1^{\prime} \ge b_2^{\prime}$. 
The Coulomb parameters are then determined from the parameter $s:=s_1$ 
by \eqref{deg_momenta_v_triv}:
\begin{align}
a_1=-\frac{1}{2}\ll s -1-\frac{n}{2}\rr \epsilon_2,
\quad
a_2=\frac{1}{2}\ll s -1-\frac{n}{2}\rr \epsilon_2,
\end{align}
and the mass parameters $\bm=(m_1, m_2)$ and $\bm^{\prime}=(m_1^{\prime}, m_2^{\prime})$ 
are determined from the parameters in 
$\bs_1, \bs_2$ and $\bs_3, \bs_4$, respectively, 
by \eqref{deg_momenta_triv}.

Let us consider the case of $(N,n)=(2,2)$ with  
the rational $\Omega$-background $\epsilon_1/\epsilon_2=-2$ in 
\eqref{omega_trivial}.%
\footnote{\,
Examples \ref{ex1:inst}, \ref{ex2:inst} and \ref{ex3:inst} are 
confirmed up to $O(\fq^6)$.
}

\begin{exam}[$\emptyset - \emptyset - (\emptyset) - \emptyset - \emptyset$]
\label{ex1:inst}
Consider the Burge-reduced instanton partition function 
$\widehat{\mathcal{Z}}_{[2,0];(\ell)}^{(0,0),(0,0)}(\fq)$ and 
take $\ell=0$ in the fundamental chamber, 
which respects the fusion rules, as in Conjecture \ref{conj:ex1}. 
Here $\bs=\bs_1=\bs_2=\bs_3=\bs_4=[3,1]$ are fixed. 
Then, the Burge-reduced instanton partition function is obtained as
\begin{align}
\widehat{\mathcal{Z}}_{[2,0];(0)}^{(0,0),(0,0)}(\fq) =
\left(  1-\fq \right) ^{-2h_{\emptyset}}=1,
\qquad
h_{\emptyset}=0,
\label{2_2_NNNN}
\end{align}
and Conjecture \ref{conj:ex1} is confirmed.
\end{exam}

\begin{exam}[$\emptyset - \tableau{1} - (\tableau{1}) - \tableau{1} - \emptyset$]
\label{ex2:inst}
Consider the Burge-reduced instanton partition function 
$\widehat{\mathcal{Z}}_{[1,1];(\ell)}^{(0,0),(0,0)}(\fq)$ and 
take $\ell=0$ in the fundamental chamber as in Conjecture \ref{conj:ex2}. 
Here $\bs=\bs_2=\bs_3=[2,2]$ and $\bs_1=\bs_4=[3,1]$ are fixed. 
Then we see that the Burge-reduced instanton partition function is   
\begin{align}
\widehat{\mathcal{Z}}_{[1,1];(0)}^{(0,0),(0,0)}(\fq) =
\left(  1-\fq \right) ^{-2h_{\tableau{1}}}=
1+\frac{3 \fq}{8}+\frac{33 \fq^2}{128}+\frac{209 \fq^3}{1024}+
\frac{5643 \fq^4}{32768}+\frac{39501 \fq^{5}}{262144}+\cdots \,  \, ,
\label{2_2_NRRN}
\end{align}
where $h_{\tableau{1}}=3/16$, and Conjecture \ref{conj:ex2} is confirmed.
\end{exam}

\begin{exam}[$\tableau{1} - \tableau{1} - (\emptyset) - \tableau{1} - \tableau{1}$ and $\tableau{1} - \tableau{1} - (\tableau{2}) - \tableau{1} - \tableau{1}$]
\label{ex3:inst}
For Conjecture \ref{conj:ex3}, consider, first, 
the Burge-reduced instanton partition function $\widehat{\mathcal{Z}}_{[2,0];(\ell)}^{(1,0),(1,0)}(\fq)$,  
where $\bs=[3,1]$ and $\bs_1=\bs_2=\bs_3=\bs_4=[2,2]$ are fixed. 
Then we find that the Burge-reduced instanton partition functions for 
$\ell=0, -1$ in the fundamental chamber are
\begin{align}
\begin{split}
\widehat{\mathcal{Z}}_{[2,0];(0)}^{(1,0),(1,0)}(\fq) &=
\left(  1-\fq \right) ^{2h_{\tableau{1}}-\frac34} 
{_{2}F_{1}}\left( -\frac{1}{4},\frac{1}{4};\frac{1}{2};\fq\right) 
\\
&=
1+\frac{\fq}{4}+\frac{11 \fq^2}{64}+\frac{35 \fq^3}{256}+\frac{949 \fq^4}{8192}+\frac{3333 \fq^5}{32768}+\frac{47909 \fq^6}{524288}+\cdots \,  \, ,
\\
\widehat{\mathcal{Z}}_{[2,0];(-1)}^{(1,0),(1,0)}(\fq) &=
\frac{\fq^{\frac12}}{2} \left(  1-\fq \right) ^{2h_{\tableau{1}}-\frac34} 
{_{2}F_{1}}\left( \frac{1}{4},\frac{3}{4};\frac{3}{2};\fq\right) 
\\
&=
\frac{\fq^{\frac12}}{2}+\frac{\fq^{\frac32}}{4}+
\frac{23 \fq^{\frac52}}{128}+\frac{37 \fq^{\frac72}}{256}+
\frac{2013 \fq^{\frac92}}{16384}+\frac{3537 \fq^{\frac{11}{2}}}{32768}+\cdots \,  \, ,
\label{2_2_RRRR1}
\end{split}
\end{align}
where $h_{\tableau{1}}=3/16$, and the second one respects the fusion rules by 
\eqref{fusion_check}. Consider, next, the Burge-reduced instanton partition function 
$\widehat{\mathcal{Z}}_{[0,2];(\ell)}^{(1,0),(1,0)}(\fq)$, 
where $\bs=[3,1]$ and $\bs_1=\bs_2=\bs_3=\bs_4=[2,2]$ are fixed. 
Then we obtain the Burge-reduced instanton partition functions for $\ell=0, 1$ 
in the fundamental chamber as
\begin{align}
\begin{split}
\widehat{\mathcal{Z}}_{[0,2];(0)}^{(1,0),(1,0)}(\fq) &=
\left(  1-\fq \right) ^{2h_{\tableau{1}}-\frac34} 
{_{2}F_{1}}\left( -\frac{1}{4},\frac{1}{4};\frac{1}{2};\fq\right) 
=\widehat{\mathcal{Z}}_{[2,0];0}^{(1,0),(1,0)}(\fq) ,
\\
\widehat{\mathcal{Z}}_{[0,2];(1)}^{(1,0),(1,0)}(\fq) &=
\frac{\fq^{\frac12}}{2} \left(  1-\fq \right) ^{2h_{\tableau{1}}-\frac34} 
{_{2}F_{1}}\left( \frac{1}{4},\frac{3}{4};\frac{3}{2};\fq\right) 
=\widehat{\mathcal{Z}}_{[2,0];-1}^{(1,0),(1,0)}(\fq) ,
\label{2_2_RRRR2}
\end{split}
\end{align}
where the second one respects the fusion rules by \eqref{fusion_check}.
The above results \eqref{2_2_RRRR1} and \eqref{2_2_RRRR2}  
support Conjecture \ref{conj:ex3}. By
\begin{align}
{_{2}F_{1}}\left( -\frac{1}{4},\frac{1}{4};\frac{1}{2};\fq\right) =
\ll \frac{1+\sqrt{1-\fq}}{2}\rr ^{\frac12},
\quad
\frac{\fq^{\frac12}}{2}\, {_{2}F_{1}}\left( \frac{1}{4},\frac{3}{4};\frac{3}{2};\fq\right) =
\ll \frac{1-\sqrt{1-\fq}}{2}\rr ^{\frac12},
\end{align}
they are also consistent with the results in \cite{Belavin:2012aa} by 
Belavin and Mukhametzhanov.%
\footnote{\,
More precisely, in \cite{Belavin:2012aa}, the generic $\Omega$-background, without 
the Burge conditions, was discussed. 
Then the first one of (\ref{2_2_RRRR1}) and the second one of (\ref{2_2_RRRR2}), with 
$\mathfrak{c}=0$, were obtained as prefactors combined with the $\mathcal{N}=1$ 
super-Virasoro Ramond conformal blocks $H_{\pm}(\fq)$, $F_{\pm}(\fq)$, 
$\widetilde{H}_{\pm}(\fq)$ and $\widetilde{F}_{\pm}(\fq)$. 
What we found is that, when we impose the specific Burge conditions, 
the conformal blocks are trivialized as 
$H_{\pm}(\fq), F_{\pm}(\fq) \to 1$ and $\widetilde{H}_{\pm}(\fq), \widetilde{F}_{\pm}(\fq) \to 0$, 
and only the prefactors are obtained.\label{footnote:BM}}
\end{exam}

\subsection{$(N,n)=(2,3)$ and $\widehat{\mathfrak{sl}}(3)_2$ WZW model}

For $(N,n)=(2,3)$, there are six highest-weight representations
\begin{align}
\emptyset=[2,0,0],\quad
\tableau{1}=[1,1,0],\quad
\tableau{2}=[0,2,0],\quad
\tableau{1 1}=[1,0,1],\quad
\tableau{2 1}=[0,1,1],\quad
\tableau{2 2}=[0,0,2],
\end{align}
with conformal dimensions 
\begin{align}
\begin{split}
h_{[k_0,k_1,k_2]}&=\frac{k_1^2+k_2^2+k_1 k_2+3k_1+3k_2}{15}:
\\
h_{\emptyset}&=0,\quad
h_{\tableau{1}}=h_{\tableau{1 1}}=\frac{4}{15},\quad
h_{\tableau{2}}=h_{\tableau{2 2}}=\frac{2}{3},\quad
h_{\tableau{2 1}}=\frac{3}{5}.
\end{split}
\end{align}

\subsubsection{Burge-reduced generating functions of coloured Young diagrams}

The $\ft$-refined Burge-reduced generating functions \eqref{inst_red_ch} for $(N,n)=(2,3)$ are obtained as
\begin{align}
\begin{split}
\widehat{X}_{[2,0,0]}^{\mathrm{red}}(\fq, (\ft_1,\ft_2))&=
\left( \fq;\fq \right)_{\infty}\, 
\sum_{(\ell_1,\ell_2) \in {\IZ}^2}
{X}_{(0,0);(-\ell_1,-\ell_2)}^{[4,1]}(\fq)\, 
\ft_1^{ 2\ell_1-\ell_2 }\, 
\ft_2^{ -\ell_1+2\ell_2 }
\\
&=
X^{[2,0,0]}_{[2,0,0]}(\fq)\, f_{00}(\fq, \hat{\ft}_1 , \hat{\ft}_2) + 
X^{[2,0,0]}_{[0,1,1]}\, g_{00}(\fq, \hat{\ft}_1 , \hat{\ft}_2),
\\
\widehat{X}_{[0,2,0]}^{\mathrm{red}}(\fq, (\ft_1,\ft_2))&=
\left( \fq;\fq \right)_{\infty}\, 
\sum_{(\ell_1,\ell_2) \in {\IZ}^2}
{X}_{(1,1);(-\ell_1,-\ell_2)}^{[4,1]}(\fq)\, 
\ft_1^{ 2+2\ell_1-\ell_2 }\, 
\ft_2^{ -\ell_1+2\ell_2 }
\\
&=
X^{[0,2,0]}_{[0,2,0]}(\fq)\, f_{10}(\fq, \hat{\ft}_1 , \hat{\ft}_2) + 
X^{[0,2,0]}_{[1,0,1]}(\fq)\, g_{10}(\fq, \hat{\ft}_1 , \hat{\ft}_2),
\\
\widehat{X}_{[0,0,2]}^{\mathrm{red}}(\fq, (\ft_1,\ft_2))&=
\left( \fq;\fq \right)_{\infty}\, 
\sum_{(\ell_1,\ell_2) \in {\IZ}^2}
{X}_{(2,2);(-\ell_1,-\ell_2)}^{[4,1]}(\fq)\, 
\ft_1^{ 2\ell_1-\ell_2 }\, 
\ft_2^{ 2-\ell_1+2\ell_2 }
\\
&=
X^{[0,0,2]}_{[0,0,2]}(\fq)\, f_{01}(\fq, \hat{\ft}_1 , \hat{\ft}_2) + 
X^{[0,0,2]}_{[1,1,0]}(\fq)\, g_{01}(\fq, \hat{\ft}_1 , \hat{\ft}_2),
\\
\widehat{X}_{[1,1,0]}^{\mathrm{red}}(\fq, (\ft_1,\ft_2))&=
\left( \fq;\fq \right)_{\infty}\, 
\sum_{(\ell_1,\ell_2) \in {\IZ}^2}
{X}_{(1,0);(-\ell_1,-\ell_2)}^{[3,2]}(\fq)\, 
\ft_1^{ 1+2\ell_1-\ell_2 }\, 
\ft_2^{ -\ell_1+2\ell_2 }
\\
&=
X^{[1,1,0]}_{[1,1,0]}\, g_{01}(\fq, \hat{\ft}_1 , \hat{\ft}_2) + 
X^{[1,1,0]}_{[0,0,2]}(\fq)\, f_{01}(\fq, \hat{\ft}_1 , \hat{\ft}_2),
\\
\widehat{X}_{[0,1,1]}^{\mathrm{red}}(\fq, (\ft_1,\ft_2))&=
\left( \fq;\fq \right)_{\infty}\, 
\sum_{(\ell_1,\ell_2) \in {\IZ}^2}
{X}_{(2,1);(-\ell_1,-\ell_2)}^{[3,2]}(\fq)\, 
\ft_1^{ 1+2\ell_1-\ell_2 }\, 
\ft_2^{ 1-\ell_1+2\ell_2 }
\\
&=
X^{[0,1,1]}_{[0,1,1]}(\fq)\, g_{00}(\fq, \hat{\ft}_1 , \hat{\ft}_2) + 
X^{[0,1,1]}_{[2,0,0]}(\fq)\, f_{00}(\fq, \hat{\ft}_1 , \hat{\ft}_2),
\\
\widehat{X}_{[1,0,1]}^{\mathrm{red}}(\fq, (\ft_1,\ft_2))&=
\left( \fq;\fq \right)_{\infty}\, 
\sum_{(\ell_1,\ell_2) \in {\IZ}^2}
{X}_{(2,0);(-\ell_1,-\ell_2)}^{[2,3]}(\fq)\, 
\ft_1^{ 2\ell_1-\ell_2 }\, 
\ft_2^{ 1-\ell_1+2\ell_2 }
\\
&=
X^{[1,0,1]}_{[1,0,1]}(\fq)\, g_{10}(\fq, \hat{\ft}_1 , \hat{\ft}_2) + 
X^{[1,0,1]}_{[0,2,0]}(\fq)\, f_{10}(\fq, \hat{\ft}_1 , \hat{\ft}_2),
\label{N2n3_character}
\end{split}
\end{align}
where $\hat{\ft}_1 =\fq^{-\frac13}\, \ft_1$, $\hat{\ft}_2=\fq^{-\frac13}\, \ft_2$,
\begin{align}
\begin{split}
X^{[2,0,0]}_{[2,0,0]}(\fq)&=
1 + 2 \fq + 8 \fq^2 + 20 \fq^3 + 52 \fq^4 + 116 \fq^5 + 256 \fq^6 + 522 \fq^7 + \cdots \,  \, ,
\\
X^{[2,0,0]}_{[0,1,1]}(\fq)&=
\fq^{\frac13} + 4 \fq^{\frac43} + 12 \fq^{\frac73} + 32 \fq^{\frac{10}{3}} 
+ 77 \fq^{\frac{13}{3}} + 172 \fq^{\frac{16}{3}} + 365 \fq^{\frac{19}{3}} 
+ 740 \fq^{\frac{22}{3}} + \cdots \,  \, ,
\\
X^{[0,1,1]}_{[0,1,1]}(\fq)&=
1 + 4 \fq + 13 \fq^2 + 36 \fq^3 + 89 \fq^4 + 204 \fq^5 + 441 \fq^6 + 908 \fq^7 + \cdots \,  \, ,
\\
X^{[0,1,1]}_{[2,0,0]}(\fq) 
& =
2 \fq^{\frac23} + 7 \fq^{\frac53} + 22 \fq^{\frac83} + 56 \fq^{\frac{11}{3}} 
+ 136 \fq^{\frac{14}{3}} + 300 \fq^{\frac{17}{3}} + 636 \fq^{\frac{20}{3}} 
+ 1280 \fq^{\frac{23}{3}} + \cdots \,  \, ,
\\
X^{[0,2,0]}_{[0,2,0]}(\fq)&=X^{[0,0,2]}_{[0,0,2]}(\fq)=X^{[2,0,0]}_{[2,0,0]}(\fq),
\quad
X^{[0,2,0]}_{[1,0,1]}(\fq)=X^{[0,0,2]}_{[1,1,0]}(\fq)=X^{[2,0,0]}_{[0,1,1]}(\fq),
\\
X^{[1,1,0]}_{[1,1,0]}(\fq)&=X^{[1,0,1]}_{[1,0,1]}(\fq)=X^{[0,1,1]}_{[0,1,1]}(\fq),
\quad
X^{[1,1,0]}_{[0,0,2]}(\fq)=X^{[1,0,1]}_{[0,2,0]}(\fq)=X^{[0,1,1]}_{[2,0,0]}(\fq),
\label{string_fn_N2n3}
\end{split}
\end{align}
and $f_{00}$, $f_{10}$, $f_{01}$, $g_{00}$, $g_{10}$, $g_{01}$ are
\begin{align}
\begin{split}
f_{\sigma_1\sigma_2}(\fq, \hat{\ft}_1 , \hat{\ft}_2)&=
\mathop{\sum_{(j_1, j_2) \in \left\{2\, {\IZ}\right\}^2}}
\limits_{j_1-j_2 \in 6\, {\IZ}+2 \left(\sigma_1-\sigma_2\right)}
\fq^{\frac16 \ll j_1^2 + j_2^2 +j_1\, j_2\rr }\, 
\hat{\ft}_1^{\, j_1}\, \hat{\ft}_2^{\, j_2},
\\
g_{\sigma_1\sigma_2}(\fq, \hat{\ft}_1 , \hat{\ft}_2)&=
\mathop{\sum_{(j_1, j_2) \in \left\{2\, {\IZ}+1 \right\}^2}}
\limits_{j_1-j_2 \in 6\, {\IZ}+2 \left(\sigma_1-\sigma_2\right)}
\fq^{\frac16 \ll j_1^2 + j_2^2 +j_1\, j_2\rr  + \frac16}\, 
\hat{\ft}_1^{\, j_1}\, \hat{\ft}_2^{\, j_2}
\\
& +
\mathop{\sum_{(j_1, j_2) \in \left\{2\, {\IZ}\right\} \times 
\left\{2\, {\IZ}+1 \right\}}}
\limits_{j_1-j_2 \in 6\, {\IZ}+1+2 \left(\sigma_1-\sigma_2\right)}
\fq^{\frac16 \ll j_1^2 + j_2^2 +j_1\, j_2\rr  + \frac16} 
\ll \hat{\ft}_1^{\, j_1}\, \hat{\ft}_2^{\, j_2} + 
\hat{\ft}_1^{\, j_2}\, \hat{\ft}_2^{\, j_1}\rr.
\end{split}
\end{align}
The Burge-reduced generating functions \eqref{N2n3_character} agree with 
the $\widehat{\mathfrak{sl}}(3)_2$ WZW characters computed by \eqref{wzw_ch_formula},
\begin{align}
\begin{split}
&
\widehat{X}_{[2,0,0]}^{\mathrm{red}}(\fq, (\ft_1, \ft_2))=
\chi_{[2,0,0]}^{\widehat{\mathfrak{sl}}(3)_2}( \fq, (\hat{\ft}_1 , \hat{\ft}_2)) ,
\quad
\widehat{X}_{[0,2,0]}^{\mathrm{red}}(\fq, (\ft_1, \ft_2))=
\chi_{[0,2,0]}^{\widehat{\mathfrak{sl}}(3)_2}( \fq, (\hat{\ft}_1 , \hat{\ft}_2)) ,
\\
&
\widehat{X}_{[0,0,2]}^{\mathrm{red}}(\fq, (\ft_1, \ft_2))=
\chi_{[0,0,2]}^{\widehat{\mathfrak{sl}}(3)_2}( \fq, (\hat{\ft}_1 , \hat{\ft}_2)) ,
\quad
\widehat{X}_{[1,1,0]}^{\mathrm{red}}(\fq, (\ft_1, \ft_2))=
\fq^{\frac{1}{15}}\,
\chi_{[1,1,0]}^{\widehat{\mathfrak{sl}}(3)_2}( \fq, (\hat{\ft}_1 , \hat{\ft}_2)) ,
\\
&
\widehat{X}_{[0,1,1]}^{\mathrm{red}}(\fq, (\ft_1, \ft_2))=
\fq^{\frac{1}{15}}\,
\chi_{[0,1,1]}^{\widehat{\mathfrak{sl}}(3)_2}( \fq, (\hat{\ft}_1 , \hat{\ft}_2)) ,
\quad
\widehat{X}_{[1,0,1]}^{\mathrm{red}}(\fq, (\ft_1, \ft_2))=
\fq^{\frac{1}{15}}\,
\chi_{[1,0,1]}^{\widehat{\mathfrak{sl}}(3)_2}( \fq, (\hat{\ft}_1 , \hat{\ft}_2)) ,
\end{split}
\end{align}
and Corollary \ref{prop:Ngch} is confirmed. 
Up to an overall factor, the functions \eqref{string_fn_N2n3} are 
the $\widehat{\mathfrak{sl}}(3)$ string functions of level-2 
in \cite{Kac:1984mq} and given by 
(\textit{cf.} Corollary \ref{Cor:Chern=Char1}),
\begin{align}
\begin{split}
&
X^{[2,0,0]}_{[2,0,0]}(\fq) - \fq^{\frac16}\, X^{[2,0,0]}_{[0,1,1]}(\fq)=
\frac{\left( \fq^{\frac12};\fq^{\frac12} \right)_{\infty}
\left( \fq;\fq^{\frac52} \right)_{\infty}
\left( \fq^{\frac32};\fq^{\frac52} \right)_{\infty}
\left( \fq^{\frac52};\fq^{\frac52} \right)_{\infty}}
{\left( \fq;\fq \right)_{\infty}^4},
\\
&
X^{[2,0,0]}_{[0,1,1]}(\fq)= \fq^{\frac13} \, 
\frac{\left( \fq^{2};\fq^{2} \right)_{\infty}
\left( \fq^2;\fq^{10} \right)_{\infty}
\left( \fq^{8};\fq^{10} \right)_{\infty}
\left( \fq^{10};\fq^{10} \right)_{\infty}}
{\left( \fq;\fq \right)_{\infty}^4},
\\
&
X^{[0,1,1]}_{[0,1,1]}(\fq)=
\frac{\left( \fq^{2};\fq^{2} \right)_{\infty}
\left( \fq^4;\fq^{10} \right)_{\infty}
\left( \fq^{6};\fq^{10} \right)_{\infty}
\left( \fq^{10};\fq^{10} \right)_{\infty}}
{\left( \fq;\fq \right)_{\infty}^4},
\\
&
\fq^{\frac16} \, X^{[0,1,1]}_{[0,1,1]}(\fq) - X^{[0,1,1]}_{[2,0,0]}(\fq)= \fq^{\frac16} \, 
\frac{\left( \fq^{\frac12};\fq^{\frac12} \right)_{\infty}
\left( \fq^{\frac12};\fq^{\frac52} \right)_{\infty}
\left( \fq^{2};\fq^{\frac52} \right)_{\infty}
\left( \fq^{\frac52};\fq^{\frac52} \right)_{\infty}}
{\left( \fq;\fq \right)_{\infty}^4}.
\end{split}
\end{align}
By taking $\ft_1=\ft_2=1$, the principal characters of $\widehat{\mathfrak{sl}}(3)$ 
are obtained as in \eqref{inst_pr_red_ch}:
\begin{align}
\begin{split}
\widehat{X}_{[2,0,0]}^{\mathrm{red}}(\fq, (1, 1))&=
\widehat{X}_{[0,2,0]}^{\mathrm{red}}(\fq, (1, 1))=
\widehat{X}_{[0,0,2]}^{\mathrm{red}}(\fq, (1, 1))=
\mathrm{Pr}\, \chi_{[2,0,0]}^{\widehat{\mathfrak{sl}}(3)}(\fq)
\\
&=
\frac{\left( \fq;\fq \right)_{\infty}}
{\left( \fq^{\frac13} ;\fq^{\frac13} \right)_{\infty}
\left( \fq^{\frac23} ;\fq^{\frac53} \right)_{\infty}
\left( \fq ;\fq^{\frac53} \right)_{\infty}},
\\
\widehat{X}_{[1,1,0]}^{\mathrm{red}}(\fq, (1, 1))&=
\widehat{X}_{[0,1,1]}^{\mathrm{red}}(\fq, (1, 1))=
\widehat{X}_{[1,0,1]}^{\mathrm{red}}(\fq, (1, 1))=
\mathrm{Pr}\, \chi_{[1,1,0]}^{\widehat{\mathfrak{sl}}(3)}(\fq)
\\
&=
\frac{\left( \fq;\fq \right)_{\infty}}
{\left( \fq^{\frac13} ;\fq^{\frac13} \right)_{\infty}
\left( \fq^{\frac13} ;\fq^{\frac53} \right)_{\infty}
\left( \fq^{\frac43} ;\fq^{\frac53} \right)_{\infty}}.
\label{pr_ch_sl32}
\end{split}
\end{align}

\subsubsection{Burge-reduced instanton partition functions}

For $(N,n)=(2,3)$, the rational $\Omega$-background \eqref{omega_trivial} 
yields $\epsilon_1/\epsilon_2=-5/2$. The parameters 
in $\bs=[s_0,s_1]\in P^{++}_{2,5}$ and 
$\bs_r=[s_{r,0},s_{r,1}]\in P^{++}_{2,5}$, $r=1,2,3,4$, which determine 
the Burge-reduced instanton partition functions, are fixed as in \eqref{bc_N2_1}.%
\footnote{\,
Examples \ref{ex4:inst}, \ref{ex5:inst} and \ref{ex6:inst} are 
confirmed up to $O(\fq^5)$.
}
 
\begin{exam}[$\emptyset - \emptyset - (\emptyset) - \emptyset - \emptyset$]
\label{ex4:inst}
Consider the Burge-reduced instanton partition function
$\widehat{\mathcal{Z}}_{[2,0,0];(\ell_1,\ell_2)}^{(0,0),(0,0)}( \fq )$ and 
take $(\ell_1,\ell_2)=(0,0)$ in the fundamental chamber as in Conjecture \ref{conj:ex1}. 
Here $\bs=\bs_1=\bs_2=\bs_3=\bs_4=[4,1]$ are fixed.
Then we see that 
the Burge-reduced instanton partition function is   
\begin{align}
\widehat{\mathcal{Z}}_{[2,0,0];(0,0)}^{(0,0),(0,0)}( \fq ) =
\left(  1-\fq \right) ^{-2h_{\emptyset}}=1,
\qquad
h_{\emptyset}=0,
\label{rZ1_2_3}
\end{align}
and Conjecture \ref{conj:ex1} is confirmed.
\end{exam}

\begin{exam}[$\emptyset - \tableau{1 1} - (\tableau{1}) - \tableau{1} - \emptyset$]
\label{ex5:inst}
Consider the Burge-reduced instanton partition function
$\widehat{\mathcal{Z}}_{[1,1,0];(\ell_1,\ell_2)}^{(0,0),(0,0)}( \fq )$
and take $(\ell_1,\ell_2)=(0,0)$ in the fundamental chamber as in 
Conjecture \ref{conj:ex2}. 
Here $\bs=\bs_3=[3,2]$, $\bs_1=\bs_4=[4,1]$ and $\bs_2=[2,3]$ are fixed. 
Then the Burge-reduced instanton partition function is   
\begin{align}
\widehat{\mathcal{Z}}_{[1,1,0];(0,0)}^{(0,0),(0,0)}( \fq ) =
\left(  1-\fq \right) ^{-2h_{\tableau{1}}}=
1+\frac{8 \fq}{15}+\frac{92 \fq^2}{225}+\frac{3496 \fq^3}{10125}+
\frac{46322 \fq^{4}}{151875}+
\frac{3149896 \fq^{5}}{11390625}+\cdots \,  \, ,
\label{rZ2_2_3}
\end{align}
where $h_{\tableau{1}}=4/15$, and Conjecture \ref{conj:ex2} is confirmed.
\end{exam}

\begin{exam}[$\tableau{1} - \tableau{1} - (\emptyset) - \tableau{1} - \tableau{1 1}$ 
and $\tableau{1} - \tableau{1} - (\tableau{2 1}) - \tableau{1} - \tableau{1 1}$]
\label{ex6:inst}
For Conjecture \ref{conj:ex3}, consider, first, 
the Burge-reduced instanton partition function 
$\widehat{\mathcal{Z}}_{[2,0,0];(\ell_1,\ell_2)}^{(1,0),(2,0)}( \fq )$, 
where $\bs=[4,1]$, $\bs_1=\bs_2=\bs_3=[3,2]$ and $\bs_4=[2,3]$ are fixed. 
Then we find that 
the Burge-reduced instanton partition functions for 
$(\ell_1,\ell_2)=(0,0)$ and $(-1,-1)$ in the fundamental chamber are
\begin{align}
\begin{split}
\widehat{\mathcal{Z}}_{[2,0,0];(0,0)}^{(1,0),(2,0)}( \fq ) &=
\left(  1-\fq \right) ^{2h_{\tableau{1}}-\frac45} 
{_{2}F_{1}}\left( -\frac{1}{5},\frac{1}{5};\frac{2}{5};\fq\right) 
\\
&=
1+\frac{\fq}{6}+\frac{34 \fq^2}{315}+\frac{67 \fq^3}{810}+
\frac{49309 \fq^{4}}{722925}
+\frac{254267 \fq^{5}}{4337550}+\cdots \,  \, ,
\\
\widehat{\mathcal{Z}}_{[2,0,0];(-1,-1)}^{(1,0),(2,0)}( \fq ) &=
\frac{\fq^{\frac13}}{2}
\left(  1-\fq \right) ^{2h_{\tableau{1}}-\frac45} 
{_{2}F_{1}}\left( \frac{1}{5},\frac{4}{5};\frac{7}{5};\fq\right) 
\\
&=
\frac{\fq^{\frac13}}{2}+\frac{4 \fq^{\frac43}}{21}+\frac{79 \fq^{\frac73}}{630}
+\frac{4619 \fq^{\frac{10}{3}}}{48195}
+\frac{16237 \fq^{\frac{13}{3}}}{206550}+\cdots \,  \, ,
\label{rZ3_2_3}
\end{split}
\end{align}
where $h_{\tableau{1}}=4/15$, and the second one respects the fusion rules by 
\eqref{fusion_check}. Consider, next, 
the Burge-reduced instanton partition function 
$\widehat{\mathcal{Z}}_{[0,1,1];(\ell_1,\ell_2)}^{(1,0),(2,0)}( \fq )$, where 
$\bs=\bs_1=\bs_2=\bs_3=[3,2]$ and $\bs_4=[2,3]$ are fixed. 
Then we see that the Burge-reduced instanton partition functions 
for $(\ell_1,\ell_2)=(0,0)$ and $(1,1)$ in the fundamental chamber are
\begin{align}
\begin{split}
\widehat{\mathcal{Z}}_{[0,1,1];(0,0)}^{(1,0),(2,0)}( \fq ) &=
\left(  1-\fq \right) ^{2h_{\tableau{1}}-\frac45} 
{_{2}F_{1}}\left( -\frac{1}{5},\frac{2}{5};\frac{3}{5};\fq\right) 
\\
&=
1+\frac{2 \fq}{15}+\frac{13 \fq^2}{150}+\frac{8792 \fq^3}{131625}+
\frac{218507 \fq^{4}}{3948750}
+\frac{54190157 \fq^{5}}{1135265625}+\cdots \,  \, ,
\\
\widehat{\mathcal{Z}}_{[0,1,1];(1,1)}^{(1,0),(2,0)}( \fq ) &=
\frac{\fq^{\frac23}}{3}
\left(  1-\fq \right) ^{2h_{\tableau{1}}-\frac45} 
{_{2}F_{1}}\left( \frac{2}{5},\frac{4}{5};\frac{8}{5};\fq\right) 
\\
&=
\frac{\fq^{\frac23}}{3}+\frac{7 \fq^{\frac53}}{45}+
\frac{1867 \fq^{\frac83}}{17550}+\frac{32582 \fq^{\frac{11}{3}}}{394875}
+\frac{18575621 \fq^{\frac{14}{3}}}{272463750}+\cdots \,  \, ,
\label{rZ4_2_3}
\end{split}
\end{align}
where the second one respects the fusion rules by \eqref{fusion_check}.
The above results \eqref{rZ3_2_3} and \eqref{rZ4_2_3} support Conjecture \ref{conj:ex3}.
\end{exam}

\subsection{$(N,n)=(3,2)$ and $\widehat{\mathfrak{sl}}(2)_3$ WZW model}

  For  $(N,n)=(3,2)$, there are four highest-weight representations
\begin{align}
\emptyset=[3,0],\quad
\tableau{1}=[2,1],\quad
\tableau{2}=[1,2],\quad
\tableau{3}=[0,3],
\end{align}
with conformal dimensions 
\begin{align}
h_{[k_0,k_1]}=\frac{k_1 \, (k_1+2)}{20}:\quad
h_{\emptyset}=0,\quad
h_{\tableau{1}}=\frac{3}{20},\quad
h_{\tableau{2}}=\frac{2}{5},\quad
h_{\tableau{3}}=\frac{3}{4}.
\end{align}

\subsubsection{Burge-reduced generating functions of coloured Young diagrams}

The $\ft$-refined Burge-reduced generating functions \eqref{inst_red_ch} for $(N,n)=(3,2)$ are obtained as
\begin{align}
\begin{split}
\widehat{X}_{[3,0]}^{\mathrm{red}}(\fq, \ft)&=
\left( \fq;\fq \right)_{\infty}\, 
\sum_{\ell \in {\IZ}}
{X}_{(0,0,0);(-\ell)}^{[3,1,1]}(\fq)\, \ft^{2\ell} 
=
X^{[3,0]}_{[3,0]}(\fq)\, f_{0}(\fq, \hat{\ft} ) + 
X^{[3,0]}_{[1,2]}(\fq)\, g_{0}(\fq, \hat{\ft} ),
\\
\widehat{X}_{[0,3]}^{\mathrm{red}}(\fq, \ft)&=
\left( \fq;\fq \right)_{\infty}\, 
\sum_{\ell \in {\IZ}}
{X}_{(1,1,1);(-\ell)}^{[3,1,1]}(\fq)\, \ft^{2\ell+3} 
=
X^{[0,3]}_{[0,3]}(\fq)\, f_{1}(\fq, \hat{\ft} ) + 
X^{[0,3]}_{[2,1]}(\fq)\, g_{1}(\fq, \hat{\ft} ),
\\
\widehat{X}_{[2,1]}^{\mathrm{red}}(\fq, \ft)&=
\left( \fq;\fq \right)_{\infty}\, 
\sum_{\ell \in {\IZ}}
{X}_{(1,0,0);(-\ell)}^{[2,2,1]}(\fq)\, \ft^{2\ell+1} 
=
X^{[2,1]}_{[2,1]}(\fq)\, g_{1}(\fq, \hat{\ft} ) +  
X^{[2,1]}_{[0,3]}(\fq)\, f_{1}(\fq, \hat{\ft} ),
\\
\widehat{X}_{[1,2]}^{\mathrm{red}}(\fq, \ft)&=
\left( \fq;\fq \right)_{\infty}\, 
\sum_{\ell \in {\IZ}}
{X}_{(1,1,0);(-\ell)}^{[2,1,2]}(\fq)\, \ft^{2\ell+2} 
=
X^{[1,2]}_{[1,2]}(\fq)\, g_{0}(\fq, \hat{\ft} ) +  
X^{[1,2]}_{[3,0]}(\fq)\, f_{0}(\fq, \hat{\ft} ),
\label{N3n2_character}
\end{split}
\end{align}
where $\hat{\ft} =\fq^{-\frac14}\, \ft$,
\begin{align}
\begin{split}
X^{[3,0]}_{[3,0]}(\fq)&=
1 + \fq + 3 \fq^2 + 6 \fq^3 + 12 \fq^4 + 21 \fq^5 + 39 \fq^6 + 64 \fq^7 + 108 \fq^8 + \cdots \,  \, ,
\\
X^{[3,0]}_{[1,2]}(\fq)&=
\fq^{\frac12} + 2 \fq^{\frac32} + 5 \fq^{\frac52} + 9 \fq^{\frac72} + 18 \fq^{\frac92} + 31 \fq^{\frac{11}{2}} 
+ 55 \fq^{\frac{13}{2}} + 90 \fq^{\frac{15}{2}} + 149 \fq^{\frac{17}{2}} + \cdots \,  \, ,
\\
X^{[1,2]}_{[1,2]}(\fq)&=
1 + 2 \fq + 5 \fq^2 + 10 \fq^3 + 20 \fq^4 + 36 \fq^5 + 64 \fq^6 + 108 \fq^7 + 180 \fq^8 + \cdots \,  \, ,
\\
X^{[1,2]}_{[3,0]}(\fq)&=
\fq^{\frac12} + 3 \fq^{\frac32} + 6 \fq^{\frac52} + 13 \fq^{\frac72} + 24 \fq^{\frac92} + 
44 \fq^{\frac{11}{2}} 
+ 76 \fq^{\frac{13}{2}} + 129 \fq^{\frac{15}{2}} + 210 \fq^{\frac{17}{2}} + \cdots \,  \, ,
\\
X^{[0,3]}_{[0,3]}(\fq)&=X^{[3,0]}_{[3,0]}(\fq),\quad
X^{[0,3]}_{[2,1]}(\fq)=X^{[3,0]}_{[1,2]}(\fq),\quad
X^{[2,1]}_{[2,1]}(\fq)=X^{[1,2]}_{[1,2]}(\fq),\quad
X^{[2,1]}_{[0,3]}(\fq)=X^{[1,2]}_{[3,0]}(\fq),
\label{string_fn_N3n2}
\end{split}
\end{align}
and
\begin{align}
\begin{split}
f_{\sigma}(\fq, \hat{\ft} )=
\sum_{j \in 6\, {\IZ}+3\, \sigma}
\fq^{\frac{1}{12}\, j^2}\, \hat{\ft}^{\, j},
\quad
g_{\sigma}(\fq, \hat{\ft} )=
\sum_{j \in 6\, {\IZ}\pm \, (2 - \sigma)} 
\fq^{\frac{1}{12}\, j^2 + \frac16}\, \hat{\ft}^{\, j},
\quad \sigma=0,1.
\end{split}
\end{align}
The Burge-reduced generating functions \eqref{N3n2_character} agree with 
the $\widehat{\mathfrak{sl}}(2)_3$ WZW characters computed by \eqref{wzw_ch_formula},
\begin{align}
\begin{split}
&
\widehat{X}_{[3,0]}^{\mathrm{red}}(\fq, \ft)=
\chi_{[3,0]}^{\widehat{\mathfrak{sl}}(2)_3}( \fq, \hat{\ft} ) ,
\quad
\widehat{X}_{[0,3]}^{\mathrm{red}}(\fq, \ft)=
\chi_{[0,3]}^{\widehat{\mathfrak{sl}}(2)_3}( \fq, \hat{\ft} ) ,
\\
&
\widehat{X}_{[2,1]}^{\mathrm{red}}(\fq, \ft)=
\fq^{\frac{1}{10}}\,
\chi_{[2,1]}^{\widehat{\mathfrak{sl}}(2)_3}( \fq, \hat{\ft} ) ,
\quad
\widehat{X}_{[1,2]}^{\mathrm{red}}(\fq, \ft)=
\fq^{\frac{1}{10}}\,
\chi_{[1,2]}^{\widehat{\mathfrak{sl}}(2)_3}( \fq, \hat{\ft} ) ,
\end{split}
\end{align}
and Corollary \ref{prop:Ngch} is confirmed. 
Up to an overall factor, 
the functions \eqref{string_fn_N3n2} are 
the $\widehat{\mathfrak{sl}}(2)$ string functions of level-3 
in \cite{Kac:1984mq} and given by
(\textit{cf.} Corollary \ref{Cor:Chern=Char1}),
\begin{align}
\begin{split}
&
X^{[3,0]}_{[3,0]}(\fq) - \fq^{\frac16}\, X^{[3,0]}_{[1,2]}(\fq)=
\frac{\left( \fq^{\frac23};\fq^{\frac53} \right)_{\infty}
\left( \fq;\fq^{\frac53} \right)_{\infty}
\left( \fq^{\frac53};\fq^{\frac53} \right)_{\infty}}
{\left( \fq;\fq \right)_{\infty}^2},
\\
&
X^{[3,0]}_{[1,2]}(\fq)= \fq^{\frac12} \, 
\frac{\left( \fq^{3};\fq^{15} \right)_{\infty}
\left( \fq^{12};\fq^{15} \right)_{\infty}
\left( \fq^{15};\fq^{15} \right)_{\infty}}
{\left( \fq;\fq \right)_{\infty}^2},
\\
&
X^{[1,2]}_{[1,2]}(\fq)=  
\frac{\left( \fq^{6};\fq^{15} \right)_{\infty}
\left( \fq^{9};\fq^{15} \right)_{\infty}
\left( \fq^{15};\fq^{15} \right)_{\infty}}
{\left( \fq;\fq \right)_{\infty}^2},
\\
&
\fq^{\frac16} \, X^{[1,2]}_{[1,2]}(\fq) - X^{[1,2]}_{[3,0]}(\fq)= 
\fq^{\frac16}\,
\frac{\left( \fq^{\frac13};\fq^{\frac53} \right)_{\infty}
\left( \fq^{\frac43};\fq^{\frac53} \right)_{\infty}
\left( \fq^{\frac53};\fq^{\frac53} \right)_{\infty}}
{\left( \fq;\fq \right)_{\infty}^2}.
\end{split}
\end{align}
By taking $\ft=1$, 
the principal characters of $\widehat{\mathfrak{sl}}(2)$ are obtained as 
in \eqref{inst_pr_red_ch}:
\begin{align}
\begin{split}
&
\widehat{X}_{[3,0]}^{\mathrm{red}}(\fq, 1)=
\widehat{X}_{[0,3]}^{\mathrm{red}}(\fq, 1)=
\mathrm{Pr}\, \chi_{[3,0]}^{\widehat{\mathfrak{sl}}(2)}(\fq)=
\frac{\left( -\fq^{\frac12};\fq^{\frac12} \right)_{\infty}}
{\left( \fq ;\fq^{\frac52} \right)_{\infty}
\left( \fq^{\frac32};\fq^{\frac52} \right)_{\infty}},
\\
&
\widehat{X}_{[2,1]}^{\mathrm{red}}(\fq, 1)=
\widehat{X}_{[1,2]}^{\mathrm{red}}(\fq, 1)=
\mathrm{Pr}\, \chi_{[2,1]}^{\widehat{\mathfrak{sl}}(2)}(\fq)=
\frac{\left( -\fq^{\frac12};\fq^{\frac12} \right)_{\infty}}
{\left( \fq^{\frac12};\fq^{\frac52} \right)_{\infty}
\left( \fq^{2};\fq^{\frac52} \right)_{\infty}}.
\end{split}
\end{align}
Note that, these principal characters are related to the principal characters of 
$\widehat{\mathfrak{sl}}(3)$ in \eqref{pr_ch_sl32} by
\begin{align}
\frac{\mathrm{Pr}\, \chi_{[3,0]}^{\widehat{\mathfrak{sl}}(2)}(\fq^2)}
{\left( \fq^{2};\fq^{2} \right)_{\infty}}
=
\frac{\mathrm{Pr}\, \chi_{[2,0,0]}^{\widehat{\mathfrak{sl}}(3)}(\fq^3)}
{\left( \fq^{3};\fq^{3} \right)_{\infty}},
\qquad
\frac{\mathrm{Pr}\, \chi_{[2,1]}^{\widehat{\mathfrak{sl}}(2)}(\fq^2)}
{\left( \fq^{2};\fq^{2} \right)_{\infty}}
=
\frac{\mathrm{Pr}\, \chi_{[1,1,0]}^{\widehat{\mathfrak{sl}}(3)}(\fq^3)}
{\left( \fq^{3};\fq^{3} \right)_{\infty}}.
\end{align}

\subsubsection{Burge-reduced instanton partition functions}

For $N=3$ with general $n$, 
the Burge-reduced instanton partition functions \eqref{norm_red_inst_pf} are 
determined from the parameters in 
$\bs=[s_0, s_1, s_2]$, 
$\bs_1=[s_{1,0}, s_{1,1}, s_{1,2}]$, $\bs_2=[s_{2,0}, 1, s_{2,2}]$, 
$\bs_3=[s_{3,0}, s_{3,1}, 1]$ and $\bs_4=[s_{4,0}, s_{4,1}, s_{4,2}]$ 
in $P^{++}_{3,n+3}$ that are fixed by the relations 
\eqref{Eq:sigma2s}, \eqref{boundary_fix}: 
\begin{align}
s_I = \sigma_I-\sigma_{I+1}+1, \quad 
s_{1,I} = b_I - b_{I+1} + 1, \quad 
s_{4,I} = b_I^{\prime} - b_{I+1}^{\prime} + 1, \quad I=1,2,
\label{bc_N3_1}
\end{align}
and \eqref{boundary_fix_c} 
from the ordered charges $\sigma_{1}\ge \sigma_{2} \ge\sigma_{3}$, 
$b_1 \ge b_2 \ge b_3$, $b_1^{\prime} \ge b_2^{\prime} \ge b_3^{\prime}$. 
The Coulomb parameters are then determined from 
$\bs$ by \eqref{deg_momenta_v_triv}:
\begin{align}
\begin{split}
a_1&=\frac{1}{3}\sum_{I=1,2}\ll I-3\rr 
\ll s_I - 1 - \frac{n}{3}\rr \epsilon_2,
\\
a_2&=\frac{1}{3}\sum_{I=1,2}\ll 3-2I\rr 
\ll s_I - 1 - \frac{n}{3}\rr \epsilon_2,
\\
a_3&=\frac{1}{3}\sum_{I=1,2} I
\ll s_I - 1 - \frac{n}{3}\rr \epsilon_2,
\end{split}
\end{align}
and the mass parameters $\bm=(m_1,\ldots,m_N)$ and $\bm^{\prime}=(m_1^{\prime},\ldots,m_N^{\prime})$ 
are determined from the parameters in $\bs_1, \bs_2$ and 
$\bs_3, \bs_4$, respectively, by \eqref{deg_momenta_triv}.

We now consider the case of $(N,n)=(3,2)$ with 
the rational $\Omega$-background $\epsilon_1/\epsilon_2=-5/3$ 
in \eqref{omega_trivial}.%
\footnote{\,
Examples \ref{ex7:inst}, \ref{ex8:inst} and \ref{ex9:inst} are 
confirmed up to $O(\fq^{\frac{11}{2}})$.
}

\begin{exam}[$\emptyset - \emptyset - (\emptyset) - \emptyset - \emptyset$]
\label{ex7:inst}
Consider the Burge-reduced instanton partition function
$\widehat{\mathcal{Z}}_{[3,0];(\ell)}^{(0,0,0),(0,0,0)}( \fq )$ and 
take $\ell=0$ in the fundamental chamber, which respects the fusion rules, 
as in Conjecture \ref{conj:ex1}. 
Here $\bs=\bs_1=\bs_2=\bs_3=\bs_4=[3,1,1]$ are fixed. 
Then we see that the Burge-reduced instanton partition function is   
\begin{align}
\widehat{\mathcal{Z}}_{[3,0];(0)}^{(0,0,0),(0,0,0)}( \fq ) =
\left(  1-\fq \right) ^{-2h_{\emptyset}}=1,
\qquad
h_{\emptyset}=0,
\label{rZ1_3_2}
\end{align}
and Conjecture \ref{conj:ex1} is confirmed.
\end{exam}

\begin{exam}[$\emptyset - \tableau{1} - (\tableau{1}) - \tableau{1} - \emptyset$]
\label{ex8:inst}
Consider the Burge-reduced instanton partition function
$\widehat{\mathcal{Z}}_{[2,1];(\ell)}^{(0,0,0),(0,0,0)}( \fq )$ and take 
$\ell=0$ in the fundamental chamber as in Conjecture \ref{conj:ex2}, where 
$\bs=\bs_3=[2,2,1]$, $\bs_1=\bs_4=[3,1,1]$ and $\bs_2=[2,1,2]$ are fixed. 
Then the Burge-reduced instanton partition function is obtained as
\begin{align}
\widehat{\mathcal{Z}}_{[2,1];(0)}^{(0,0,0),(0,0,0)}( \fq ) =
\left(  1-\fq \right) ^{-2h_{\tableau{1}}}=
1+\frac{3 \fq}{10}+\frac{39 \fq^2}{200}+\frac{299 \fq^3}{2000}+
\frac{9867 \fq^4}{80000}+\frac{424281 \fq^{5}}{4000000}+\cdots \,  \, ,
\label{rZ2_3_2}
\end{align}
where $h_{\tableau{1}}=3/20$, and Conjecture \ref{conj:ex2} is confirmed.
\end{exam}

\begin{exam}[$\tableau{1} - \tableau{1} - (\emptyset) - \tableau{1} - \tableau{1}$ and $\tableau{1} - \tableau{1} - (\tableau{2}) - \tableau{1} - \tableau{1}$]
\label{ex9:inst}
For Conjecture \ref{conj:ex3}, consider, first, the Burge-reduced instanton partition function 
$\widehat{\mathcal{Z}}_{[3,0];(\ell)}^{(1,0,0),(1,0,0)}( \fq )$, 
where $\bs=[3,1,1]$, $\bs_1=\bs_3=\bs_4=[2,2,1]$ and $\bs_2=[2,1,2]$ are fixed. 
Then, we find that the Burge-reduced instanton partition functions for $\ell=0,-1$ 
in the fundamental chamber are
\begin{align}
\begin{split}
\widehat{\mathcal{Z}}_{[3,0];(0)}^{(1,0,0),(1,0,0)}( \fq ) &=
\left(  1-\fq \right) ^{2h_{\tableau{1}}-\frac35} 
{_{2}F_{1}}\left( -\frac{1}{5},\frac{2}{5};\frac{3}{5};\fq\right) 
\\
&=
1+\frac{\fq}{6}+\frac{13 \fq^2}{120}+\frac{87 \fq^3}{1040}+
\frac{8669 \fq^4}{124800}+\frac{344797 \fq^{5}}{5740800}+\cdots \,  \, ,
\\
\widehat{\mathcal{Z}}_{[3,0];(-1)}^{(1,0,0),(1,0,0)}( \fq ) &=
\frac{\fq^{\frac12}}{3}
\left(  1-\fq \right) ^{2h_{\tableau{1}}-\frac35} 
{_{2}F_{1}}\left( \frac{2}{5},\frac{4}{5};\frac{8}{5};\fq\right) 
\\
&=
\frac{\fq^{\frac12}}{3}+\frac{\fq^{\frac32}}{6}+\frac{61 \fq^{\frac52}}{520}
+\frac{289 \fq^{\frac72}}{3120}+\frac{222529 \fq^{\frac92}}{2870400}+
\frac{25723 \fq^{\frac{11}{2}}}{382720}+\cdots \,  \, ,
\label{rZ3_3_2}
\end{split}
\end{align}
where $h_{\tableau{1}}=3/20$, and the second one respects the fusion rules by 
\eqref{fusion_check}. Consider, next, the Burge-reduced instanton partition function 
$\widehat{\mathcal{Z}}_{[1,2];(\ell)}^{(1,0,0),(1,0,0)}( \fq )$, where 
$\bs=\bs_2=[2,1,2]$ and $\bs_1=\bs_3=\bs_4=[2,2,1]$ are fixed. 
Then we find that the Burge-reduced instanton partition functions for 
$\ell=0,1$ in the fundamental chamber are
\begin{align}
\begin{split}
\widehat{\mathcal{Z}}_{[1,2];(0)}^{(1,0,0),(1,0,0)}( \fq ) &=
\left(  1-\fq \right) ^{2h_{\tableau{1}}-\frac35} 
{_{2}F_{1}}\left( -\frac{1}{5},\frac{1}{5};\frac{2}{5};\fq\right) 
\\
&=
1+\frac{\fq}{5}+\frac{183 \fq^2}{1400}+\frac{353 \fq^3}{3500}+
\frac{796073 \fq^4}{9520000}+\frac{17182143 \fq^{5}}{238000000}+\cdots \,  \, ,
\\
\widehat{\mathcal{Z}}_{[1,2];(1)}^{(1,0,0),(1,0,0)}( \fq ) &=
\frac{\fq^{\frac12}}{2}
\left(  1-\fq \right) ^{2h_{\tableau{1}}-\frac35} 
{_{2}F_{1}}\left( \frac{1}{5},\frac{4}{5};\frac{7}{5};\fq\right) 
\\
&=
\frac{\fq^{\frac12}}{2}+\frac{29 \fq^{\frac32}}{140}+
\frac{393 \fq^{\frac52}}{2800}+\frac{51949 \fq^{\frac72}}{476000}+
\frac{1725293 \fq^{\frac92}}{19040000}+
\frac{74432711 \fq^{\frac{11}{2}}}{952000000}+\cdots \,  \, ,
\label{rZ4_3_2}
\end{split}
\end{align}
where the second one respects the fusion rules by \eqref{fusion_check}.
The above results \eqref{rZ3_3_2} and \eqref{rZ4_3_2} support Conjecture \ref{conj:ex3}.
\end{exam}

\section{Summary of results and remarks}
\label{sec:remarks}

\subsection{Summary of results}
The point of this paper is to compute conformal blocks in 
integral-level WZW models. 
Starting from the $SU(N)$ instanton partition functions 
on ${\IC}^2/{\IZ}_n$, with rational $\Omega$-deformation, 
based on the algebra $\mathcal{A}(N,n;p)$ in \eqref{alg_A_intro}, 
we proposed (in Conjectures 
\ref{conj:ex1}, \ref{conj:ex2} and \ref{conj:ex3}) a way 
to compute integral-level, integrable $\widehat{\mathfrak{sl}}(n)_N$ 
WZW conformal blocks, with rational central charges, where one
has to deal with the issue of null states.
By considering a rational $\Omega$-background 
$\frac{\epsilon_1}{\epsilon_2}=-1-\frac{n}{N}$ in \eqref{omega_trivial} 
and imposing appropriate Burge conditions in \eqref{sp_burge_c} to 
eliminate the null states, we trivialized the coset factor in 
the algebra $\mathcal{A}(N,n;N)$ as in \eqref{alg_A_triv_intro}, 
and were left with an integral-level WZW model. 
Further, we showed, in Corollary \ref{prop:Ngch}, that 
the Chern classes \eqref{fc_label} of the gauge bundle, which labels the 
instanton partition functions on the gauge side, can be interpreted as 
the eigenvalues of Chevalley elements in the Cartan subalgebra of 
$\slchap{\dimaff}$ on the CFT side. 

\subsection{The work of Alday and Tachikawa}
In \cite{Alday:2010vg}, Alday and Tachikawa, using results from 
\cite{Braverman:2004vv, Braverman:2004cr, FFNR:08, Negut:08}, as well as
AGT, found 
that $SU(2)$ instanton partition functions on $(z_1, z_2) \in {\IC}^2$ 
with generic $\Omega$-deformation, and in the presence of a \textit{full} surface 
operator at $z_2=0$, agree with 
$\widehat{\mathfrak{sl}}(2)$ conformal blocks that are modified by 
a $\mathcal{K}$-operator insertion, 
at generic-level $k=-2-\frac{\epsilon_2}{\epsilon_1}$. 
A generalization to the relation between 
$SU(N)$ instanton partition functions in the presence of a full surface 
operator and modified $\widehat{\mathfrak{sl}}(N)$ conformal blocks 
at generic-level
$$
k=-N-\frac{\epsilon_2}{\epsilon_1},
$$
was proposed in \cite{Kozcaz:2010yp}. 

In analogy with the moduli space of $U(N)$ instantons on ${\IC}^2/{\IZ}_n$ 
without surface operators described in Section \ref{sec:instantons}, 
to describe the moduli space of $U(N)$ instantons on ${\IC}^2$ in 
the presence of a full surface operator, 
one can use the moduli space of $U(N)$ instantons on 
${\IC} \times ({\IC}/{\IZ}_N)$ \cite{FFNR:08, FR:10, Kanno:2011fw}. 
Unlike the $\widehat{\mathfrak{sl}}(n)_N$ conformal blocks discussed 
in our work, these conformal blocks are at generic-level, and modified 
by the $\mathcal{K}$-operator insertion. 

\subsection{The work of Belavin and Mukhametzhanov}
In \cite{Belavin:2012aa}, Belavin and Mukhametzhanov obtained integrable 
WZW conformal blocks for $(N,n)=(2,2)$, (see footnote \ref{footnote:BM}). 
They found that, starting from the $SU(2)$ instanton partition functions 
on ${\IC}^2/{\IZ}_2$ 
with generic $\Omega$-deformation, the $\widehat{\mathfrak{sl}}(2)_2$ 
WZW conformal blocks in Examples \ref{ex1:inst}, \ref{ex2:inst} and 
\ref{ex3:inst} are obtained as prefactors of $\mathcal{N}=1$ 
super-Virasoro conformal blocks with generic central charge. 
In our work, with suitable rational choices of the parameters and by 
imposing Burge conditions, we trivialized the super-Virasoro conformal 
blocks (and their higher $(N, n)$ analogues), and computed conformal 
blocks for rational central charges, for more values of $(N, n)$. We 
conjecture that our approach works, for rational central charges, 
for all $(N, n)$, $N, n \in \IZ_{\geq 1}$.

\section*{Acknowledgements}
We would like to thank Piotr Su{\l}kowski for useful comments on 
the manuscript. OF wishes to thank Vladimir Belavin, Jean-Emile 
Bourgine and Raoul Santachiara for discussions on the subject of 
this work and related topics. We thank the Australian Research 
Council for support of this work.

\newpage

\appendix


\renewcommand{\dimfin}{{M}}    
\renewcommand{\dimaff}{{M}}    
\renewcommand{\levelaff}{{m}}    

\section{Lie algebras, affine Lie algebras and notation}
\label{app:cartan}

\textit{\noindent
Here we describe the notation we use from the structure and
representation theories of finite dimensional and
affine Lie algebras, as it pertains to
$\slfin{\dimfin}$ and $\slchap{\dimaff}$.
For a more comprehensive treatment, see \cite{kac.book.1990}.}

\subsection{The finite dimensional Lie algebra $\slfin{\dimfin}$}
\label{Sec:Notationfin}

Define the index set $\indfin=\{1,2,\ldots,\dimfin-1\}$.
The Lie algebra $\slfin{\dimfin}$ has Chevalley generators
$\{H_i,E_i,F_i\,|\,i\in\indfin\}$ with
$\{H_i\,|\,i\in\indfin\}$ a basis for its Cartan subalgebra $\ocartan$.
The Cartan matrix $\oA$ of $\slfin{\dimfin}$ is the
$(\dimfin-1)\times(\dimfin-1)$ matrix having entries
$\oA_{ij}=2\delta_{ij}-\delta_{i,j+1}-\delta_{i,j-1}$
for $i,j\in\indfin$.
The dual $\ocartan^\ast$ of $\ocartan$ has basis
$\{\alpha_j\,|\,j\in\indfin\}$
where the simple root $\alpha_j$ is defined by
$\alpha_j(H_i)=\oA_{ij}$ for $i,j\in\indfin$.
For $i\in\indfin$, the fundamental weight $\ofwt_i\in\ocartan^\ast$
is uniquely defined by $\ofwt_i(H_j)=\delta_{ij}$ for
$j\in\indfin$.
It follows that $\ofwt_i=\sum_{j\in\indfin} (\oA^{-1})_{ij}\,\alpha_j$,
where $\oA^{-1}$,
the inverse of $\oA$, has entries
\begin{equation}\label{Eq:InvCartan}
\Bigl(\oA^{-1}\Bigr)_{ij}
=\min\{i,j\}-\frac{ij}{\dimfin}
\end{equation}
for $i,j\in\indfin$.
Note that $\{\ofwt_1,\ofwt_2,\ldots,\ofwt_{\dimfin-1}\}$
is also a basis for $\ocartan^\ast$.
The Weyl vector $\orho$ is defined by $\orho=\sum_{i\in\indfin}\ofwt_i$.


It is convenient to embed $\ocartan^\ast$ in $\IC^{\dimfin}$,
by choosing an orthonormal basis $\{\be_1,\be_2,\ldots,\be_\dimfin\}$
for $\IC^{\dimfin}$,
and setting $\alpha_i=\be_i-\be_{i+1}$ for $i\in\indfin$.
The standard inner product $\inner{\cdot}{\cdot}$ on $\IC^{\dimfin}$
then leads to
\begin{equation}\label{Eq:FinInners}
\inner{\alpha_i}{\alpha_j}=\oA_{ij},\qquad
\inner{\alpha_i}{\ofwt_j}=\delta_{ij},\qquad
\inner{\ofwt_i}{\ofwt_j}=\Bigl(\oA^{-1}\Bigr)_{ij}
\end{equation}
for $i,j\in\indfin$.
Note that $\ocartan^\ast$ is the $(\dimfin-1)$-dimensional subspace of
$\IC^\dimfin$ that is perpendicular to $\be_1+\be_2+\cdots+\be_{\dimfin}$.
For convenience, we set
$\be_0=\frac{1}{\dimfin}(\be_1+\be_2+\cdots+\be_{\dimfin})$.
It is then easily confirmed that
\begin{equation}\label{Eq:FinFwts}
\ofwt_i=\sum_{k=1}^i\be_k-i \, \be_0
\qquad\text{and}\qquad
\orho=\frac12\sum_{k=1}^\dimfin(\dimfin-2k+1)\,\be_k\,.
\end{equation}

\subsection{The affine Lie algebra $\slchap{\dimaff}$}
\label{Sec:Notationaff}

Define the index set $\indaff=\{0,1,\ldots,\dimaff-1\}$.
The affine Lie algebra $\slchap{\dimaff}$ has Chevalley generators
$\{D,H_i,E_i,F_i\,|\,i\in\indaff\}$ with
$\{D,H_i\,|\,i\in\indaff\}$ a basis for its Cartan subalgebra $\cartan$
(which is $(\dimaff+1)$-dimensional).
The element $C=\sum_{i\in\indaff}H_i$ is central in $\slchap{\dimaff}$.
The Cartan matrix $A$ of $\slchap{\dimaff}$ is the $\dimaff\times\dimaff$
matrix having entries
$A_{ij}=2\delta^{(\dimaff)}_{ij}-\delta^{(\dimaff)}_{i,j+1}
                                -\delta^{(\dimaff)}_{i,j-1}$
for $i,j\in\indaff$,
where $\delta^{(\dimaff)}_{ij}=1$ if $i\equiv j\pmod\dimaff$ and
$\delta^{(\dimaff)}_{ij}=0$ otherwise.
The dual $\cartan^\ast$ of $\cartan$ has basis
$\{\fwt_0,\alpha_j\,|\,j\in\indaff\}$,
where the simple root $\alpha_j$ is defined by
$\alpha_j(H_i)=A_{ij}$ for $i,j\in\indaff$ and $\alpha_j(D)=\delta_{j0}$,
and the fundamental weight $\fwt_0$ is defined by
$\fwt_0(H_i)=\delta_{i0}$ for $i\in\indaff$ and $\fwt_0(D)=0$.
For $j\in\indaffm$ we also define $\fwt_j\in\cartan^\ast$ by
setting $\fwt_j(H_i)=\delta_{ij}$ for $i\in\indaff$,
and $\fwt_j(D)=0$.
The Weyl vector $\rho$ is defined by $\rho=\sum_{i\in\indaff}\fwt_i$.
The null root $\delta$ is defined by $\delta=\sum_{i\in\indaff}\alpha_i$,
and is such that $\delta(H_i)=0$ for $i\in\indaff$ and $\delta(D)=1$.
Note that $\{\fwt_0,\fwt_1,\ldots,\fwt_{\dimaff-1},\delta\}$ is also
a basis for $\cartan^\ast$.

Because $\slfin{\dimaff}$ appears canonically as a subalgebra
of $\slchap{\dimaff}$,
we may identify the $\alpha_i$ in the two cases for $i\in\indaffm$.
In addition, $\fwt_i=\ofwt_i+\fwt_0$ for $i\in\indaffm$.
Then $\rho=\orho+\dimaff\fwt_0$.
We again use the set of orthonormal vectors
$\be_1,\be_2,\ldots,\be_{\dimaff}$,
supplementing it with $\fwt_0$ and $\delta$ to give a basis for
$\cartan^\ast\cup\IC\be_0$.
In terms of these,
\begin{equation}
\alpha_0=\be_\dimaff-\be_1+\delta
\qquad\text{and}\qquad
\alpha_i=\be_i-\be_{i+1}
\end{equation}
for $i\in\indaffm$.
The inner product $\inner{\cdot}{\cdot}$ on $\cartan^\ast\cup\IC\be_0$
is defined by setting
\begin{equation}
\inner{\delta}{\fwt_0}=1,\qquad
\inner{\delta}{\delta}=\inner{\fwt_0}{\fwt_0}=
\inner{\delta}{\be_i}=\inner{\fwt_0}{\be_i}=0\,,
\end{equation}
in addition to $\inner{\be_i}{\be_j}=\delta_{ij}$,
for $i,j\in\indaffm$.
Then
\begin{equation}\label{Eq:AffInners}
\inner{\alpha_i}{\alpha_j}=A_{ij},\qquad
\inner{\alpha_i}{\fwt_j}=\delta_{ij},\qquad
\inner{\delta}{\fwt_j}=1,\qquad
\inner{\fwt_i}{\fwt_j}
=\min\{i,j\}-\frac{ij}{\dimaff}
\end{equation}
for $i,j\in\indaff$.
It follows that for each $\beta\in\cartan^\ast$,
we have $\beta(D)=\inner{\beta}{\fwt_0}$ and
$\beta(H_i)=\inner{\beta}{\alpha_i}$ for $i\in\indaff$,
as is easily checked.

The $\slchap{\dimaff}$ weight lattice $P_{\dimaff}$,
the level-$\levelaff$ weight lattice $P_{\dimaff,\levelaff}$,
the dominant weight lattice $P^{+}_{\dimaff}$,
the level-$\levelaff$ dominant weight lattice $P^{+}_{\dimaff,\levelaff}$,
the regular dominant weight lattice $P^{++}_{\dimaff}$,
and the level-$\levelaff$ regular dominant weight lattice
$P^{++}_{\dimaff,\levelaff}$,
are defined by%
\footnote{\,
Usually, the set $\IC\delta$ is adjoined to these sets.
The theory then proceeds with little change.}
\begin{align}
\begin{split}
P_{\dimaff}&=\bigoplus_{i\in\indaff}\IZ\fwt_i\,,
\qquad\quad
P_{\dimaff,\levelaff}=\{\Lambda\in P_{\dimaff}\,|\,
     \inner{\delta}{\Lambda}=\levelaff\}\,,
\\
P^{+}_{\dimaff}&=\bigoplus_{i\in\indaff}\IZ_{\ge0}\fwt_i\,,
\qquad
P^{+}_{\dimaff,\levelaff}=P^{+}_{\dimaff}\cap P_{\dimaff,\levelaff}\,,
\\
P^{++}_{\dimaff}&=\bigoplus_{i\in\indaff}\IZ_{>0}\fwt_i\,,
\qquad
P^{++}_{\dimaff,\levelaff}=P^{++}_{\dimaff}\cap P_{\dimaff,\levelaff}\,.
\end{split}
\end{align}
We will often use $[d_0,d_1,\ldots,d_{\dimaff-1}]$ to denote
the element $\sum_{i\in\indaff} d_i\fwt_i$ from any of these sets.
Note that if $[d_0,d_1,\ldots,d_{\dimaff-1}]\in P_{\dimaff,\levelaff}$
then $\sum_{i\in\indaff} d_i=\levelaff$.

\subsection{Affine weights and partitions}
\label{Sec:AffPars}

For $\Lambda=[d_0,d_1,\ldots,d_{\dimaff-1}]\in P^{+}_{\dimaff}$,
it is convenient to define a partition
$\lambda=(\lambda_1,\lambda_2,\ldots)$
by setting
\begin{equation}\label{Eq:Wt2Par}
\lambda_i=
\begin{cases}
\sum_{j=i}^{\dimaff-1} d_j&\text{if $1\le i<\dimaff$},\\
0&\text{if $i\ge\dimaff$}\,.
\end{cases}
\end{equation}
We will denote the partition so obtained by $\partit(\Lambda)$.
Note that if
$\Lambda\in P^{+}_{\dimaff,\levelaff}$
and $\lambda=\partit(\Lambda)$ then
$\lambda_1\le\levelaff$ and $\lambda_{\dimaff}=0$.
On the other hand, given a partition $\lambda$ with
$\lambda_1\le\levelaff$ and $\lambda_{\dimaff}=0$,
there is a unique $\Lambda\in P^{+}_{\dimaff,\levelaff}$
such that $\lambda=\partit(\Lambda)$.
We then write $\Lambda=\partit^{-1}(\lambda)$.
Note that $\partit^{-1}$ is well-defined only if $\levelaff$ is specified.

\begin{lemm}\label{Lem:Par2Wt}
Consider a partition $\bsig=(\sigma_1,\sigma_2,\ldots)$
for which $\sigma_1<\dimaff$ and $\sigma_{\levelaff+1}=0$,
and define $\Lambda\in P^{+}_{\dimaff,\levelaff}$ by
$\Lambda=\sum_{i=1}^{\levelaff}\fwt_{\sigma_i}$.
If $\lambda=\partit(\Lambda)$
then $\lambda=\bsig^T$, the partition conjugate to $\bsig$.
\end{lemm}
\begin{proof}
This follows after noting that in frequency notation $\bsig$ is expressed
$((\dimaff-1)^{d_{\dimaff-1}},\ldots,2^{d_2},1^{d_1})$.
\end{proof}

\subsection{Representations and characters of $\slchap{\dimaff}$}
\label{Sec:AffChars}

For $\Lambda\in P^{+}_{\dimaff}$,
let $L(\Lambda)$ denote the highest weight $\slchap{\dimaff}$-module
whose highest weight vector $v_\Lambda$ is such that
$H(v_\Lambda)=\Lambda(H)v_\Lambda$ for all $H\in\cartan$.
Then, with $\Lambda=[d_0,d_1,\ldots,d_{\dimaff-1}]$,
we have $\Lambda(H_i)=d_i\in\IZ_{\ge0}$ for $i\in\indaff$
and $\Lambda(D)=0$.
If $\Lambda\in P^{+}_{\dimaff,\levelaff}$
then the module $L(\Lambda)$ is said to be of level $\levelaff$.
The formal character of $L(\Lambda)$ is defined to be
\begin{equation}\label{Eq:CharAffForm}
\text{ch}\,L(\Lambda)
=\sum_{\beta\in\cartan^\ast} (\dim V_\beta)\, \e^\beta,
\end{equation}
where $V_\beta$ is the subspace of $L(\Lambda)$ for which,
for each $v\in V_\beta$, we have
$H(v)=\beta(H)v$ for all $H\in\cartan$.
Instead of using the formal exponentials $\e^\beta$,
it is convenient to set $\e^{-\delta}=\qvar$ and $\e^{-\be_i}=x_i$
for $1\le i\le\dimaff$, and define
\begin{equation}\label{def_wzw_app}
\xchi^{\, \slchap{\dimaff}}_\Lambda(\qvar,\bx)=
\text{ch} \, L(\Lambda)
\Bigl\vert_{\{\e^{-\delta}=\qvar,\e^{-\be_i}=x_i\,|\,1\le i\le\dimaff\}}\,.
\end{equation}

For $\Lambda\in P^{+}_{\dimaff}$, the Weyl-Kac character formula
(\cite[eqn.~(10.4.5)]{kac.book.1990}) yields
\begin{equation}\label{Eq:CharAff}
\xchi^{\, \slchap{\dimaff}}_\Lambda(\qvar,\bx)
=
\e^\Lambda\,
\frac{{\mathcal N}^{\, \slchap{\dimaff}}_\Lambda(\qvar,\bx)}
     {{\mathcal N}^{\, \slchap{\dimaff}}_0(\qvar,\bx)}
\end{equation}
with ${\mathcal N}^{\, \slchap{\dimaff}}_\Lambda(\qvar,\bx)$ given by
\begin{equation}\label{Eq:CharAffNum}
{\mathcal N}^{\, \slchap{\dimaff}}_\Lambda(\qvar,\bx)=
\sum_{\begin{subarray}{c}
       k_1,\ldots,k_\dimaff\in\IZ\\[0.4ex]
       k_1+ \cdots+k_\dimaff=0
       \end{subarray}}
\det_{1\le i,j\le \dimaff}
\left( x_{i}^{-(\dimaff+\levelaff)k_{i}-\lambda_j+j+\lambda_{i}-i}
\qvar^{\, (\lambda_j-j)k_{i} +\frac12 (\dimaff+\levelaff) k_i^2 } \right),
\end{equation}
where $\lambda=\partit(\Lambda)$.
Using the $\slchap{\dimaff}$ Macdonald identity
\cite[eqn.~(10.4.4)]{kac.book.1990},
the denominator of \eqref{Eq:CharAff} may be alternatively expressed:
\begin{equation}\label{Eq:CharAffDen}
{\mathcal N}^{\, \slchap{\dimaff}}_0(\qvar,\bx)
=\pochinf{\qvar}{\qvar}^{\dimaff-1}
\prod_{1\le i<j\le \dimaff}
\vpochinf{\frac{x_i}{x_j}}{\qvar}
\vpochinf{\frac{x_j}{x_i}\, \fq}{\qvar}\,.
\end{equation}
(For a more detailed derivation of \eqref{Eq:CharAffNum}
and \eqref{Eq:CharAffDen}, see \cite[Appendix B.2]{Foda:2015bsa}).

In the case in which $\Lambda$ is of level one,
so that $\Lambda=\fwt_k$ for $k\in\indaff$,
the character $\text{ch}\,L(\Lambda)$ has the explicit expression
\cite[eqn.~(12.13.6)]{kac.book.1990}
\begin{equation}
\text{ch}\,L(\Lambda_k)=
\frac{\e^{\fwt_k}}{\pochinf{\qvar}{\qvar}^{\dimaff-1}}
\sum_{\eta\in \bigoplus_{i=1}^{\dimaff-1} {\IZ}\, \alpha_i}
\e^{\eta+\frac12\inner{\fwt_k}{\fwt_k}\delta
-\frac12\inner{\eta+\ofwt_k}{\eta+\ofwt_k}\delta}\,.
\end{equation}
Then, use of the inner products in Appendix \ref{Sec:Notationaff},
and noting that
$\inner{\eta+\ofwt_k}{\eta+\ofwt_k}=
\inner{\eta+\fwt_k}{\eta+\fwt_k}=
\inner{\eta}{\eta}+
2\inner{\eta}{\fwt_k}+\inner{\fwt_k}{\fwt_k}$,
leads to
\begin{equation}\label{Eq:CharAffLevel1}
\xchi^{\, \slchap{\dimaff}}_{\Lambda_k}(\qvar,\bx)=
\frac{
\e^{\fwt_k}}{\pochinf{\qvar}{\qvar}^{\dimaff-1}}
\sum_{\bell\in\IZ^{\dimaff-1}}
\qvar^{\, \ell_k+\sum_{i=1}^{\dimaff-1}(\ell_i^2-\ell_i\ell_{i-1})}
\prod_{i=1}^{\dimaff-1}
\left(\frac{x_{i+1}}{x_i}\right)^{\ell_i},
\end{equation}
where $\bell=(\ell_1,\ldots,\ell_{\dimaff-1})$ with $\ell_0=0$.

It will be useful to note that
$\xchi^{\, \slchap{\dimaff}}_{\Lambda}(\qvar,\bx)$
is invariant on multiplying each of the $x_i$ by the same non-zero constant.
This is obvious in \eqref{Eq:CharAffLevel1}, and consideration of
the determinant in \eqref{Eq:CharAffNum} shows that it holds also
for both the numerator and denominator of \eqref{Eq:CharAff}.

\subsection{Principally specialised characters of $\slchap{\dimaff}$}
\label{Sec:Principal}

Although the denominator
${\mathcal N}^{\, \slchap{\dimaff}}_0(\qvar,\bx)$
of \eqref{Eq:CharAff} can be written in product form,
this is not the case with the numerator
${\mathcal N}^{\, \slchap{\dimaff}}_\Lambda(\qvar,\bx)$.
However, after substituting $\e^{-\alpha_i}\to \qvar^{1/\dimaff}$
for each simple root $\alpha_i$,
${\mathcal N}^{\, \slchap{\dimaff}}_\Lambda(\qvar,\bx)$
can be written in product form.
The same is then true for
$\xchi^{\, \slchap{\dimaff}}_\Lambda(\qvar,\bx)$
\cite[Proposition 10.9]{kac.book.1990}.
This specialisation is effected by substituting
$x_i\to \qvar^{-i/\dimaff}$ into \eqref{Eq:CharAffNum}.
So define the principally specialised character
\begin{equation}\label{Eq:PrincipalDef}
\mathrm{Pr}\,{\chi}^{\slchap{\dimaff}}_\Lambda(\qvar)
=\left.
e^{-\Lambda}\, \xchi^{\, \slchap{\dimaff}}_\Lambda(\qvar,\bx)
\right|_{\{x_i\to \qvar^{-i/\dimaff},1\le i\le\dimaff\}}\,.
\end{equation}
The result can be conveniently expressed by, for
$\Lambda\in P^{+}_{\dimaff,\levelaff}$,
setting $\lambda=\partit(\Lambda)$ and
defining the set
$\Omega(\Lambda)=\{\levelaff+j-\lambda_j\,|\,j=1,\ldots,\dimaff\}$.
Then
\begin{equation} \label{Eq:Principal}
\mathrm{Pr}\,{\chi}^{\slchap{\dimaff}}_\Lambda(\qvar)=
\frac
{\vpochinf{\qvar}{\qvar}}
{\vpochinf{\qvar^{1+\frac{\levelaff}{\dimaff}}}{\qvar^{1+\frac{\levelaff}{\dimaff}}}}
\prod_{\begin{subarray}{c}
        1\le i<j\le \dimaff+\levelaff\\[.5ex]
        i\notin\Omega(\Lambda),j\in\Omega(\Lambda)
       \end{subarray}}
\frac{1}
{\vpochinf{\qvar^{\frac{j-i}{\dimaff}}}
          {\qvar^{1+\frac{\levelaff}{\dimaff}}}}
\prod_{\begin{subarray}{c}
        1\le i<j\le \dimaff+\levelaff\\[.4ex]
        i\in\Omega(\Lambda),j\notin\Omega(\Lambda)
       \end{subarray}}
\frac{1}
{\vpochinf{\qvar^{1+\frac{\levelaff-j+i}{\dimaff}}}
          {\qvar^{1+\frac{\levelaff}{\dimaff}}}}\,.
\end{equation}
(See \cite{Foda:2015bsa} for more details: the substitution
$q\to \qvar^{1/\dimaff}$ there gives the normalisation used here.)

\subsection{Characters of $\slchap{\dimaff}$ via crystal graphs}
\label{Sec:Crystal}

Here we describe the crystal graph enumeration of characters
of $\slchap{\dimaff}$
that was developed by the Kyoto group
\cite{DJKMO:1989c,DJKMO:1989,JMMO:1990}.
The formulation that we use is similar to that in
\cite[Section 2]{FLOTW:1999}.

For $\Lambda=[d_0,d_1,\ldots,d_{\dimaff-1}]\in P^+_{\dimaff,\levelaff}$,
let $\bsig=(\sigma_1,\sigma_2,\ldots)$ be the partition
for which $\sigma_{\levelaff+1}=0$ and
$\Lambda=\sum_{i=1}^{\levelaff}\fwt_{\sigma_i}$.%
\footnote{\, 
\label{Fn:Partit}By Lemma \ref{Lem:Par2Wt},
if $\lambda=\partit(\Lambda)$ then $\bsig=\lambda^T$.}
Let $\multicol{\bsig}$ be the set of $\levelaff$-tuples of
coloured Young diagrams $\bY=(Y_{1},Y_{2},\ldots,Y_{\levelaff})$
whose row lengths $Y_{\ell,i}$ are constrained by
\begin{equation}\label{Eq:Cylindrical}
\begin{split}
&Y_{\ell,i}\ge Y_{\ell+1,i+\sigma_\ell-\sigma_{\ell+1}}\quad
\text{for }i\ge1, \ 1\le\ell<\levelaff;\\
&Y_{\levelaff,i}\ge Y_{1,i+\sigma_\levelaff-\sigma_{1}+\dimaff}\quad
\text{for }i\ge1,
\end{split}
\end{equation}
and where the box $(i,j)$ of $Y_\ell$ is coloured
$(\sigma_\ell+i-j)\bmod\dimaff$.%
\footnote{\,
Note that this colouring convention differs from that
defined in Section \ref{sec:characterisation}.
It also differs from that used in \cite{FLOTW:1999}, but
is appropriate to the conventions used here and in \cite{Foda:2015bsa}.}
The elements of $\multicol{\bsig}$ are called cylindrical multipartitions
in \cite{FLOTW:1999}.
For $\bY\in\multicol{\bsig}$ and $i\in\indaff$, define $k_i(\bY)$
to be the number of boxes in $\bY$ that are coloured $i$,
and then set $\delta k_{i}(\bY)=k_{i}(\bY)-k_0(\bY)$.

The crystal graph theory
shows that the character $\text{ch}\,L(\Lambda)$ of the
irreducible $\slchap{\dimaff}$ representation with highest weight
$\Lambda$ can be expressed as a sum over a subset 
$\multicolHL{\bsig}$ of $\multicol{\bsig}$,
the elements of $\multicolHL{\bsig}$ being
known as highest-lift multipartitions.
We will not define $\multicolHL{\bsig}$ here
(the definition can be found in \cite[Proposition 2.11]{FLOTW:1999}),
because we will just use and state its pivotal property.
Using, as before, $\e^{-\delta}=\qvar$ and $\e^{-\be_i}=x_i$
for $i\in\indaffm$, this expression is \cite[Theorem 1.2]{DJKMO:1989}:
\begin{equation}\label{Eq:CrystalCharAst}
\begin{split}
\xchi^{\, \slchap{\dimaff}}_\Lambda(\qvar,\bx)
&=
\e^\Lambda\,\sum_{\bY\in\multicolHL{\bsig}}
\qvar^{\, k_0(\bY)}
\prod_{i=0}^{\dimaff-1} \left(\frac{x_i}{x_{i+1}}\right)^{k_i(\bY)}\\
&=
\e^\Lambda\,\sum_{\bY\in\multicolHL{\bsig}}
\qvar^{\, k_0(\bY)}
\prod_{i=1}^{\dimaff-1}
\left(\frac{x_i}{x_{i+1}}\right)^{\delta k_i(\bY)}\,,
\end{split}
\end{equation}
the second equality following on using
$k_i(\bY)=\delta k_i(\bY)+k_0(\bY)$ for $i\in\indaff$,
and noting that, in particular, $\delta k_0(\bY)=0$.

The property of $\multicolHL{\bsig}$ that is needed is that there
is a natural bijection%
\footnote{\,
Obtained using an abacus with $\dimaff$ rungs --- see 
\cite{Foda:2015bsa}.}
\begin{equation}
\multicol{\bsig}\to\multicolHL{\bsig}\times\partitset,
\end{equation}
where $\partitset$ is the set of all partitions, such that
if $\bY\mapsto(\bY^*,\lambda)$ then
$\delta k_i(\bY)=\delta k_i(\bY^\ast)$
for $i\in\indaff$, and $k_0(\bY)=k_0(\bY^\ast)+|\lambda|$.
Because the generating function for $\partitset$ is
\begin{equation}
\sum_{\lambda\in\partitset} \qvar^{\, |\lambda|}=\frac{1}{\pochinf{\qvar}{\qvar}},
\end{equation}
the expression \eqref{Eq:CrystalCharAst} yields
\begin{equation}\label{Eq:CrystalChar}
\begin{split}
\xchi^{\, \slchap{\dimaff}}_\Lambda(\qvar,\bx)
=
\e^\Lambda\, \pochinf{\qvar}{\qvar}
\sum_{\bY\in\multicol{\bsig}}
\qvar^{\, k_0(\bY)}
\prod_{i=1}^{\dimaff-1}
\left(\frac{x_i}{x_{i+1}}\right)^{\delta k_i(\bY)}\,.
\end{split}
\end{equation}
Now for a vector
$\bell=(\ell_1,\ell_2,\ldots,\ell_{\dimaff-1})
 \in\IZ^{\dimaff-1}$,
define $\multicolref{\bsig}{\bell}\subset\multicol{\bsig}$
to be the set of all $\bY\in\multicol{\bsig}$
for which $\delta k_i(\bY)=\ell_i$ for each $i\in\indfin$.
Also set $\ell_0=\ell_{\dimaff}=0$ for convenience.
We can then write \eqref{Eq:CrystalChar} in the form
\begin{equation}\label{Eq:CrystalChar2}
\begin{split}
\xchi^{\, \slchap{\dimaff}}_\Lambda(\qvar,\bx)
&=
\e^\Lambda\, 
\pochinf{\qvar}{\qvar}
\sum_{\bell\in\IZ^{\dimaff-1}}
\sum_{\bY\in\multicolref{\bsig}{\bell}}
\qvar^{\, k_0(\bY)}
\prod_{i=1}^{\dimaff-1}
\left(\frac{x_i}{x_{i+1}}\right)^{\ell_i}
\\
&=
\e^\Lambda
\sum_{\bell\in\IZ^{\dimaff-1}}
\xa^\Lambda_{\bell}(\qvar)
\prod_{i=1}^{\dimaff}
x_i^{\, \ell_i-\ell_{i-1}},
\end{split}
\end{equation}
where the $\slchap{\dimaff}$ (normalised) string function
$\xa^\Lambda_{\bell}(\qvar)$ is given by
\begin{equation}\label{Eq:SLstring}
\xa^\Lambda_{\bell}(\qvar)
=
\pochinf{\qvar}{\qvar}
\sum_{\bY\in\multicolref{\bsig}{\bell}}
\qvar^{\, k_0(\bY)}\,.
\end{equation}
%
Tabulations of the coefficients of the string functions
$\xa^\Lambda_{\bell}(\qvar)$ in the cases $2\le\dimaff\le 9$ for weights
$\Lambda\in P^{+}_{\dimaff,\levelaff}$ of various
small levels $m$ can be found in \cite{KMPS:1990}.%
\footnote{\,
This string function $\xa^\Lambda_{\bell}(\qvar)$
is usually denoted as $a^\Lambda_\gamma(\fq)$
where $\gamma=\Lambda-\sum_{i=1}^{\dimaff-1}\ell_i\alpha_i$
(see Appendix \ref{Sec:WZWChars}, especially (\ref{Eq:StringEq})).}


\subsection{WZW Characters}
\label{Sec:WZWChars}

For any affine Lie algebra $\mathfrak g$,
the Sugawara construction demonstrates a homomorphism from
the Virasoro algebra \emph{Vir} to $U_c(\mathfrak g)$,
a completion of the universal enveloping algebra of
$\mathfrak g$, for any level $\levelaff\ge0$
(see \cite[\S12.8]{kac.book.1990} for details).
Consequently, for $\Lambda\in P^{+}_{\dimaff,\levelaff}$,
the $\slchap{\dimaff}$-module $L(\Lambda)$ also
serves as a \emph{Vir}-module.
The central charge $c$ and
conformal dimension $h_\Lambda$ of this \emph{Vir}-module are given
in \eqref{Eq:CC-CD}.

Through the homomorphism $\emph{Vir}\to U_c(\slchap{\dimaff})$,
the Virasoro generator $L_0$ acts on $L(\Lambda)$ by
$L_0\mapsto h_\Lambda\text{Id}-D$,
where $\text{Id}$ is the identity operator
\cite[Corollary 12.8]{kac.book.1990}.
Consequently, the definition \eqref{Eq:WZWCharDef}
yields
\begin{equation}\label{Eq:WZWChar}
\chi_{\Lambda}^{\slchap{\dimaff}_\levelaff}(\qvar,\hat{\bft})
=
\fq^{\, h_\Lambda}\,
\mathrm{Tr}_{L(\Lambda)}\,
\qvar^{-D}
\prod_{i=1}^{\dimaff-1} \hat{\ft}_{i}^{\, H_i }\,.
\end{equation} 
Now note that \eqref{Eq:CrystalChar2} can be written as a
sum of terms $\exp(\beta)$ with $\beta$ of the form
\begin{equation}\label{Eq:BetaForm}
\beta
=\Lambda-k\delta-\sum_{j=1}^\dimaff \e_j(\ell_j-\ell_{j-1})
=\Lambda-k\delta-\sum_{j=1}^\dimaff \ell_j\alpha_j
\end{equation}
for some $k\in\IZ$.
Then, because $\beta(D)=-k$ and
$\beta(H_i)=d_i+\ell_{i-1}-2\ell_i+\ell_{i+1}$,
\eqref{Eq:WZWChar} yields:
\begin{equation}\label{Eq:SLCharDynkin}
\qchi^{\slchap{\dimaff}_{\levelaff}}_\Lambda(\qvar,\hat{\bft})
=
\fq^{\, h_{\Lambda}}\,
\sum_{\bell\in\IZ^{\dimaff-1}}
\xa^\Lambda_{\bell}(\qvar)
\prod_{i=1}^{\dimaff-1}
\hat{\ft}_i^{\, d_i+\ell_{i-1}-2\ell_i+\ell_{i+1}}\,.
\end{equation}
Alternatively, this may be expressed as
\begin{equation}\label{Eq:SLCharDynkin2}
\qchi^{\slchap{\dimaff}_{\levelaff}}_\Lambda(\qvar,\hat{\bft})
=
\fq^{\, h_{\Lambda}}\,
\sum_{\bell\in\IZ^{\dimaff-1}}
a^\Lambda_{\gamma(\bell)}(\qvar)
\prod_{i=1}^{\dimaff-1}
\hat{\ft}_i^{\, \gamma(\bell)_i},
\end{equation}
after defining
$\gamma(\bell)=[\gamma_0,\gamma_1,\ldots,\gamma_{\dimaff-1}]
\in P_{\dimaff,\levelaff}$
by setting
\begin{equation}\label{Eq:ell2dinkin}
\gamma_i=
d_i+\ell_{i-1}-2\ell_i+\ell_{i+1}
=d_i-\sum_{j=1}^{\dimaff-1} A_{ij} \ell_j
\end{equation}
for each $i\in\indaff$, and defining
\begin{equation}\label{Eq:StringEq}
a^\Lambda_{\gamma(\bell)}(\fq)=\xa^\Lambda_{\bell}(\fq)\,.
\end{equation}
By using \eqref{Eq:SLstring}, we can also express \eqref{Eq:SLCharDynkin}
as
\begin{equation}\label{Eq:SLCharDynkin3}
\qchi^{\slchap{\dimaff}_{\levelaff}}_\Lambda(\qvar,\hat{\bft})
=
\fq^{\, h_{\Lambda}}\,
\pochinf{\qvar}{\qvar}
\sum_{\bY\in\multicol{\bsig}}
\qvar^{\, k_0(\bY)}
\prod_{i=1}^{\dimaff-1}
\hat{\ft}_i^{\, d_i+\delta k_{i-1}(\bY)-2\delta k_i(\bY)+\delta k_{i+1}(\bY)},
\end{equation}
where we set $\delta k_0(\bY)=\delta k_\dimaff(\bY)=0$.

In the level one case where $\Lambda=\Lambda_k$,
comparing $\eqref{Eq:CharAffLevel1}$ with \eqref{Eq:CrystalChar2}
shows that
\begin{equation}\label{Eq:CharAffLevel1xa}
\xa^{\Lambda_k}_{\bell}(\fq)=
\frac{1}
{\pochinf{\qvar}{\qvar}^{\dimaff-1}}\,
\qvar^{-\ell_k+\sum_{i=1}^{\dimaff-1}(\ell_i^2-\ell_i\ell_{i-1})},
\end{equation}
where $\ell_0=0$.
Then \eqref{Eq:SLCharDynkin} gives
\begin{equation}\label{Eq:CharAffLevel1t}
\qchi^{\slchap{\dimaff}_{1}}_{\Lambda_k}(\qvar,\hat{\bft})
=
\frac{\fq^{\, h_{\Lambda_k}}}{\pochinf{\qvar}{\qvar}^{\dimaff-1}}
\sum_{\bell\in\IZ^{\dimaff-1}}
\qvar^{-\ell_k+\sum_{i=1}^{\dimaff-1}(\ell_i^2-\ell_i\ell_{i-1})}
\prod_{i=1}^{\dimaff-1}
\hat{\ft}_i^{\, \delta_{ik}+\ell_{i-1}-2\ell_i+\ell_{i+1}},
\end{equation}
where $\ell_0=\ell_{\dimaff}=0$.

\subsection{Converting between the $\bfx$ and $\hat{\bft}$ variables}
\label{Sec:Convert}

In Appendix \ref{Sec:AffChars},
in expressing the character of $L(\Lambda)$,
the formal exponentials $e^\beta$ were exchanged
for $\e^{-\delta}=\qvar$ and $\e^{-\be_i}=x_i$ for $1\le i\le\dimaff$.
However, in Appendix \ref{Sec:WZWChars}, the same character
was expressed using
$\qvar$ along with $\hat{\ft}_1,\ldots,\hat{\ft}_{\dimaff-1}$.
Here, we convert between these variables, and give a form
of the Weyl-Kac formula that expresses
$\qchi^{\slchap{\dimaff}_\levelaff}_\Lambda(\qvar,\hat{\bft})$
in the $\hat{\bft}$ variables.

In terms of Dynkin components, \eqref{Eq:BetaForm} takes the form
$\beta=-k\delta+\sum_{i\in\indaff}
       (d_i+\ell_{i-1}-2\ell_i+\ell_{i+1})\Lambda_i.$
Therefore (see \eqref{Eq:SLCharDynkin}),
$\qchi^{\slchap{\dimaff}_{\levelaff}}_\Lambda(\qvar,\hat{\bft})$
is obtained from $\fq^{h_\Lambda} L(\Lambda)$ by
exchanging $\e^{-\delta}=\qvar$, $e^{\fwt_0}=1$ and
$\e^{\fwt_i}=\hat{\ft}_i$ for $1\le i<\dimaff$.

Because $\alpha_i=\be_i-\be_{i+1}$ and
$\alpha_i=-\fwt_{i-1}+2\fwt_i-\fwt_{i+1}$,
the $\bx$ and $\hat{\bft}$ variables are related by
\begin{equation}
\frac{x_i}{x_{i+1}}=\frac{\hat{\ft}_{i-1}\hat{\ft}_{i+1}}{\hat{\ft}_i^2}
\quad
\Longleftrightarrow
\quad
x_i=\frac{\hat{\ft}_{i-1}}{\hat{\ft}_{i}}
\frac{\hat{\ft}_{M}}{\hat{\ft}_{M-1}} \, x_M
\end{equation}
for $1\le i<\dimaff$, where we set $\hat{\ft}_0=\hat{\ft}_{\dimaff}=1$.
In view of the last paragraph of Appendix \ref{Sec:AffChars},
we then obtain 
$\fq^{-h_\Lambda}\qchi^{\slchap{\dimaff}_{\levelaff}}_\Lambda(\qvar,\hat{\bft})$
by substituting $x_i\to\hat{\ft}_{i-1}/\hat{\ft}_i$ into \eqref{Eq:CharAff}.
For $\Lambda=[d_0,d_1,\ldots,d_{\dimaff-1}]\in P^{+}_{\dimaff,\levelaff}$
this gives, after also making use of \eqref{Eq:CharAffNum}
and \eqref{Eq:CharAffDen},
\begin{equation}\label{wzw_ch_formula}
\qchi^{\slchap{\dimaff}_\levelaff}_\Lambda(\qvar,\hat{\bft})
=
\fq^{\, h_\Lambda}\,
\frac{{\mathcal N}_\Lambda(\qvar,\hat{\bft})}
     {
\pochinf{\qvar}{\qvar}^{\dimaff-1}
\prod_{1\le i<j\le \dimaff}
\vpochinf{\hat{\ft}_{i-1}\hat{\ft}_j/\hat{\ft}_{i}\hat{\ft}_{j-1}}{\qvar}
\vpochinf{\qvar\,\hat{\ft}_{i}\hat{\ft}_{j-1}/\hat{\ft}_{i-1}\hat{\ft}_{j}}{\qvar}
}
\,\prod_{i=1}^{\dimaff-1}\hat{\ft}_i^{\, d_i},
\end{equation}
where, with $\lambda=\partit(\Lambda)$,
\begin{equation}\label{Eq:CharAffNum2t}
{\mathcal N}_\Lambda(\qvar,\hat{\bft})=
\sum_{\begin{subarray}{c}
       k_1,\ldots,k_\dimaff\in\IZ\\[0.4ex]
       k_1+ \cdots+k_\dimaff=0
       \end{subarray}}
\det_{1\le i,j\le \dimaff}
\left( \bigl(\hat{\ft}_i/\hat{\ft}_{i-1}\bigr)^{(\dimaff+\levelaff)k_{i}
                                    +\lambda_j-j-\lambda_{i}+i} \,
\qvar^{\, (\lambda_j-j)k_{i} +\frac12 (\dimaff+\levelaff) k_i^2 } \right)\,.
\end{equation}

\section{Some AGT correspondences}\label{app:agt_check}

\textit{\noindent
Following \cite{Alday:2009aq, Wyllard:2009hg, Mironov:2009by, Belavin:2011tb}, 
we summarize some explicit AGT correspondences to identify 
our conventions 
in Section \ref{subsec:4pt_inst} and to confirm the $U(1)$ 
factor $Z_{\mathcal{H}}\left( \bm, \bm^{\prime}; \fq \right)$ 
in \eqref{u1_factor}.}

\subsection{$(N,n)=(2,1)$ and Virasoro conformal blocks}

For $(N,n)=(2,1)$, the $SU(2)$ instanton partition function \eqref{inst_vec_pf} 
with $a_1=-a_2=a$ is computed as
\begin{align}
\begin{split}
Z_{\boldsymbol{0};\emptyset}^{\boldsymbol{0},\boldsymbol{0}}
(  a, \bm, \bm^{\prime}; \fq ) &=
1+\fq \left\lgroup \frac{(a-m_1)\, (a-m_2)\, (a+m_1^{\prime}-\epsilon_1-\epsilon_2)  
\, (a+m_2^{\prime}-\epsilon_1-\epsilon_2) }
{2\, a\, \epsilon_1\, \epsilon_2 \, (-2\, a+\epsilon_1+\epsilon_2) }
\right.
\\
&\qquad\left.
-\frac{(a+m_1)\, (a+m_2) \, (a-m_1^{\prime}+\epsilon_1+\epsilon_2)  
\, (a-m_2^{\prime}+\epsilon_1+\epsilon_2) }
{2\, a\, \epsilon_1\, \epsilon_2 \, (2\, a+\epsilon_1+\epsilon_2) }
\right\rgroup 
+O\ll \fq^2\rr.
\end{split}
\end{align}
In \cite{Alday:2009aq} 
(see \cite{Gaiotto:2009ma, Marshakov:2009gn} for non-conformal/Whittaker limits) 
it was found that, by subtracting the $U(1)$ factor \eqref{u1_factor} for 
$(N,n)=(2,1)$, the normalized instanton partition function
\begin{align}
\widehat{Z}_{\boldsymbol{0};\emptyset}^{\boldsymbol{0},\boldsymbol{0}}
( a, \bm, \bm^{\prime}; \fq ) :=
\left( 1-\fq \right) ^{-\frac{\ll \sum_{I=1}^2 m_I\rr 
\ll \epsilon_1+\epsilon_2 - \frac{1}{2}\sum_{I=1}^2 m_I^{\prime}\rr }{\epsilon_1\, \epsilon_2}}\,
Z_{\boldsymbol{0};\emptyset}^{\boldsymbol{0},\boldsymbol{0}}
(  a, \bm, \bm^{\prime}; \fq ) 
\label{nek_N2n1_norm}
\end{align}
gives the 
$c=1+6\,\frac{(\epsilon_1+\epsilon_2)^2}{\epsilon_1\, \epsilon_2}$ 
Virasoro conformal blocks of 4-point function \eqref{para_WN_block} on ${\IP}^1$ 
by the parameter identifications \eqref{mass_rel_base} and \eqref{coulomb_rel}:
\begin{align}
\begin{split}
\mu^v=\frac{\epsilon_1+\epsilon_2}{2} + a,
\quad
\mu_1&=\frac{\epsilon_1+\epsilon_2}{2} + \frac{m_1-m_2}{2},
\quad
\mu_2=\frac{m_1+m_2}{2},
\\
\mu_4&=\frac{\epsilon_1+\epsilon_2}{2} - \frac{m_1^{\prime}-m_2^{\prime}}{2},
\quad
\mu_3=\frac{m_1^{\prime}+m_2^{\prime}}{2}.
\label{agt_para_N2n1}
\end{split}
\end{align}
Here, note that, in the $n=1$ cases, 
the WZW factor $\widehat{\mathfrak{sl}}(n)_N$ in 
the algebra \eqref{alg_A} is absent. 
For example, the Virasoro conformal block at level 1,
\begin{align}
\frac{
\ll \Delta_{\mu^v} - \Delta_{\mu_1} + \Delta_{\mu_2}\rr  
\ll \Delta_{\mu^v} + \Delta_{\mu_3} - \Delta_{\mu_4}\rr  
}{
2\, \Delta_{\mu^v}},
\qquad
\Delta_{\mu} =
\frac{
\mu \ll \epsilon_1+\epsilon_2-\mu \rr  
}{
\epsilon_1\, \epsilon _2},
\end{align} 
agrees with the coefficient of $\fq$ in \eqref{nek_N2n1_norm}.

\subsection{$(N,n)=(3,1)$ and $\mathcal{W}_3$ conformal blocks}

For $(N,n)=(3,1)$, the $SU(3)$ instanton partition function \eqref{inst_vec_pf} 
with $a_3=-a_1-a_2$ is computed as
\begin{tiny}
\begin{align}
\begin{split}
&
Z_{\boldsymbol{0};\emptyset}^{\boldsymbol{0},\boldsymbol{0}}
( \ba, \bm, \bm^{\prime}; \fq ) 
\\
&= 1
+\fq \left\lgroup \frac{(a_1-m_1)\, (a_1-m_2)\, (a_1-m_3) 
\, ( -a_1-m_1^{\prime}+\epsilon_1+\epsilon_2)  
\, ( -a_1-m_2^{\prime}+\epsilon_1+\epsilon_2)  
\, ( -a_1-m_3^{\prime}+\epsilon_1+\epsilon_2) }
{\epsilon_1\, \epsilon_2\, (a_1-a_2)\, (2\, a_1+a_2) 
\, ( -2\, a_1-a_2+\epsilon_1+\epsilon_2)  
\, ( -a_1+a_2+\epsilon_1+\epsilon_2) }
\right.
\\
&\quad
+
\frac{(a_2-m_1)\, (a_2-m_2)\, (a_2-m_3) 
\, (-a_2-m_1^{\prime}+\epsilon_1+\epsilon_2)  
\, (-a_2-m_2^{\prime}+\epsilon_1+\epsilon_2)  
\, (-a_2-m_3^{\prime}+\epsilon_1+\epsilon_2) }
{\epsilon_1\, \epsilon_2\, (a_2-a_1)\, (a_1+2\, a_2) 
\, (-a_1-2\, a_2+\epsilon_1+\epsilon_2)  
\, (a_1-a_2+\epsilon_1+\epsilon_2) }
\\
&\quad
\left.
-\frac{(a_1+a_2+m_1)\, (a_1+a_2+m_2)\, (a_1+a_2+m_3) 
\, (a_1+a_2-m_1^{\prime}+\epsilon_1+\epsilon_2)  
\, (a_1+a_2-m_2^{\prime}+\epsilon_1+\epsilon_2)  
\, (a_1+a_2-m_3^{\prime}+\epsilon_1+\epsilon_2) }
{\epsilon_1\, \epsilon_2\, (2\, a_1+a_2)\, (a_1+2\, a_2) 
\, (2\, a_1+a_2+\epsilon_1+\epsilon_2)  
\, (a_1+2\, a_2+\epsilon_1+\epsilon_2) }\right\rgroup 
\\
&\quad
+O\ll \fq^2\rr.
\end{split}
\end{align}
\end{tiny}

By subtracting the $U(1)$ factor \eqref{u1_factor} for $(N,n)=(3,1)$, 
one finds that the normalized instanton partition function
\begin{align}
\widehat{Z}_{\boldsymbol{0};\emptyset}^{\boldsymbol{0},\boldsymbol{0}}
( \ba, \bm, \bm^{\prime}; \fq ) :=
\left( 1-\fq \right) ^{-\frac{\ll \sum_{I=1}^3 m_I\rr 
\ll \epsilon_1+\epsilon_2 - \frac{1}{3}\sum_{I=1}^3 m_I^{\prime}\rr }{\epsilon_1\, \epsilon_2}}\,
Z_{\boldsymbol{0};\emptyset}^{\boldsymbol{0},\boldsymbol{0}}
( \ba, \bm, \bm^{\prime}; \fq )
\label{nek_N3n1_norm}
\end{align}
gives the $\mathcal{W}_3$ conformal blocks of 4-point 
function \eqref{para_WN_block} on ${\IP}^1$, with 
$c=2+24\,\frac{(\epsilon_1+\epsilon_2)^2}{\epsilon_1\, \epsilon_2}$, 
by the parameter identifications \eqref{mass_rel_base} and \eqref{coulomb_rel} 
\cite{Wyllard:2009hg, Mironov:2009by} 
(see \cite{Taki:2009zd, Keller:2011ek, Kanno:2012xt} for non-conformal/Whittaker limits):
\begin{align}
\begin{split}
\mu^v_1&=
\frac{\epsilon_1+\epsilon_2}{2} + \frac{a_1-a_2}{2},
\quad
\mu^v_2=
\frac{\epsilon_1+\epsilon_2}{2} + \frac{a_1+2\, a_2}{2},
\\
\mu_{1,1}&=
\frac{\epsilon_1+\epsilon_2}{2} + \frac{m_1-m_2}{2},
\quad
\mu_{1,2}=
\frac{\epsilon_1+\epsilon_2}{2} + \frac{m_2-m_3}{2},
\quad
\mu_{2}=\frac{m_1+m_2+m_3}{2},
\\
\mu_{4,1}&=
\frac{\epsilon_1+\epsilon_2}{2} - \frac{m_1^{\prime}-m_2^{\prime}}{2},
\quad
\mu_{4,2}=
\frac{\epsilon_1+\epsilon_2}{2} - \frac{m_2^{\prime}-m_3^{\prime}}{2},
\quad
\mu_{3}=\frac{m_1^{\prime}+m_2^{\prime}+m_3^{\prime}}{2}.
\end{split}
\end{align}
For example, the $\mathcal{W}_3$ conformal block at level 1,

\begin{align}
\begin{split}
&
\frac{
\ll  \Delta_{\bmu^v} - \Delta_{\bmu_1} + \Delta_{0,\mu_2}\rr   
\ll  \Delta_{\bmu^v} - \Delta_{\bmu_4} + \Delta_{\mu_3,0}\rr 
}{
2\, \Delta_{\bmu^v}}
\\
&+
\ll -\frac{w_{\bmu^v}}{2} - w_{\bmu_1} + \frac{w_{0,\mu_2}}{2} 
+ \frac{3\left( \Delta_{\bmu^v}-\Delta_{\bmu_1}\right) w_{0,\mu_2}}
{2\, \Delta_{0,\mu_2}}
- \frac{3\left( \Delta_{0,\mu_2}-\Delta_{\bmu_1}\right) w_{\bmu^v}}
{2\, \Delta_{\bmu^v}}\rr 
\\
&\times
\ll -\frac{w_{\bmu^v}}{2} - w_{\bmu_4} + \frac{w_{\mu_3,0}}{2} 
+ \frac{3\left( \Delta_{\bmu^v}-\Delta_{\bmu_4}\right) w_{\mu_3,0}}
{2\, \Delta_{\mu_3,0}}
- \frac{3\left( \Delta_{\mu_3,0}-\Delta_{\bmu_4}\right) w_{\bmu^v}}
{2\, \Delta_{\bmu^v}}\rr 
\\
&\times
\ll \Delta_{\bmu^v}
\ll \frac{4\, \epsilon_1\, \epsilon_2\, \Delta_{\bmu^v}}
{4\, \epsilon_1\, \epsilon_2 + 15 \left( \epsilon_1+\epsilon_2\right) ^2}
-\frac{3 \left( \epsilon_1+\epsilon_2\right) ^2}
{4\, \epsilon_1\, \epsilon_2 + 15 \left( \epsilon_1+\epsilon_2\right)^2}\rr 
-\frac{9\, w_{\bmu^v}^2}{2\, \Delta_{\bmu^v}}\rr ^{-1},
\end{split}
\end{align}
agrees with the coefficient of $\fq$ in \eqref{nek_N3n1_norm}, 
where
\begin{footnotesize}
\begin{align}
\begin{split}
\Delta_{\bmu}&=\Delta_{\mu_1,\mu_2}=
-\frac{2 \ll 2\, \mu_1^2 + 2\, \mu_1\, \mu_2 + 2\, \mu_2^2
-3\, (\epsilon_1+\epsilon_2)\, (\mu_1+\mu_2) \rr }
{3\, \epsilon_1\, \epsilon_2},
\\
w_{\bmu}&=w_{\mu_1,\mu_2}=
\frac{1}{\epsilon_1\, \epsilon_2}
\ll \frac23\, (2\, \mu_1+\mu_2) -\, (\epsilon_1+\epsilon_2) \rr 
\ll \frac23\, (\mu_1+2\, \mu_2) -\, (\epsilon_1+\epsilon_2) \rr 
\ll \frac23\, (\mu_1-\mu_2) \rr 
\\
&\qquad\times
\ll \frac{-6}{4\, \epsilon_1\, \epsilon_2 + 15\, (\epsilon_1+\epsilon_2)^2}\rr ^{\frac12}.
\end{split}
\end{align}
\end{footnotesize}

\subsection{$(N,n)=(2,2)$ and $\mathcal{N}=1$ super-Virasoro conformal blocks}

For $(N,n)=(2,2)$, the $SU(2)$ instanton partition functions \eqref{inst_vec_pf} 
with $a_1=-a_2=a$ are computed as \textit{e.g.},
\begin{tiny}
\begin{align}
\begin{split}
&
Z_{(0,0);(0)}^{(0,0),(0,0)}
(  a, \bm, \bm^{\prime}; \fq ) 
\\
&=
1 + \fq \left\lgroup \frac{(a-m_1)\, (a-m_2)\, (a+m_1^{\prime}-\epsilon_1-\epsilon_2)  
\, (a+m_2^{\prime}-\epsilon_1-\epsilon_2) }
{4\, a \, \epsilon_2\, (\epsilon_1-\epsilon_2)  
\, (-2\, a+\epsilon_1+\epsilon_2) }
+\frac{(a-m_1)\, (a-m_2) \, (a+m_1^{\prime}-\epsilon_1-\epsilon_2)  
\, (a+m_2^{\prime}-\epsilon_1-\epsilon_2) }
{4\, a\, \epsilon_1 \, (\epsilon_2-\epsilon_1)  
\, (-2\, a+\epsilon_1+\epsilon_2) }\right.
\\
&\left.
+\frac{(a+m_1)\, (a+m_2) \, (a-m_1^{\prime}+\epsilon_1+\epsilon_2)  
\, (a-m_2^{\prime}+\epsilon_1+\epsilon_2) }
{4\, a\, \epsilon_1 \, (\epsilon_1-\epsilon_2)  
\, (2\, a+\epsilon_1+\epsilon_2) }
-\frac{(a+m_1)\, (a+m_2) \, (a-m_1^{\prime}+\epsilon_1+\epsilon_2)  
\, (a-m_2^{\prime}+\epsilon_1+\epsilon_2) }
{4\, a \,  \epsilon_2 \, (\epsilon_1-\epsilon_2) 
\, (2\, a+\epsilon_1+\epsilon_2) }\right\rgroup 
+O\ll \fq^2\rr ,
\\
&
Z_{(1,1);(1)}^{(0,0),(0,0)}
(  a, \bm, \bm^{\prime}; \fq ) 
=
\fq^{\frac12} \ll \frac{1}{2\, a \, (-2\, a+\epsilon_1+\epsilon_2) }
-\frac{1}{2\, a \, (2\, a+\epsilon_1+\epsilon_2) }\rr 
+O\ll \fq^{\frac32}\rr, 
\label{N2n2_inst1}
\end{split}
\end{align}
\end{tiny}\noindent
for the vanishing first Chern class $c_1=0$ in \eqref{first_chern}.

We consider the subtraction of the $U(1)$ factor \eqref{u1_factor} 
for $(N,n)=(2,2)$ from the instanton partition functions
\begin{align}
\widehat{Z}_{\bsig;(\ell)}^{\bb,\bb^{\prime}}
(  a, \bm, \bm^{\prime}; \fq ) :=
\left( 1-\fq \right) ^{-\frac{\ll \sum_{I=1}^2 m_I\rr 
\ll \epsilon_1+\epsilon_2 - \frac{1}{2}\sum_{I=1}^2 m_I^{\prime}\rr }{2\,\epsilon_1\, \epsilon_2}}\,
Z_{\bsig;(\ell)}^{\bb,\bb^{\prime}}
(  a, \bm, \bm^{\prime}; \fq ).
\label{nek_N2n2_norm}
\end{align}
In \cite{Belavin:2011tb, Belavin:2012aa} (see also \cite{Bonelli:2011kv}), 
it was shown that the normalized instanton partition functions \eqref{nek_N2n2_norm} give 
the $\mathcal{N}=1$ super-Virasoro conformal blocks of 4-point 
function \eqref{para_WN_block} on ${\IP}^1$, with 
$c=\frac{3}{2}+3\,\frac{( \epsilon_1+\epsilon_2)^2}{\epsilon_1\, \epsilon_2}$, by 
the parameter identifications \eqref{agt_para_N2n1} 
(see \cite{Belavin:2011pp, Bonelli:2011jx, Ito:2011mw} for non-conformal/Whittaker limits). 
For example, the instanton partition functions \eqref{N2n2_inst1} 
correspond to the conformal blocks of four NS primary fields, and 
actually the conformal block at level 1,
\begin{align}
\frac{
\ll \Delta_{\mu^v} - \Delta_{\mu_1} + \Delta_{\mu_2} \rr 
\ll \Delta_{\mu^v} + \Delta_{\mu_3} - \Delta_{\mu_4} \rr 
}{
2\, \Delta_{\mu^v}},
\qquad
\Delta_{\mu} = 
\frac{
\mu \ll \epsilon_1+\epsilon_2-\mu \rr 
}{
2\,\epsilon_1\, \epsilon _2},
\end{align} 
agrees with the coefficient of $\fq$ in 
$\widehat{Z}_{(0,0);(0)}^{(0,0),(0,0)} (a, \bm, \bm^{\prime}; \fq )$, and 
the conformal blocks
\begin{align}
\begin{split}
\textrm{at level $\frac12$}:&\qquad
\frac{1}{2\, \Delta_{\mu^v}},
\\
\textrm{at level $\frac32$}:&\qquad
\frac{
\ll 1+2\, \Delta_{\mu^v} - 2\, \Delta_{\mu_1} + 2\, \Delta_{\mu_2}\rr 
\ll 1+2\, \Delta_{\mu^v} + 2\, \Delta_{\mu_3} - 2\, \Delta_{\mu_4}\rr }
{8\, \Delta_{\mu^v} \ll 1+2\, \Delta_{\mu^v} \rr }
\\
&\qquad
+ \frac{
6 
\ll \Delta_{\mu_2} - \Delta_{\mu_1}\rr  
\ll \Delta_{\mu_3} - \Delta_{\mu_4}\rr 
}{
\ll  c - \ll  9-2\, c\rr \Delta_{\mu^v} + 6\, \Delta_{\mu^v}^2 \rr 
\ll  1 + 2\, \Delta_{\mu^v} \rr },
\end{split}
\end{align}
agree with the coefficients of $\fq^{\frac12}$ and $\fq^{\frac32}$ in 
$2\epsilon_1\, \epsilon_2\, \widehat{Z}_{(1,1);(1)}^{(0,0),(0,0)} (a, \bm, \bm^{\prime}; \fq )$ \cite{Belavin:2011tb}.

\section{Integrable $\widehat{\mathfrak{sl}}(n)_N$ WZW 4-point conformal blocks 
for fundamental representations}
\label{app:wzw_4pt}

The integrable $\widehat{\mathfrak{sl}}(n)_N$ WZW conformal blocks of 
4-point function on ${\IP}^1$ 
of primary fields with (anti-)fundamental representations $\tableau{1}$, 
$\tableau{1}$, $\overline{\tableau{1}}$, and $\overline{\tableau{1}}$, 
schematically denoted by
\begin{align}
\left<\overline{\tableau{1}}(\infty)\,  \tableau{1}(1)\,  \tableau{1}(z)\,  \overline{\tableau{1}}(0) \right>_{\IP^1}^{\widehat{\mathfrak{sl}}(n)_N},
\end{align}
were obtained in 
\cite{Knizhnik:1984nr} (see also \cite{DiFrancesco:1997nk}), 
as solutions to the Knizhnik-Zamolodchikov equation, as
\begin{align}
\begin{split}
\mathcal{F}_1^{(0)}(z)&=
z^{-2h_{\tableau{1}}} \left( 1-z\right) ^{h_{\theta}-2h_{\tableau{1}}}
{_{2}F_{1}}\left( -\frac{1}{n+N},\frac{1}{n+N};\frac{N}{n+N};z\right) ,
\\
\mathcal{F}_2^{(0)}(z)&=\frac{1}{N}\,
z^{1-2h_{\tableau{1}}} \left( 1-z\right) ^{h_{\theta}-2h_{\tableau{1}}}
{_{2}F_{1}}\left( 1-\frac{1}{n+N},1+\frac{1}{n+N};1+\frac{N}{n+N};z\right) ,
\\
\mathcal{F}_1^{(1)}(z)&=
z^{h_{\theta}-2h_{\tableau{1}}} \left( 1-z\right) ^{h_{\theta}-2h_{\tableau{1}}}
{_{2}F_{1}}\left( \frac{n-1}{n+N},\frac{n+1}{n+N};1+\frac{n}{n+N};z\right) ,
\\
\mathcal{F}_2^{(1)}(z)&=-n\,
z^{h_{\theta}-2h_{\tableau{1}}} \left( 1-z\right) ^{h_{\theta}-2h_{\tableau{1}}}
{_{2}F_{1}}\left( \frac{n-1}{n+N},\frac{n+1}{n+N};\frac{n}{n+N};z\right) ,
\label{wzw_bc_kz}
\end{split}
\end{align}
where $h_{\tableau{1}}=\frac{n^2-1}{2n(n+N)}$ is the conformal dimension of 
the four primary fields, and $h_{\theta}=\frac{n}{n+N}$ is 
the conformal dimension of the adjoint field with weight $\theta=[N-1,1,0,\ldots,0,1]$. 
These four solutions correspond to two choices of the representations of 
states in the internal channel which follow from the fusion of 
$\tableau{1}$ and $\overline{\tableau{1}}$, 
and $\mathcal{F}_1^{(0)}(z), \mathcal{F}_2^{(0)}(z)$ 
(resp. $\mathcal{F}_1^{(1)}(z), \mathcal{F}_2^{(1)}(z))$ corresponds to 
the identity (resp. adjoint) field conformal block of ``$s$-channel''. 
Under a hypergeometric transformation
\begin{align}
z \quad \to \quad
\fq:=\frac{z}{z-1},
\end{align}
the Gauss hypergeometric function transforms as
\begin{align}
{_{2}F_{1}}\left( \alpha,\beta;\gamma;z\right) =
\left( 1-\fq \right) ^{\alpha}
{_{2}F_{1}}\left( \alpha,\gamma-\beta;\gamma;\fq\right) ,
\end{align}
and the $\widehat{\mathfrak{sl}}(n)_N$ WZW 4-point conformal blocks \eqref{wzw_bc_kz} 
are expressed, in the $\fq$-module, as
\begin{align}
\begin{split}
\widehat{\mathcal{F}}_1^{(0)}(\fq)&:=
z^{2h_{\tableau{1}}} \mathcal{F}_1^{(0)}(z)=
\left( 1-\fq \right) ^{2h_{\tableau{1}}-\frac{n+1}{n+N}}
{_{2}F_{1}}\left( -\frac{1}{n+N},\frac{N-1}{n+N};\frac{N}{n+N};\fq\right) ,
\\
\widehat{\mathcal{F}}_2^{(0)}(\fq)&:=
z^{2h_{\tableau{1}}} \mathcal{F}_2^{(0)}(z)=
-\frac{\fq}{N}
\left( 1-\fq \right) ^{2h_{\tableau{1}}-\frac{n+1}{n+N}}
{_{2}F_{1}}\left( \frac{N-1}{n+N},1-\frac{1}{n+N};1+\frac{N}{n+N};\fq\right) ,
\\
\widehat{\mathcal{F}}_1^{(1)}(\fq)&:=
\frac{z^{2h_{\tableau{1}}}}{n} \mathcal{F}_1^{(1)}(z)=
\frac{\left( -\fq \right) ^{h_{\theta}}}{n} 
\left( 1-\fq \right) ^{2h_{\tableau{1}}-\frac{n+1}{n+N}}
{_{2}F_{1}}\left( \frac{n-1}{n+N},1-\frac{1}{n+N};1+\frac{n}{n+N};\fq\right) ,
\\
\widehat{\mathcal{F}}_2^{(1)}(\fq)&:=
\frac{z^{2h_{\tableau{1}}}}{n} \mathcal{F}_2^{(1)}(z)=
-\left( -\fq \right) ^{h_{\theta}} 
\left( 1-\fq \right) ^{2h_{\tableau{1}}-\frac{n+1}{n+N}}
{_{2}F_{1}}\left( -\frac{1}{n+N},\frac{n-1}{n+N};\frac{n}{n+N};\fq\right).
\label{wzw_bc_kz_trans}
\end{split}
\end{align}


\end{document}